\documentclass[pra,twocolumn,nofootinbib]{revtex4-1}

\usepackage{amssymb}
\usepackage{amsthm}
\usepackage{amsmath, multirow,color}
\usepackage{graphicx}

\usepackage{afterpage}

\newtheorem{theorem}{Theorem}
\newtheorem{definition}[theorem]{Definition}
\newtheorem{lemma}[theorem]{Lemma}
\newtheorem{corollary}[theorem]{Corollary}
\newtheorem{proposition}[theorem]{Proposition}

\newcommand{\M}{\mathcal{M}}
\renewcommand{\P}{\mathcal{P}}
\renewcommand{\S}{\mathcal{S}}
\newcommand{\an}[1]{\left\langle{#1}\right\rangle}
\newcommand{\unit}{\hat{\mathbf{1}}}
\newcommand{\ket}[1]{\left|{#1}\right\rangle}
\newcommand{\bra}[1]{\left\langle{#1}\right|}
\newcommand{\half}{\frac{1}{2}}
\newcommand{\nn}{\nonumber}
\newcommand{\s}[1]{\hat \sigma_{#1}}
\newcommand{\proj}[1]{\ket{#1}\!\bra{#1}}
\newcommand{\Alice}{Abydos }
\newcommand{\Aliceian}{Abydosian }

\definecolor{nblue}{rgb}{0.2,0.2,0.7}
\definecolor{ngreen}{rgb}{0.2,0.6,0.2}
\definecolor{nred}{rgb}{0.9,0,0}
\definecolor{nblack}{rgb}{0,0,0}

\newcommand{\red}{\color{nred}}

\newcommand{\sBox}[1]{
{\begin{center}\fbox{\begin{minipage}{.95\columnwidth}\small \red #1\end{minipage}}\end{center}}}

\begin{document}

\title{Specker's Parable of the Over-protective Seer: A Road to Contextuality, Nonlocality and Complementarity {\red PLUS AN ERRATUM}}

\author{Yeong-Cherng Liang }
\email{ycliang@mail.ncku.edu.tw
(Former affiliations: Group of Applied Physics, University of Geneva, CH-1211 Geneva 4, Switzerland \& School of Physics, The University of Sydney, Sydney, New South Wales 2006, Australia)}
\affiliation{Department of Physics, National Cheng Kung University, Tainan 701, Taiwan}

\author{Robert W. Spekkens}
\email{rspekkens@perimeterinstitute.ca (Corresponding author)}
\affiliation{Perimeter Institute for Theoretical Physics, 31 Caroline St. N, Waterloo,
Ontario, Canada N2L 2Y5}
\author{Howard M. Wiseman}
\email{h.wiseman@griffith.edu.au}
\affiliation{ARC Centre for Quantum Computation and Communication Technology, and Centre for Quantum Dynamics, Griffith University, Brisbane, Queensland 4111, Australia}
\date{\today}

\begin{abstract}
In 1960, the mathematician Ernst Specker described a simple example of nonclassical correlations, the counterintuitive features of which he dramatized using a parable about a seer who sets an impossible prediction task to his daughter's suitors.  We revisit this example here, using it as an entr\'ee to three central concepts in quantum foundations: contextuality,
Bell-nonlocality, and complementarity.  Specifically, we show that Specker's parable offers a narrative thread
that weaves together a large number of results, including: the impossibility of
measurement-noncontextual and outcome-deterministic ontological
models of quantum theory (the 1967 Kochen-Specker theorem), in particular the recent
state-specific pentagram proof of Klyachko;
the impossibility  of Bell-local models of quantum
theory (Bell's theorem), especially the proofs by Mermin and Hardy and extensions thereof;
the impossibility of a preparation-noncontextual ontological model of
quantum theory; and the existence
of triples of positive operator valued measures (POVMs) that can be
measured jointly pairwise but not triplewise.  Along the way,
several novel results are presented, including: a generalization of
a theorem by Fine connecting the existence of a
joint distribution over outcomes of counterfactual measurements to the existence
of a measurement-noncontextual and outcome-deterministic ontological model; a generalization of
Klyachko's proof of  the Kochen-Specker theorem  from pentagrams to a family of star polygons;
 a proof of the Kochen-Specker theorem in the style of Hardy's proof of Bell's theorem
 (i.e., one that makes use of the failure of the transitivity of implication for
counterfactual statements);  a categorization of contextual and Bell-nonlocal correlations in terms of frustrated networks; a derivation of a new inequality
testing preparation noncontextuality; and lastly, some novel results on the joint measurability of POVMs and the question of whether these can be modeled noncontextually. Finally, we emphasize that
Specker's parable of the over-protective seer provides a novel type of foil to quantum theory, challenging us to explain why the particular sort of contextuality and complementarity embodied therein does not arise in a quantum world.
\end{abstract}

\maketitle
\tableofcontents

\section{Introduction}

In the field of quantum foundations, the mathematician Ernst Specker is rightly
famous for introducing, with co-author Simon Kochen, the concept of a
\emph{noncontextual hidden variable model} and proving that such a model
cannot underly quantum theory.
This 1967 result, known as the Kochen-Specker theorem
\cite{KS67}, continues to be an active subject of research
today (see Ref.~\cite{Cabello:0012089} for a bibliography). One
finds precursors to this result in the 1957 work of Gleason
\cite{Gle57} and Bell's 1966 review article on hidden variable
models (which refers to Gleason's result) \cite{Bell66}, but
also in a 1960 paper by Specker entitled ``The logic of
propositions that are not simultaneously decidable''
\cite{Spe60}\footnote{The 1967 Kochen-Specker
theorem~\cite{KS67} improves upon many of these earlier results
by making use of a finite set of observables.   It should
be noted, however, that Bell's 1964 proof~\cite{Bell64} of
quantum nonlocality is also a proof of contextuality using only
a finite set of observables; unlike the Kochen-Specker proof,
it is {\em state-specific}, the first example of this kind.}.
This article studied logical features of quantum theory, in
particular the question of the consistency of counterfactual
propositions concerning the values of observables that are
not comeasurable\footnote{Specker did not use the modern term
``counterfactual'', but instead referred to
``infuturabilities'', which had been discussed in a scholastic
context in connection with the problem of whether God's
omniscience extended to knowing the truths of propositions
concerning what would have occurred if some event which did not
happen had in fact happened.}.  One of the points of the
paper was to show that it is possible to conceive of an
implication relation that is not transitive.

The idea is
illustrated with a parable wherein an over-protective seer sets a simple
prediction task to his daughter's suitors. The challenge cannot be met
because the seer asks the suitors for a noncontextual assignment of values but
measures a system for which the statistics are inconsistent with such an
assignment.  The present article considers the parable anew and seeks to
connect it with modern developments in quantum foundations. In particular, we
explore the extent to which the sorts of correlations instantiated in the
seer's prediction game can be achieved in a quantum world. Although the
precise correlations that are required by the seer do not occur in quantum
theory, the prediction game is found to be a good pump for quantum intuitions.
It leads quite naturally to proofs of nonlocality and contextuality, to a
novel kind of complementarity and to the notion of stronger-than-quantum
correlations. Indeed, it provides a narrative thread that is able to weave
together a great number of important modern results. That so much can be
gleaned from this little prediction game is a testament to the depth of
Specker's work. We offer this article as a small tribute to him on the occasion
of his 90th birthday.

\subsection{The parable of the over-protective seer}
We begin by reproducing Specker's parable of the overprotective seer\footnote{Our translation is an amalgam of those provided by Stairs \cite{Spe60} and Seevinck \cite{Seevinck}.}, with
clarifications by us in square brackets:
\begin{quotation}
At the Assyrian School of Prophets in Arba'ilu in the time of King Asarhaddon [(681-669 BCE)],
there taught a seer from Nineva. He was a distinguished representative of his
faculty (eclipses of the sun and moon) and aside from the heavenly bodies, his
interest was almost exclusively in his daughter. His teaching success was
limited; the subject proved to be dry and required a previous knowledge of
mathematics which was scarcely available. If he did not find the student
interest which he desired in class, he did find it elsewhere in overwhelming
measure. His daughter had hardly reached a marriageable age when he was flooded
with requests for her hand from students and young graduates. And though he
did not believe that he would always have her by his side, she was in any case
still too young and her suitors in no way worthy.  In order that the suitors might convince themselves of their unworthiness, he promised them that she would be wed to the one who could solve a prediction task that was posed to them.

Each suitor was taken before a table on which three little boxes stood in a row, [each of which might or might not contain a gem], and was asked to predict which of the boxes contained a gem and which did not.  But no matter how many times they tried, it seemed impossible to succeed in this task.  After each suitor had made his prediction, he was ordered by the father to open any two boxes which he had predicted to be both empty or any two boxes which he had predicted to be both full [in accordance with whether he had predicted there to be at most one gem among the three boxes, or at least two gems, respectively]. But it always turned out that one contained a gem and the other one did not, and furthermore the stone was sometimes in the first and sometimes in the second of the boxes that were opened. But how can it be possible, given three boxes, to neither be able to pick out two as empty nor two as full?

The daughter would have remained unmarried until the father's death, if not for the fact that, after the prediction of the son of a prophet [whom she fancied], she quickly opened two boxes herself, one of which had been indicated to be full and the other empty, and the suitor's prediction [for these two boxes] was found, in this case, to be correct. Following the weak protest of her father that he had wanted two other boxes opened, she tried to open the third. But this proved impossible whereupon the father grudgingly admitted that the prediction, being unfalsified, was valid. [The daughter and the suitor were married and lived happily ever after.]
\end{quotation}

\subsection{Contextuality and Complementarity} \label{Sec:CC}

Specker's parable presents us with apparently impossible correlations; as
he says ``But how can it be possible, given three boxes, to neither be able to pick out two as empty nor two as full?" Indeed, if a suitor reasons classically, then he expects that even if he chooses a configuration of gems at random from among the eight possibilities, it will be the true configuration one time out of eight, and when he opens two boxes he has marked both empty or both full, his prediction will be found to be correct one time out of four. The fact that no suitor manages to succeed after many trials suggests that
this reasoning must be flawed and that whichever two boxes are opened, one
will be found full and the other empty. Such correlations are contextual
in the sense that if one wishes to explain the measurements (opening a
box) as revealing a pre-existing property, then one must imagine that the
outcome of a measurement (or equivalently, the property that is measured)
is context-dependent - whether a gem is seen or not in the first box
depends on whether that box was opened together with the second or
together with the third. The seer's challenge cannot be met by the suitors
because he asks them for a noncontextual assignment of outcomes (i.e. a
specification of whether a gem will be found or not in each box,
independent of which other box is opened with it) but measures a system
for which the statistics are inconsistent with such an
assignment.\footnote{Because the suitors do not fathom this possibility,
they are led to interpret their consistent failure to provide a correct
prediction as a confirmation of the seer's assessment of their worth.  It
is in this sense that the seer's task is devised ``[i]n order that the
suitors might convince themselves of their unworthiness''.}

To imagine a world wherein the parable might occur, Specker must
effectively posit the existence of a system that exhibits a particular
kind of complementarity: the system must be such that three distinct
measurements can implemented upon it, any pair of which can be measured
jointly, but where a joint measurement of all three is not possible. To
see this, one need only note that if all three binary-outcome measurements
\emph{could} be implemented jointly, some pair would necessarily be found
to have correlated outcomes. 

We now ask the obvious question: Can the parable be implemented in quantum
theory? The reader is urged to pause and give this question some thought
before reading on.

There is of course a trivial sense in which the parable can be implemented in
a quantum world, namely the same way that it can be implemented
in a classical world: through a hidden mechanism under the seer's
table and under his control, which inserts and removes gems from
the closed boxes at his will. Such a mechanism would allow the seer to
enforce complementarity and contextuality ``by hand'', so to speak.
However this is clearly not
what Specker had in mind, because had that been the case, the seer
would not have been so easily stymied by his daughter's trick, as
there would have been no reason why the third box could not have
been opened. Rather, the seer seems to be in possession of a set of
``magic'' boxes that have particular, rather than arbitrary, correlations.
Thus in asking the question whether the parable can be
implemented in quantum theory, we mean: does quantum theory allow
for this sort of ``magic'', which would be truly
surprising for a naive suitor familiar only with classical theories,
which do not incorporate contextuality and complementarity at a
fundamental level?

Certainly, both complementarity and contextuality are required
at a fundamental level in quantum theory --
measurements that cannot be implemented jointly, and
correlations that cannot be explained by noncontextual
pre-existing properties (see Ref.~\cite{Peres:Book} for a review).  But what about
the particular correlations of the Specker parable? To get this
kind of contextuality, it is necessary to find a situation
wherein there are very specific sorts of limitations on joint
measurability --- there must exist a triple of measurements that
can only be implemented jointly in pairs. For projective
measurements in quantum theory, this sort of limitation on
joint measurability does \emph{not }occur. The reason is as
follows. Two Hermitian operators can be jointly measured if
and only they are jointly diagonalizable. But if we have three
Hermitian operators $\hat{A}_{1}$, $\hat{A}_{2}$, and
$\hat{A}_{3}$, and each pair of operators is jointly
diagonalizable, then all three are jointly diagonalizable. \
This is true for any number of Hermitian operators
--- one can implement all jointly if and only if one can implement every pair
jointly.

Nonetheless, one can imagine modifying the parable in various different ways to
obtain something for which an analogue \emph{can }be found in quantum theory,
and these different modifications are the topics of the different sections of our article.
In the following we outline each of them in turn.

\subsection{Outline}

We begin by providing,  in Sec.~\ref{sec:prelim}, a formalization of the original parable  as well as some refinements and elaborations thereof, together with definitions of the
key concepts. We then present the four different themes inspired by the parable,  with an interlude on frustrated networks.

\textbf{A double-query $n$-box system allowing only
adjacent queries (Sec.~\ref{sec:quantumcontextuality}).}
The seer could have a set of $n$ boxes, arranged in a
ring, for which only \emph{adjacent} pairs of boxes can be opened jointly. \
For $n$ odd, classical intuition leads one to expect that there must exist at
least one adjacent pair of boxes that are either both full or both empty, but
we can imagine that the seer has a special system wherein, regardless of which
adjacent pair of boxes is opened, it is always the case that one is found full
and the other empty.\ The $n=3$ case, which corresponds to the original
parable, is exceptional because the adjacent pairs constitute all the pairs. \
For $n>3,$ this is not the case, and so there is no longer any obstacle to
finding a set of projective measurements that have the same pattern of joint
measurability as these boxes. Indeed, one can find such sets. There are then
two ways of trying to obtain a quantum analogue of the new parable.

i)\emph{\ Klyachko's proof of contextuality.}\textbf{\ }Find a quantum state
that yields a nonzero probability of anti-correlation for every adjacent pair.
When the overall probability is higher than one could account for
classically, we arrive at a Klyachko-type proof of quantum
contextuality~\cite{A.A.Klyachko:PRL:020403}.

ii) \emph{A new variant of Klyachko's proof of
contextuality.}\textbf{\ } Find
a quantum state that supports the implication from one outcome
to the opposite outcome for every adjacent pair in the ring and  that assigns a non-zero probability to the first outcome in the sequence of inferences. In
conjunction with the transitivity of implication (a consequence of noncontextuality), and the fact that the ring contains an odd number of boxes, this gives rise to a contradiction, thereby
demonstrating the contextuality of quantum theory.

\textbf{A separated pair of single-query 3-box systems (Sec.~\ref{sec:separated}). }
One can imagine
that the seer's three-box system is such that only a single box (rather than
a pair of boxes) can be opened at any given time, but that it is possible to
prepare a pair of three-box systems such that by opening a single box on each
element of the pair, one reproduces the seer's correlations.  Specifically,
if the \emph{same} box is opened on each member of the pair, they are always
found to be both full or both empty, while if \emph{different} boxes are
opened on the two systems, one is always found full and the other empty.
(Classically, one would expect that some pair of boxes on a given wing are
both full or both empty, and by the assumed perfect correlation between the
wings, the same pair is similarly configured on the other wing, implying that
it is impossible to open different boxes on the two systems and always find
anti-correlation rather than correlation.)
Here, we are postulating six distinct measurements (three on each wing) only certain pairs
of which can be implemented jointly, namely, pairs that have one member from
each wing.  So again, there is no obstacle to finding a set of projectors
having this pattern of joint measurability. There are once again two ways of trying
to obtain a quantum analogue of the new parable.

i) \emph{Mermin's proof of Bell-nonlocality. } Find a quantum state that
yields perfect correlation when the same measurement is implemented on the
two wings. Demonstrate that the extent to which it can yield anti-correlation
when different boxes are opened is greater than is possible
in a Bell-local model~\cite{Mermin85}.

ii) \emph{Hardy's proof of Bell-nonlocality. }  Find a chain of choices of
measurement, alternating  between the two parties, and find a quantum state
that yields  implications connecting particular outcomes of all but one measurement
within this chain.
 Demonstrate that there is a nonzero
probability for the kind of correlation exhibited by the last pair in the
chain to be opposite to what one would expect by the transitivity of
implication~\cite{L.Hardy:PRL:1665}.

We also consider generalizations of these nonlocality proofs to rings of $n$
measurements where only adjacent members can be implemented jointly.

\textbf{Interlude on frustrated networks (Sec.~\ref{sec:frustratednetworks}). }
By representing correlations between binary-valued observables by frustrated networks, we provide a simple categorization of some of the contextual and Bell-nonlocal correlations outlined above.

\textbf{A diachronic pair of single-queries of a 3-box system (Sec.~\ref{quantumprepcontextuality}).}
In this case, the seer's three-box system is modified so that only a
single box can be opened at any given time, but that it is
possible to implement two \emph{consecutive} measurements in
such a way that if the same box is opened at the two times,
then the result of the measurement is always reproduced
faithfully, while if different boxes are opened at the two
times, then the results are always different.  In addition, we
impose the constraint that no measurement at the second time can yield any
information about the choice of the measurement at the first
time.

Now it is natural for a suitor to assume that statistical
indistinguishability among a set of choices implies that they
are equivalent at the level of an ontological model.  This
assumption is known as preparation noncontextuality
\cite{Spe05}.  It can be shown that no such preparation-noncontextual model can reproduce the diachronic
(two-time) correlations stated above.
But in quantum mechanics (which violates
preparation noncontextuality \cite{Spe05}),  there are sets of measurements for
which these correlations can be approximated even though the
quantum state after the first measurement reveals no
information about the identity of this measurement.

\textbf{Joint measurability of POVMs (Sec.~\ref{sec:JMPOVM}).}
A final path to a quantum analogue of the  overprotective seer
 (OS) parable is to
ignore the counter-intuitive correlations, and rather
concentrate on the complementarity exhibited by the three
boxes. As discussed above, the pairwise but not triplewise
joint measurability of three observables cannot exist in
quantum mechanics for traditional (projective) measurements.
However, this does not rule out the possibility that there
exists a triple of generalized measurements, described by POVMs
(positive operator valued measures), that can be jointly
measured pairwise but not triplewise. Indeed,  we will exhibit
two specific examples of such a triple of nonprojective
measurements. This thread
connects with some recent results on joint measurability of
POVMs~\cite{BGL95,HRS08,Wolf:PRL:2009,Yu:PRA:2010}.  We demonstrate that this example is not useful for approximating the OS correlations, nor for proving the contextuality of quantum theory.

\section{Preliminaries and a formalization of the parable} \label{sec:prelim}

\subsection{Joint measurability} \label{sec:JM}

We wish to flesh out the original parable by being more specific about the
nature of the correlations posited therein. We shall do this within the
context of operational theories. This is natural because the OS parable was
originally presented by Specker as a  ``toy  theory'' with similarities
to quantum theory, not as a scenario that arises within quantum theory.
We thus need a unified framework to compare the OS theory   both with quantum theory, and with classical
theories (i.e. theories without contextuality or complementarity).
 Also, to make the most of the OS parable we need to embellish the narrative (in a formal way)
 by adding extra assumptions, and this requires considering measurements and preparations
 beyond those discussed by Specker. Finally, we note that in the fields of quantum foundations
 and quantum information, there is currently considerable interest in operational
 ``foil'' theories  such as  Popescu-Rohrlich (PR)  boxes \cite{PR94} and the toy-bit
 theory \cite{Spe04}.

An operational theory is one that specifies the probabilities of each possible
outcome $X$ of each possible measurement procedure $M$ given each
possible preparation procedure $P$. We denote these probabilities by $%
p(X|M;P).$ It will be important for the later discussion of
contextuality to distinguish between a measurement procedure
$M$, which is a specification of a list of instructions of what
to do in the laboratory, and an equivalence class $\mathcal{M}$
of measurement procedures, where two procedures are equivalent
if they yield the same statistics for all preparation
procedures. For instance, the equivalence class
associated with a particular measurement procedure $M_{1}$ is%
\begin{equation}
\mathcal{M}_{1}\mathcal{\equiv }\left\{ M\mid \forall
P:p(X|M;P)=p(X|M_{1};P)\right\} .  \label{eq:opeq4mmts}
\end{equation}
We will refer to this equivalence relation over procedures as \emph{%
operational equivalence.} We will refer to the equivalence classes as simply
\emph{measurements}, and denote them by calligraphic font, while the
measurement procedures\ will be denoted by italic font. Similarly, we define
equivalence classes of preparation procedures. For instance, the equivalence
class associated with a particular preparation procedure $P_{1}$ is%
\begin{equation}
\mathcal{P}_{1}\mathcal{\equiv }\left\{ P\mid \forall
M:p(X|M;P)=p(X|M;P_{1})\right\} . \label{eq:opeq4preps}
\end{equation}
 Given that probabilities of outcomes of measurements depend only on the equivalence classes of the preparation and the measurement procedures, we typically condition only on the latter and write $p(X|\M;\P)$.

We begin by providing an operational definition of joint measurability\footnote{This
definition is a generalization of the notion of \emph{coexistence} of quantum
observables provided in Ref.~\cite{BGL95}.}.  We
consider only measurements with a discrete set of outcomes.

Joint measurability of a set of $N$ measurements can be defined
(recursively) as follows.

\begin{definition}[joint measurability]\label{Dfn:JointMeasurability}
 A set of $N$ measurements $\left\{
\mathcal{M}_{1},\mathcal{M}_{2},...,\mathcal{M}_{N}\right\} $
is jointly measurable if there exists a measurement
$\mathcal{M}$ with the following features: (i) The outcome set of
$\mathcal{M}$ is the Cartesian
product of the outcome sets of $\left\{ \mathcal{M}_{1},\mathcal{M}_{2},...,
\mathcal{M}_{N}\right\}$
and (ii) the outcome distributions for every joint measurement of a
subset $\left\{ \mathcal{M}_{s}|s\in S\right\} \subset
\left\{ \mathcal{M}_{1},\mathcal{M}_{2},...,\mathcal{M}_{N}\right\} $ are
recovered as marginals of the outcome distribution of $\mathcal{M}$ for all preparations $\mathcal{P}$.
Denoting a joint measurement of the subset $S$ by $\mathcal{M}_{S}$, and its outcome by $X_{S}$,
\begin{equation}\label{Eq:SubsetOfMmts}
\mathcal{M}_{S} \equiv \left\{ \mathcal{M}_{s}|s\in S\right\}, \;\;\;
X_{S} \equiv \left\{ X_{s}|s\in S\right\}
\end{equation}
the condition can be expressed as
\begin{equation}
\forall S,\forall \mathcal{P}:p\left( X_{S}|\mathcal{M}_{S};\mathcal{P}%
\right) =\sum_{X_{t}:t\notin S}p\left( X_{1},X_{2},...,X_{N}|\mathcal{M};%
\mathcal{P}\right) .
\end{equation}
\end{definition}

\begin{definition}[$n$-tuple-wise joint measurability]
A set of measurements $\left\{ \mathcal{M}_{1},\mathcal{M}_{2},...,\mathcal{M}%
_{N}\right\} $ is $n$-tuple-wise jointly measurable if every $n$-element subset
(i.e. every $n$-tuple of measurements) is jointly measurable.
\end{definition}

 Clearly joint measurability of all $n$-tuples implies joint
measurability of all $\left( n-1\right) $-tuples, but not vice-versa.

 Finally, we shall sometimes say that the measurements in a set $\left\{
\mathcal{M}_{1},\mathcal{M}_{2},...,\mathcal{M}_{N}\right\} $ exhibit \emph{complementary} if they are not jointly measurable.

We can now be precise about the nature of the correlations in the
overprotective seer's prediction game.

Abstracting from the story of boxes and gems, the parable
posits that there are three distinct measurement procedures,
which we shall denote by $M_{1},$ $M_{2}$ and $M_{3}$
(corresponding to the choice of box).  A key assumption that is
not explicit in Specker's description of the prediction game is
that these three measurement procedures are not operationally
equivalent.  That is, for every pair, there is a preparation procedure
that distinguishes them, that is, some $P$ such that
$p(X|M_1;P)\ne p(X|M_2;P)$.  Making this assumption, we
see that the game assumes the existence of three distinct
equivalence classes of measurement procedures,  which we denote
by $\mathcal{M}_{1},$ $\mathcal{M}_{2}$ and $\mathcal{M}_{3}$.
Furthermore, it is assumed that these are pairwise jointly
measurable.  It follows that there exist three joint
measurements, which we shall denote by $\mathcal{M}_{12},$
$\mathcal{M}_{13}$ and $\mathcal{M}_{23}$ and which, by virtue
of the definition of joint measurability, must have statistics
that reproduce the statistics of $\mathcal{M}_{1},$
$\mathcal{M}_{2}$ and $\mathcal{M}_{3}$ as marginals.  Note that, as the notation suggests, $\mathcal{M}_{12},$ $\mathcal{M}_{13}$ and $\mathcal{M}_{23}$ correspond to distinct equivalence classes of measurement procedures, a fact that follows from the operational distinguishability of $\mathcal{M}_{1},$ $\mathcal{M}_{2}$ and $\mathcal{M}_{3}$.

Note also that within the equivalence class of measurement procedures $\mathcal{M}%
_{1}$, there are procedures $M_{1}^{(2)}$ that involve implementing a joint
measurement of $\mathcal{M}_{1}$ and $\mathcal{M}_{2}$ and discarding the
outcome of the $\mathcal{M}_{2}$ measurement, and there are procedures $%
M_{1}^{(3)}$ that involve implementing a joint measurement of $\mathcal{M}%
_{1}$ and $\mathcal{M}_{3}$ and discarding the outcome of the $\mathcal{M}%
_{3}$ measurement. Which of these two procedures is implemented may be
relevant in a contextual hidden variable model, as we will see.

The seer's trick also requires that there is at least one preparation, call it
$\mathcal{P}_{\ast },$ that yields perfect negative correlations for the
joint measurement of any pair of $\mathcal{M}_{1}$, $\mathcal{M}_{2}$ and $%
\mathcal{M}_{3}.$
Perfect negative correlation for a single joint measurement of $\mathcal{M}_{1}$ and $\mathcal{M}_{2}$
does not imply that one must have equal probability for
$X_{1}=0$, $X_{2}=1$ and $X_{1}=1$, $X_{2}=0$ (the two ways of
achieving perfect negative correlation). However, this
equality does follow from demanding perfect negative
correlation for \emph{all three} joint measurements, as we show
in Appendix \ref{App:SpeckerCorrelations}. Consequently, the correlations are
of the form
\begin{align}
\forall i\ne j:
&p(X_i=0,X_j=1|\mathcal{M}_{i,j};\mathcal{P}_{\ast })= \frac{1}{2}
\nonumber \\
&p(X_i=1,X_j=0|\mathcal{M}_{i,j};\mathcal{P}_{\ast })= \frac{1}{2}
\label{eq:OScorrelations}
\end{align}
We call these the \emph{overprotective seer correlations}, or \emph{OS
correlations. }Note that it follows from this definition that individual
measurements have a uniformly random outcome,
\begin{equation}
p(X_{i}=0|\mathcal{M}_{i};\mathcal{P}_{\ast }) =p(X_{i}=1|\mathcal{M}_{i};\mathcal{P}_{\ast }) =\frac{1}{2}.
\end{equation}

\subsection{The existence of a joint distribution}

The question of joint measurability concerns what is physically possible, not what is logically possible.  If a physical theory postulates measurements that
cannot be jointly implemented, it could still be that there is a joint probability distribution over the
outcomes of these measurements that yields each measurement's statistics as a marginal.

\begin{definition}[existence of a joint distribution]
Consider a set of measurements $\left\{ \mathcal{M}_{1},\mathcal{M}_{2},...,\mathcal{M}%
_{N}\right\} .$ Let $S$ be a subset of their indices
and denote the joint measurement associated with this subset by
$\mathcal{M}_{S}$ and its outcome by  $X_{S}$, as in
Eq.~(\ref{Eq:SubsetOfMmts}). A joint distribution for the measurements
$\left\{ \mathcal{M}_{1},\mathcal{M}_{2},...,\mathcal{M}_{N}\right\} $ is said
to exist if there exists a distribution $p\left( X_{1}...X_{N}|\mathcal{P}%
\right) $ for every preparation $\mathcal{P}$ such that for any measurement $%
\mathcal{M}_{S},$
\begin{equation}
p\left( X_{S}|\mathcal{M}_{S};\mathcal{P}\right) =\sum_{X_{\nu }:\nu \notin
S}p(X_{1}...X_{N}|\mathcal{P}).
\end{equation}
\end{definition}

It is worth noting that within a given theory, the nonexistence of a joint
distribution for some set of measurements implies the physical impossibility of a joint measurement of these.  This follows from the fact that if a joint measurement is possible, then there must exist a joint distribution over the outcomes.  However, the converse implication need not hold.  For instance, there are theories, such as the toy theory of Ref.~\cite{Spe04} which postulate the physical impossibility of certain joint measurements, but for which a joint distribution over outcomes (effectively a hidden variable model) does exist.

The feature of the OS correlations that is at the root of their
peculiarities is the fact that they do not admit of a joint
distribution.

\begin{lemma}[no joint distribution for OS correlations]\label{lemma:noJMforOS}
There is no distribution $p(X_{1},X_{2},X_{3})$ on the three
binary variables $X_{1},X_{2}$ and $X_{3}$ such that the
marginals over pairs of these are of the form of
Eq.~(\ref{eq:OScorrelations}).
\end{lemma}

\begin{proof}
There are eight valuations of $(X_{1},X_{2},X_{3}),$ of the form $%
(0,0,0),$ $(0,0,1),$ $(0,1,0),$ $\dots $ But whichever of
these valuations is assigned non-zero probability in
$p(X_{1},X_{2},X_{3}),$ one of the three pairs $(X_{1},X_{2}),$
$(X_{1},X_{3})$ or $(X_{2},X_{3})$ will have non-zero
probability assigned either to the valuation $(0,0)$,
\textquotedblleft both boxes empty\textquotedblright , or to
the valuation $(1,1)$ \textquotedblleft both boxes
full\textquotedblright . For this pair, either $p(0,0)>0$ or
$p(1,1)>0,$ so that perfect anti-correlation is not achieved.
\end{proof}

Given the discussion above, this result has immediate (negative) consequences for the
possibility of implementing a triplewise joint measurement of $\mathcal{M}_{1},$ $%
\mathcal{M}_{2}$ and $\mathcal{M}_{3}$.
\begin{corollary}
Measurements $\mathcal{M}_{1},$ $\mathcal{M}_{2}$ and $%
\mathcal{M}_{3}$ that can be pairwise jointly measured and that achieve the OS
correlations of Eq. (\ref{eq:OScorrelations}) cannot be triplewise jointly
measured.
\end{corollary}

\subsection{Measurement-noncontextual ontological models} \label{sec:defnNC}

In this article, we will make use of the generalized notion of
noncontextuality introduced in Ref.~\cite{Spe05}, which is operational
insofar as it is defined for ontological models of any operational theory,
not just quantum theory. An ontological model of an operational theory
specifies: (i) a set $ \Lambda $ of ontic (i.e. real, physical) states
$\lambda $; (ii) for each preparation procedure $P$, a distribution
$p(\lambda|P)$ describing the probability that the ontic state of the system
subsequent to the preparation procedure $P$ is $\lambda $; (iii) for each
measurement procedure $M$, a \emph{response function} $p(X|M;\lambda )$
describing the conditional probability of obtaining outcome $X$ given ontic
state $\lambda . $ Finally, one must recover the statistics of the
operational theory as
follows:%
\begin{equation}
p(X|M;P)=\sum_{\lambda} p(X|M;\lambda )p(\lambda |P).
\end{equation}
Here we have taken $\lambda$ to be a discrete variable for simplicity.

An ontological model is said to be \emph{measurement-noncontextual} if any two measurement procedures that are
operationally equivalent [in the sense of Eq.~(\ref{eq:opeq4mmts})] are represented similarly in the model:
\begin{eqnarray}
\forall P:p(X|M;P)=p(X|M^{\prime };P) \nonumber \\
\Rightarrow \forall \lambda:p(X|M;\lambda )=p(X|M^{\prime };\lambda ).  \label{eq:MNC}
\end{eqnarray}
Equivalently, the condition is that the response function for a measurement
procedure $M$ depends only on its operational equivalence class $\mathcal{M},
$ that is,%
\begin{equation}
p(X|M;\lambda )=p(X|\mathcal{M};\lambda ).
\end{equation}%

An ontological model is said to be \emph{outcome-deterministic
}for a measurement procedure $M$ if the outcome is uniquely
determined for every
ontic state,%
\begin{equation} \label{eq:OD}
\forall \lambda \in \Lambda :p(X|M;\lambda )\in \{0,1\}.
\end{equation}%

 The traditional notion of a noncontextual ontological
model of quantum theory incorporated both the assumption of
measurement noncontextuality and that of outcome determinism
for projective measurements.  Here, we will follow
Ref.~\cite{Spe05} and distinguish these assumptions so as not
to conflate issues about determinism with issues about
noncontextuality.  To avoid terminological confusion, we shall
say that an ontological model of quantum theory is
\emph{traditionally-noncontextual} if it is both measurement-noncontextual [in the sense of Eq.~(\ref{eq:MNC})] and outcome-deterministic for projective measurements.
Any proof of the impossibility of a traditionally-noncontextual model of quantum theory will be called a proof of the \emph{Kochen-Specker theorem}.

As it turns out, there is a close connection between the
existence of a joint distribution and noncontextuality:

\begin{theorem}\label{Thm:Fine}
For a given set of measurements, if there exists a
measurement-noncontextual and outcome-deterministic ontological
model then there exists a joint distribution for their outcomes.
\end{theorem}

The proof is provided in Appendix \ref{App:ProofFineTheorem}.  This is a slight generalization of half of a theorem by Fine \cite{Fine82}. Combining this theorem with the nonexistence of a joint
distribution for the OS correlations (lemma
\ref{lemma:noJMforOS}), we have:

\begin{corollary} \label{corollary:OScontextuality}
There is no measurement-noncontextual and outcome-deterministic
ontological model of the OS correlations of
Eq.~(\ref{eq:OScorrelations}).
\end{corollary}

It is also possible to write down inequalities which must be
satisfied by the experimental statistics if these are to admit
of an explanation in terms of a measurement-noncontextual and
outcome-deterministic model.
 We will call these \emph{Kochen-Specker inequalities}.
For the case of the OS correlations, if we imagine such a model then each box must be either empty or full.  Consequently, if we choose a pair of boxes uniformly at random at most two of the three pairs could exhibit anti-correlation, so that the
probability of obtaining anti-correlated outcomes is bounded
above by $2/3.$   More precisely, if $p\left( X_{i}\neq
X_{i\oplus 1}|\mathcal{M}_{i,i\oplus 1};\mathcal{P}\right)$
denotes the probability of obtaining anti-correlated outcomes in
a joint measurement of $\mathcal{M}_{i}$ and
$\mathcal{M}_{i\oplus 1}$, where $\oplus$ denotes addition
modulo 3, then the average probability of success is
\begin{equation} \label{eq:NCboundonR3}
R_3  \equiv \sum_{i=1}^{3}\frac{1}{3}p\left( X_{i}\neq
X_{i\oplus 1}|\mathcal{M}_{i,i\oplus 1};\mathcal{P}\right) ,
\end{equation}
and it satisfies
\begin{equation}
R_3\le R^{NC}_3=\frac{2}{3}.
\end{equation}
This is a Kochen-Specker inequality.

It is sometimes useful to express Kochen-Specker inequalities in an algebraic
form. We define new variables $\bar{X}_i = (-1)^{X_i}$, so that $\bar{X}_i = +1(-1)$ when $X_i = 0(1)$.
Using angle brackets to denote averages, we consider the  following combination of  correlation functions
\begin{equation}\label{Eq:Dfn:S3}
    \S_3 \equiv \an{\bar{X}_1\bar{X}_2}+\an{\bar{X}_2\bar{X}_3}+\an{\bar{X}_3\bar{X}_1}.
\end{equation}
Then the inequality takes the form
\begin{equation}\label{Ineq:S3}
    \S_3\ge \S^{NC}_3  =-1,
\end{equation}
The OS correlations, however, require $\an{\bar{X}_i\bar{X}_{i\oplus 1}}=-1$
for all $i$ and hence $\S_3=-3$, clearly violating the bound.

With the correlations in this form, one can also express a proof of the impossibility of an outcome-deterministic noncontextual model in the algebraic manner introduced by Mermin \cite{Mermin2theorems}.  Assuming that $\bar{X}_i \in \{ +1,-1\}$ has a value independent of context, the OS correlations require that these values satisfy the following algebraic relations.
\begin{eqnarray}
\bar{X}_1 \bar{X}_2 = -1, \\ \nonumber
\bar{X}_1 \bar{X}_3 = -1, \\ \nonumber
\bar{X}_2 \bar{X}_3 = -1.
\end{eqnarray}
However, these relations cannot be satisfied because the
product of the left-hand-sides is $\bar{X}^2_1 \bar{X}^2_2
\bar{X}^2_3 = +1$, while the product of the right-hand-sides is
$-1$.

Any theory that realizes the OS correlations fails to admit of
a measurement-noncontextual and outcome-deterministic ontological model.
However, as explained in the introduction, the kind of
complementarity one requires to achieve these correlations --
three measurements that are pairwise but not triplewise jointly
measurable --- cannot arise for projective measurements in
quantum theory. In Sec.~\ref{sec:quantumcontextuality}, we turn to the
modifications of the parable that do have a counterpart in
quantum theory.

\subsection{Preparation noncontextuality}\label{sec:PNC}

The notion of measurement noncontextuality defined in Eq.~(\ref{eq:MNC}) is motivated by a kind of equivalence principle: in the absence of observable differences between measurement procedures (i.e. differences in their statistics) one should not posit differences in their representations in the ontological model.  In Ref.~\cite{Spe05} it was argued that the same principle should lead one to an assumption of noncontextuality for preparation procedures.  Specifically, an ontological model is said to be \emph{preparation noncontextual} if any two preparation procedures that are operationally equivalent [in the sense of Eq.
(\ref{eq:opeq4preps})] are represented equivalently in the model:
\begin{eqnarray}
\forall M:p(X|M;P)=p(X|M;P^{\prime }) \nonumber \\
\Rightarrow \forall \lambda:p(\lambda|P)=p(\lambda|P^{\prime }).  \label{eq:PNC}
\end{eqnarray}
Preparation noncontextuality can also be characterized as the condition that the distribution for a preparation
procedure $P$ depends only on its operational equivalence class $\mathcal{P},
$ that is,
\begin{equation}
p(\lambda|P )=p(\lambda|\mathcal{P} ).
\end{equation}

Given their similar motivations, someone who endorses measurement noncontextuality ought also to endorse preparation noncontextuality just as enthusiastically.  One should endorse both notions or neither.  Therefore, it is most natural to ask about the possibility of an ontological model that is both preparation-noncontextual and measurement-noncontextual.  We will call such models  \emph{generalized-noncontextual} \footnote{ In Ref~\cite{Spe05}, these were called ``universally noncontextual.''}.  In this paper, we will consider suitors faced with the seer's prediction problem
who are committed to the kind of equivalence principle described above and therefore to  generalized noncontextuality.

Inequalities that must be satisfied by the experimental statistics if these are to admit of a generalized-noncontextual model will be called simply \emph{noncontextuality inequalities}.  Note that our terminology distinguishes such inequalities from the Kochen-Specker inequalities of the previous section: Kochen-Specker inequalities express constraints on statistics when one assumes outcome determinism in addition to measurement noncontextuality, while noncontextuality inequalities rely on no such assumption of determinism.  An example of a noncontextuality inequality will be provided in Sec.~\ref{quantumprepcontextuality}.

\subsection{Justifying outcome determinism}\label{sec:JustifyingOD}

Note that a commitment to the kind of equivalence principle described above does not
obviously provide any grounds for assuming outcome determinism
for measurements, Eq.~\eqref{eq:OD}.  Thus, faced with the OS
correlations and corollary \ref{corollary:OScontextuality}, a
suitor might simply deny outcome determinism to salvage
measurement noncontextuality.  For instance, seeing the
correlations in the seer's prediction game, a clever suitor
might hypothesize that they are explained by the following sort
of model. There is an ontic variable that flags when the
preparation $\mathcal{P}_*$ was implemented and if it was, the
measurements $\mathcal{M}_{12}$, $\mathcal{M}_{13}$ and
$\mathcal{M}_{23}$ each generate the outcomes $(0,1)$ and
$(1,0)$ uniformly at random.  Such an ontological model would
violate outcome determinism, but would preserve measurement
noncontextuality.

On the other hand, the assumption of outcome determinism can sometimes be
shown to be a consequence of preparation noncontextuality.  If such a
justification is forthcoming, then the OS correlations cannot be explained by
any ontological model that is generalized-noncontextual. For instance, in quantum theory, the assumption of outcome
determinism for projective measurements can be derived from preparation
noncontextuality, as shown in Ref.~\cite{Spe05}.  Therefore, in quantum
theory the conjunction of measurement noncontextuality and outcome
determinism for projective measurements ---  i.e. the assumption of traditional
noncontextuality of an ontological model --- is implied by the assumption of  generalized
noncontextuality  and all the no-go theorems for the former are
no-go theorems for the latter. In Sec.~\ref{sec:quantumcontextuality}, we
will provide proofs of the failure of traditional noncontextuality in quantum
theory using a generalization of the OS correlations.  Given the result just
mentioned, such proofs also demonstrate the failure of  generalized
noncontextuality.

Much of this article makes statements about correlations that are not found in
quantum theory but can easily be imagined to occur in more general
operational theories.  In such theories, a natural analogue of the
notion of a projective measurement can be defined.
The question thus arises of whether the assumption of preparation
noncontextuality might imply outcome determinism for such measurements for
an ontological model of a general operational theory.  The question is
currently open, but we conjecture that it has a positive answer.

Fortunately, we can still draw some negative conclusions about  generalized noncontextuality  in operational theories without settling this conjecture.  Specifically, in Sec.~\ref{quantumprepcontextuality}, we will demonstrate how a slight modification of the seer's game yields a set of correlations that fails to admit of a preparation-noncontextual ontological model.

\section{No-go theorems for measurement-noncontextual and outcome-deterministic models}\label{sec:quantumcontextuality}

\subsection{A double-query $n$-box system allowing only adjacent queries}
One way to generalize Specker's parable is to consider $n>3$
boxes, and allow only certain pairs to be opened jointly.  In
particular, one can imagine the boxes to be arranged in a ring
with adjacent pairs being the only ones that can be opened
jointly.  The resulting pattern of joint measurability can be
reproduced in quantum theory because there exist ordered sets
of $n>3$ projectors for which adjacent elements commute (where adjacency is determined modulo $n$).  If $n$
is odd, then for every  deterministic and noncontextual  assignment of
gems to boxes that the suitor might make, there must exist at
least one adjacent pair of boxes that are either both full or
both empty. Indeed, given any assignment of gems to boxes, if
we choose an adjacent pair of boxes uniformly at random, the
probability of obtaining anti-correlated outcomes is bounded
above by $(n-1)/n$.  We then imagine that the seer has a
special system such that, regardless of which adjacent pair of
boxes is opened, it is always the case that one is found full
and the other empty.
\footnote{Note that for
\emph{even} $n$, one cannot develop an interesting parable
because there are assignments of gems to boxes wherein no
adjacent pair is both full or both empty.}  For these correlations, unlike those  described in the original parable, one \emph{can} find a quantum analogue.
Although this analogue does not allow the seer to always defeat the suitor's prediction,
the probability of finding perfect anti-correlation between a
pair of adjacent boxes can be greater than  the success rate of
$(n-1)/n$ expected by classical reasoning.

Let us consider this situation more carefully. We are imagining an odd number
$n\geq 5$ of measurements, $\{\mathcal{M}_a|a=1,...,n\}$, such that for all
$a$, $\mathcal{M}_a$ and $\mathcal{M}_{a\oplus 1}$ are jointly measurable by
a measurement $\mathcal{M}_{a,a\oplus1}$ (here $\oplus$ denotes addition
modulo $n$), and that there is at least one preparation, call it
$\mathcal{P_*}$, such that the outcomes of all of these pairs of measurements
are anti-correlated. By a generalization of the argument provided in Appendix~\ref{App:SpeckerCorrelations},
the correlations must be of the form
\begin{align}
\forall\, a: &p(X_a=0,X_{a\oplus 1}=1|\mathcal{M}_{a,a\oplus 1};\mathcal{P_*})=\frac{1}{2} \nonumber \\
&p(X_a=1,X_{a\oplus 1}=0|\mathcal{M}_{a,a\oplus 1};\mathcal{P_*})=\frac{1}{2}.
\end{align}
We will call these the \emph{double-query $n$-box OS correlations}.

By an argument analogous to the one proving lemma \ref{lemma:noJMforOS} one
can show that there is no joint distribution over all the $X_a$ that
reproduces these correlations as marginals.  It then follows from theorem
\ref{Thm:Fine} that there is no measurement-noncontextual and outcome-deterministic ontological model of these correlations.

Indeed, if we choose an adjacent pair of boxes uniformly at random, the probability $R_n$ of obtaining anti-correlated outcomes,
\begin{equation}
R_n\equiv \sum_{a=1}^{n}\frac{1}{n}
p\left( X_{a}\neq X_{a\oplus 1}|\mathcal{M}_{a},\mathcal{M}_{a \oplus 1};\mathcal{P}\right),
\end{equation}
is clearly bounded above,
\begin{equation} \label{Eq:NCBound4doublequerynbox}
R_n \leq 1-\frac{1}{n}
\end{equation}
(because at most $n-1$ pairs can be anti-correlated if $n$ is
odd).  The double-query $n$-box OS correlations yield $R_n=1$, maximally violating this Kochen-Specker inequality.

We may equivalently state the restriction as follows. Following the convention
established in Sec.~\ref{sec:defnNC}, we define
$\bar{X}_a = (-1)^{X_a} \in \{+1,-1\}$.
For all measurement-noncontextual and deterministic assignments of the
value $\bar{X}_a$,  at most $n-1$ of the pair-wise products can be $-1$,  so that:
\begin{equation}\label{Ineq:Noncontextual:Sn}
    \S_n \equiv \an{\bar{X}_1\bar{X}_2}+\an{\bar{X}_2\bar{X}_3}+\cdots+\an{\bar{X}_n\bar{X}_1}\ge-(n-2),
\end{equation}
whereas the double-query $n$-box OS correlations give $\S_n=-n$.

Again, a simple algebraic way of
manifesting  the fact that  the double-query $n$-box correlations do not admit of a measurement-noncontextual and outcome-deterministic model  is
that they require $\bar{X}_a \in \{ +1,-1\}$ such that
\begin{gather}
\bar{X}_1 \bar{X}_2 = -1, \nonumber\\
\bar{X}_2 \bar{X}_3 = -1, \nonumber \\
\vdots \nonumber \\
\bar{X}_{n-1} \bar{X}_n = -1, \nonumber \\
\bar{X}_n \bar{X}_1= -1,
\end{gather}
but the product of the left-hand-sides is $\bar{X}^2_1
\bar{X}^2_2 \cdots \bar{X}^2_n = +1$, while the product of the
right-hand-sides is $-1$.

We now consider what values of $R_n$ and $\S_n$ can be achieved in quantum theory.

\subsection{Klyachko's proof of the Kochen-Specker theorem}

We require $n$ Hermitian observables $\hat X_1, \dots,
\hat{X}_n$ --- each having eigenvalues 0 and 1 --- associated
with the $n$ measurements $\M_1, \ldots, \M_n$. As
discussed in Sec.~\ref{Sec:CC}, for the specific case of
$n=3$, the pairwise commutativity of $\hat X_1$, $\hat X_2$ and
$\hat X_3$ implies their triplewise commutativity and
consequently the existence of a triplewise joint measurement
and of a measurement-noncontextual and outcome-deterministic model.\footnote{Another way to see that the
Kochen-Specker inequality \eqref{Ineq:S3} cannot be violated
in quantum theory is that we can treat $\S_3$ as a polynomial
of commuting variables, and thus its minimum can be attained by
assigning the value to each variable in a noncontextual and deterministic
manner.}

Nonetheless, we can obtain something interesting for odd $n$
greater than 3. We begin with the case of $n=5$.  A no-go theorem of this sort has recently been given by Klyachko
\cite{Klyachko} (see also Refs.~\cite{A.A.Klyachko:PRL:020403}
and ~\cite{Guhne}).  The construction is as follows.  We
consider a quantum system described by a 3-dimensional
Hilbert space, and all of the states we consider
require only real-valued coefficients in some basis. Thus the system
can be visualized in 3-dimensional Euclidean space.
The observables are projectors
$\hat{X}_a=|l_a\rangle\!\langle l_a|$, where the vectors
$\{|l_a\rangle:a=1,...,5\}$ are of the form
\begin{equation} \label{eq:la}
|l_a\rangle = (\sin\theta \cos\varphi_a , \sin\theta \sin\varphi_a, \cos\theta),
\end{equation}
and $\varphi_a = \tfrac{4\pi a}{5}$
so that the sequence of vectors forms a pentagram, as in
Fig.~\ref{Fig:Pentagrams}.  The angle $\theta$ is chosen such
that vectors adjacent in the sequence are orthogonal, $\langle
l_a | l_{a\oplus1} \rangle =0$, where $\oplus$ denotes sum
modulo 5. As a result of this orthogonality relation, adjacent
observables $\hat X_a$, $\hat X_{a\oplus1}$ are indeed jointly
measurable. It is clear that such a value of $\theta$ exists
because as it varies from 0 to $\tfrac{\pi}{2}$, the angle
between adjacent vectors varies from 0 to $\frac{4\pi}{5}$.  In
fact, orthogonality is achieved at $\cos \theta =
\frac{1}{\sqrt[4]{5}}$.

\begin{figure}[h!]
    \begin{center}
    \includegraphics[width=7cm]{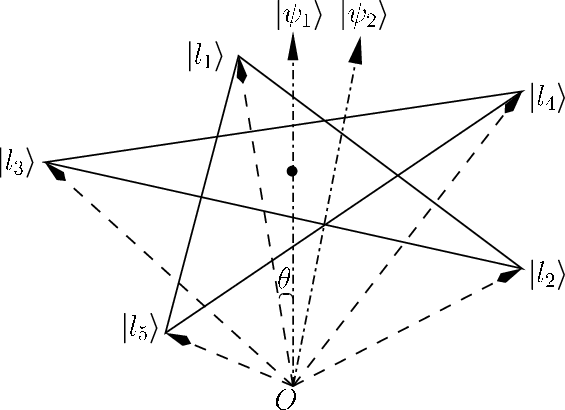}
     \end{center}
    \caption{\label{Fig:Pentagrams}
    Quantum states and observables used for Klyachko's proof of contextuality and
    the proof of contextuality via the failure of transitivity of implication.
    }
\end{figure}

Now consider a preparation of the quantum state $|\psi_1\rangle$
corresponding to the vector lying along the symmetry axis of the pentagram,
such that the angle between it and each of the $|l_a\rangle$ is $\theta$. In
a measurement of any adjacent pair of observables $\hat X_a$, $\hat
X_{a\oplus1}$, either just one of them yields the outcome 1, in which case
the outcomes are anti-correlated, or both yield the outcome 0. The probability
for anti-correlation is $ 2\cos^2(\theta)=2/\sqrt{5}$, which implies that
the Kochen-Specker bound of Eq.~(\ref{Eq:NCBound4doublequerynbox}) is
violated,
\begin{equation}
	R_5=R^{\textrm{quantum}}_5 =  \frac{2}{\sqrt{5}}\simeq  0.89442
	\nleq \frac{4}{5}.
\end{equation}
Equivalently, with the observables $\bar{X}_a=2\hat{X}_a-\unit$,
 where $\unit$ is the identity operator,  the
state $|\psi_1\rangle$ achieves $\S_5=5-4\sqrt{5}\approx -3.9443\not\ge-3$.

The value $2/\sqrt{5}$ is in fact the maximum possible quantum
violation of this Kochen-Specker inequality.  We show this in
Appendix~\ref{App:MaxQV:Klyachko} with the help from the
converging hierarchy of semidefinite programming (SDP) tools
discussed in Ref.~\cite{QMP.Hierarchy} [see also
Eq.~\eqref{Eq:MaxKlyachkoViolation:Prob:n} and
Eq.~\eqref{Eq:MaxKlyachkoViolation:n} below].

Note that unlike the no-coloring proofs of the Kochen-Specker
theorem, this is a \emph{state-specific} proof
\cite{Clifton,Cabellostatespecific}.\footnote{Note also that
the claim in Klyachko {\em et~al.}
\cite{A.A.Klyachko:PRL:020403} that the inequality in question
provides a ``test of arbitrary hidden variables model, context
free or not'' is mistaken. If the values assigned to
observables could be context-dependent, there would be no
contradiction.} In fact, for 3-dimensional quantum state, this
is a state-specific proof that involves the smallest set of
vectors $\{\ket{l_a}\}$ satisfying the orthogonality relation
$\langle l_a | l_{a\oplus1}
\rangle=0$~\cite{A.A.Klyachko:PRL:020403}.\footnote{The
2-vector and 3-vector cases are trivial. For the 4-vector case,
note that the orthogonality relation $\langle l_a |
l_{a\oplus1} \rangle=0$ implies that some pair of non-adjacent
vectors must be collinear which, in turn, implies that all the
four projectors $\proj{l_a}$ must commute, and therefore cannot
violate any Kochen-Specker inequality.}

\subsubsection{Generalization to all odd $n$.}
One can generalize Klyachko's no-go result to all odd $n$ as follows. Define $n$
observables by the projectors onto vectors $\{|l_a\rangle:a=1,...,n\}$
defined as in Eq.~(\ref{eq:la}) but with $\varphi_a = \tfrac{n-1}{n}\pi a$
and with $\theta$ chosen such that $\langle l_a | l_{a\oplus1}
\rangle =0$, where $\oplus$ denotes sum modulo $n$.  This is
achieved when $\cos^2\theta = \cos(\pi/n)/(1+\cos(\pi/n)).$
This set of $n$ vectors forms what is known as an $\{ n /
\frac{n-1}{2}\}$ star polygon \cite{Coxeter}.  The $\{5/2\}$,
$\{7/3\}$ and $\{9/4\}$ star polygons are depicted in
Fig.~\ref{Fig:StarPolygons}.  Again, preparing the quantum
state on the symmetry axis of the star polygon, the probability
of anti-correlation for adjacent observables violates the
Kochen-Specker bound of Eq.~(\ref{Eq:NCBound4doublequerynbox})
with
\begin{equation}\label{Eq:MaxKlyachkoViolation:Prob:n}
R_n=R_n^{\textrm{quantum}} =
\frac{2\cos(\tfrac{\pi}{n})}{{1+\cos(\tfrac{\pi}{n})}} \nleq 1-\frac{1}{n}
\end{equation}
or equivalently, the Kochen-Specker bound of
Eq.~\eqref{Ineq:Noncontextual:Sn} with
\begin{align}
\S_n=\S_n^{\textrm{quantum}} &= n-2nR_n^{\textrm{quantum}}\nonumber\\
\label{Eq:MaxKlyachkoViolation:n}
&=n-\frac{4n\cos(\tfrac{\pi}{n})}{1+\cos(\tfrac{\pi}{n})} \ngeq 2-n.
\end{align}
As with the $n=5$ case, these values also represent the
strongest possible quantum violation of these Kochen-Specker
inequalities, as is shown in
Appendix~\ref{App:MaxQV:Klyachko}. At large $n$, the quantum probability approaches
unity quadratically as
\begin{equation} \label{quadratic}
R_n^{\textrm{quantum}}  \simeq 1-\frac{\pi^2}{4n^2},
\end{equation}
in contrast to the linear approach to unity of the Kochen-Specker bound.

\begin{figure}[h!]
    \begin{center}
    \includegraphics[width=7cm]{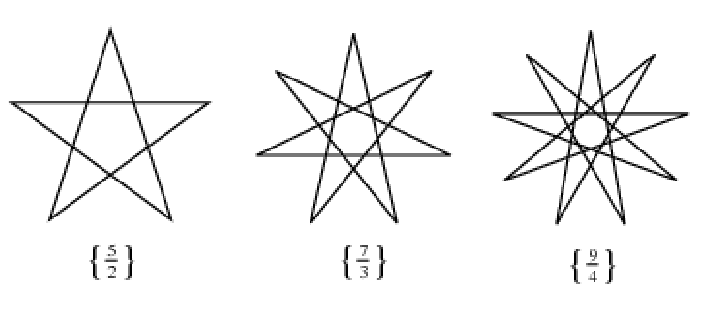}
   \end{center}
    \caption{\label{Fig:StarPolygons}
    $\{ n / \frac{n-1}{2}\}$ star polygons for $n=5,7,$ and $9$.}
\end{figure}

It is worth emphasizing that by using the quantum correlations for $n$ measurements, the seer can achieve something very close to the ends he achieved in the original parable. Specifically, the seer can construct a prediction game such that suitors who reason classically think the game is fair (i.e. they think it is highly likely that \emph{some} suitor will win) when in fact it is not (because classical reasoning does not apply and it is actually highly unlikely that any suitor will win).

The prediction game that meets the seer's ends is as follows.  The suitor is asked to pick an adjacent pair of boxes that he believes to be both empty or both full and to open those.  If his prediction for those two boxes is correct, the suitor wins, otherwise he loses.  With what probability will a suitor who reasons classically expect to win?  We presume that he knows the seer to be adversarial and so he reasons that the seer has prepared a classical configuration which makes his [the suitor's] task as difficult as possible.  He reasons therefore that the configuration is one wherein only \emph{one} adjacent pair of boxes is both full or both empty (by his classical lights, he knows that there must be at least one such pair for an odd number of boxes). Thus the suitor expects his probability of winning to be the probability that he has guessed correctly which of all the $n$ pairs is the correlated one, times the probability that he has guessed their contents correctly -- overall, a probability of $1/2n$.  In fact, the probability of the suitor's prediction coming true is only of order $1/n^2$ in the quantum scheme described above.  Let us say the number of suitors is $l$, assumed large. Then if the seer chooses the number of boxes  $n$ such that $n \ll l \ll n^2$,
the suitors believe it to be very likely that one of them will win when in fact it is very likely that none of them will win.

\subsection{A proof of the Kochen-Specker theorem based on the failure of transitivity of implication}\label{sec:variantKlyachko}

Specker's intent in introducing his parable was to demonstrate the logical possibility of a failure of the transitivity of implication.   The idea is straightforward.  Suppose $s_1$, $s_2$ and $s_3$ are propositions that assert the presence of a gem in boxes 1, 2 and 3 respectively, and $\lnot s_1$, $\lnot s_2$ and $\lnot s_3$ assert their negations.  We have $s_1 \implies \lnot s_2$ (because boxes $1$ and $2$ are never found both full), and $\lnot s_2 \implies s_3$ (because boxes $2$ and $3$ are never found both empty).  If implication were transitive, then we could conclude that $s_1 \implies s_3$.  But in fact we have  $s_1 \implies \lnot s_3$ (because boxes $1$ and $3$ are never found both full).  Therefore, assuming a gem is sometimes found in box 1, transitivity fails.

Specker's 1960 article was framed within the tradition of \emph{quantum logic}, and although some researchers have proposed that quantum theory might require us to abandon some of the rules of classical logic as rules of right-reasoning (see, for example, Ref.~\cite{I.Pitowsky:Book}), we will not consider this possibility here.  Indeed, if we incorporate the context of a measurement in the propositions, so that we distinguish $s_1$, finding a gem in box 1 in the context of measuring box 1 with box 2, from $s'_1$, finding a gem in box 1 in the context of measuring box 1 with box 3, then the transitivity of implication can be salvaged and there is no challenge to classical logic.

Nonetheless, the failure of the transitivity of implication provides another perspective on how to generate no-go results for measurement-noncontextual outcome-deterministic models. In such models, implications among value assignments of observables are necessarily transitive because these value assignments do not depend on the context of the measurement.  A failure of the transitivity of
implication therefore implies the impossibility of such a model.

In the case of the double-query $n$-box OS correlations, if $n$ is odd, the perfect
anti-correlations justify the following implications around the ring of boxes,
\begin{eqnarray}
X_1=1 \Longrightarrow X_2=0 \Longrightarrow X_3=1 \\ \nonumber
\Longrightarrow
... \Longrightarrow X_n=1 \Longrightarrow X_1=0.
\end{eqnarray}
By the transitivity of implication, we would conclude that
$X_1=1 \Longrightarrow X_1=0$.  Given that $X_1=1$ is sometimes observed, one has a contradiction.
Consequently, the observation of the double-query $n$-box OS
correlations implies the impossibility of a measurement-noncontextual
outcome-deterministic ontological model.

We now demonstrate the existence of a quantum analogue of this
argument in the case of $n=5$. Specifically, we demonstrate
that for the set of observables in Klyachko's proof, specified
in Eq.~(\ref{eq:la}) and depicted in Fig.~\ref{Fig:Pentagrams},
there is a quantum state such that
\begin{eqnarray} \label{Eq:5boxseqofimplications}
X_1=1 \Longrightarrow X_2=0 \Longrightarrow X_3=1 \\ \nonumber
\Longrightarrow
X_4=0 \Longrightarrow X_5=1 \Longrightarrow X_1=0.
\end{eqnarray}
First, note that an inference from $X_a=1$ to $X_{a\oplus 1}=0$ can be made
independently of the quantum state, because for any pair of orthogonal
projectors, at most one of them can take the value 1.  However, an inference
from $X_a=0$ to $X_{a\oplus 1}=1$ is only true for certain quantum states
because a pair of projectors may both be assigned the value 0.  To ensure
that $X_a=0$ implies $X_{a\oplus 1}=1$, we must choose a quantum state that
lies in the span of the vectors $|l_a\rangle$ and $|l_{a\oplus 1}\rangle$ in
Hilbert space.  This way, the vector orthogonal to this span is assigned
value 0, such that if $|l_a\rangle$ is assigned value 0, $|l_{a\oplus
1}\rangle$ must be assigned the value 1.  Starting with an assignment of
$X_1=1$, we need to make the $X_a=0$ to $X_{a\oplus 1}=1$ inference twice in
the pentagram: from $X_2=0$ to $X_3=1$ and from $X_4=0$ to $X_5=1$.
Consequently, we need a quantum state that lies in the subspace (plane)
spanned by $|l_2\rangle$ and $|l_3\rangle$ but also in the subspace spanned
by $|l_4\rangle$ and $|l_5\rangle$.  Fortunately, these subspaces intersect
on a ray (see Fig.~\ref{Fig:Pentagrams}), and therefore we take the quantum
state to be the one associated with that ray, indicated in
Fig.~~\ref{Fig:Pentagrams} as $|\psi_2\rangle$.

Therefore, assuming a preparation of the state
$|\psi_2\rangle$, we have the sequence of implications of
Eq.~(\ref{Eq:5boxseqofimplications}).  By the transitivity of
implication, we can conclude that $X_1=1 \Longrightarrow X_1=0$.
Given that $X_1=1$ is assigned non-zero probability by
$|\psi_2\rangle$, specifically,
$p=1-\tfrac{2}{\sqrt{5}}  \approx 0.1056$,  we have derived
a contradiction from the assumption of the transitivity of
implication, and therefore also from the assumption of an ontological model that is
measurement-noncontextual and outcome-deterministic for projective measurements
(i.e. traditionally-noncontextual)
 \footnote{A slightly different way of seeing the
contradiction is that transitivity of implication specifies
that $X_1=1$ implies $X_5=1$, whereas by a joint measurement of
$X_1$ and $X_5$, we would infer that $X_1=1$ implies $X_5=0$.}.

 This is a proof of the Kochen-Specker theorem which is analogous to Hardy's proof of
Bell's theorem, described in Sec.~\ref{Sec:Hardy}.  Interestingly,
it is not possible to generalize this type of proof to the case
of $n>5$ using a set of vectors that form an $\left\{ n /
\frac{n-1}{2}\right\}$ star polygon. For instance, in the case
of $n=7$, if we start with $X_1=1$, in order to make the
inference from $X_2=0$ to $X_3=1$, from $X_4=0$ to $X_5=1$ and
from $X_6=0$ to $X_7=1$, the quantum state would have to lie in
each of the following three subspaces: the one spanned by
$|l_2\rangle$ and $|l_3\rangle$, the one spanned by
$|l_4\rangle$ and $|l_5\rangle$ and the one spanned by
$|l_6\rangle$ and $|l_7\rangle$. But although any pair of these
subspaces intersect along a ray, the three do not, so there is
no quantum state that does the job.

The state-specific Kochen-Specker proof we have just presented turns out to be
related to Clifton's 8-ray Kochen-Specker proof \cite{Clifton}. The latter
makes use of the famous 8-vertex subgraph of the original 117-vertex
Kochen-Specker proof \cite{KS67}. Clifton's proof also has an interesting connection with the
pre and post-selection effect known as the ``three-box paradox''~\cite{3box}, as shown in Ref.~\cite{LS05}.
A connection between Klyachko's Kochen-Specker proof and the 8-ray proof (as well as Hardy's nonlocality proof) has also been noted previously in Ref.~\cite{BBCGL09}.

To see how our proof is related to Clifton's, let us denote the vector orthogonal to the span of $\left\vert l_{2}%
\right\rangle $ and $\left\vert l_{3}\right\rangle $ by $\left\vert
\chi\right\rangle $ and the one orthogonal to the span of $\left\vert
l_{4}\right\rangle $ and $\left\vert l_{5}\right\rangle $ by $\left\vert
\chi^{\prime}\right\rangle ,$ then the orthogonality relations of the eight
vectors $\{\left\vert l_{1}\right\rangle ,\left\vert l_{2}\right\rangle
,\left\vert l_{3}\right\rangle ,\left\vert l_{4}\right\rangle ,\left\vert
l_{5}\right\rangle ,\left\vert \chi\right\rangle ,\left\vert \chi^{\prime
}\right\rangle ,\left\vert \psi_{2}\right\rangle \}$ are summarized by the
diagram in Fig.~\ref{Fig:Clifton8ray} (where nodes represent rays and the presence of an edge
represents orthogonality). \ In an outcome-deterministic
measurement-noncontextual model, every vector must receive a value 0 or 1 with
exactly one member of every orthogonal triple receiving the value 1, and no
more than one member of an orthogonal pair receiving the value 1.

\begin{figure}[h!]
     \begin{center}
   \includegraphics[width=8cm]{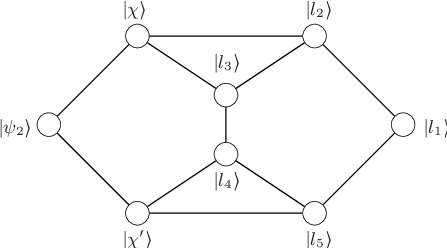}
    \end{center}
    \caption{\label{Fig:Clifton8ray}
    Clifton's 8-ray state-specific Kochen-Specker proof.}
\end{figure}

Clifton's proof can then be phrased as follows.  Given a preparation of $\left\vert \psi_{2}\right\rangle$, the vector $\left\vert \psi_{2}\right\rangle$ (considered as a measurement outcome) must be assigned the value 1 and the vector $\left\vert l_{1}\right\rangle $ has a nonzero probability of being assigned the value 1.
We denote the value assigned to vector $|\phi\rangle$ by $v(|\phi\rangle)$.  From $v\left(  \left\vert \psi_{2}\right\rangle \right)  =1$
we infer $v\left(  \left\vert \chi\right\rangle \right)  =v\left(  \left\vert
\chi^{\prime}\right\rangle \right)  =0$ and from $v\left(  \left\vert
l_{1}\right\rangle \right)  =1$ (which happens with nonzero probability) we infer $v\left(  \left\vert l_{2}%
\right\rangle \right)  =v\left(  \left\vert l_{5}\right\rangle \right)  =0.$
\ One then concludes from $v\left(  \left\vert \chi\right\rangle \right)  =0$
and $v\left(  \left\vert l_{2}\right\rangle \right)  =0$ that $v\left(
\left\vert l_{3}\right\rangle \right)  =1,$ and from $v\left(  \left\vert
\chi^{\prime}\right\rangle \right)  =0$ and $v\left(  \left\vert
l_{5}\right\rangle \right)  =0$ that $v\left(  \left\vert l_{4}\right\rangle
\right)  =1$. \ However, $v\left(  \left\vert l_{3}\right\rangle \right)  =1$
and $v\left(  \left\vert l_{4}\right\rangle \right)  =1$ is a contradiction.

This is the standard way of deriving a contradiction for the eight rays in Clifton's proof, however one could equally well use the fact that $v\left(  \left\vert
\chi\right\rangle \right)  =v\left(  \left\vert \chi^{\prime}\right\rangle
\right)  =0$ and $v\left(  \left\vert l_{1}\right\rangle \right)  =1$ to
justify anticorrelation across every edge around the ring \ $\{\left\vert
l_{1}\right\rangle ,\left\vert l_{2}\right\rangle ,\left\vert l_{3}%
\right\rangle ,\left\vert l_{4}\right\rangle ,\left\vert l_{5}\right\rangle
\}$, which is just the proof we have presented above.

\section{No-go theorems for Bell-local models}\label{sec:separated}

\subsection{A separated pair of single-query 3-box systems} \label{sec:separatedpair3box}

In this section, we consider another variation on Specker's parable.  The seer has a
novel 3-box system which allows only a single box to be opened, rather than
two.  To distinguish the two types of three-box systems, we call the former a
\emph{single-query} system and the latter a \emph{double-query} system. We
also assume that the seer can prepare a pair of single-query systems that
mimic the behavior of the double-query system as follows: if the same\ box
is opened on one system as is opened on the other, one obtains the same
result (both are always found to be full, or both empty); if different boxes
are opened, then one obtains different results (one is always full and the
other empty).  For the benefit of skeptical suitors, the seer allows
for the queries of the two different systems to be implemented at space-like
separation. We imagine that they are transported to different corners of
the Assyrian empire: one to \Alice and the other to Babylon. The suitor
dispatches two of his trusted classmates, one to each of these two cities, and instructs them
to choose a box at random.

We are therefore imagining a situation wherein two observables are measured
jointly by first preparing a pair of systems in a perfectly
correlated state and measuring one observable on each.

As we will demonstrate below, this version of the Specker parable allows us
to establish a simple proof of nonlocality in the same spirit as that
presented by Mermin in Ref.~\cite{Mermin85}. Let us denote the choices made
by the two class-mates by $a$ and $b$ respectively, taking values in the set
$\{1,2,3\}$, corresponding to the choice of box. Further, we denote the
results of box $a$ at \Alice and box $b$ at Babylon, respectively, by $A_a$
and $B_b$,  taking values in $\{1, 0\}$ corresponding to the observations
$\{$full, empty$\}$. Then we can express the condition that the outcomes must
satisfy in this two-wing version of the Specker parable as
\begin{equation}\label{Eq:WinningStrategy}
    A_a\oplus B_b=1\oplus \delta _{a,b},
\end{equation}%
where $\delta $ denotes the Kronecker delta function. To quantify the extent
to which these correlations are realized, let us define $R_3$ as the weighted sum, assuming $a$ and $b$ are chosen uniformly at random, of the probability of
achieving perfect negative correlation when $a\neq b$, and the probability of
achieving perfect correlation when $a=b$.  That is,
\begin{align}
    R_3 \equiv &\frac{1}{9}\sum_{a, b: b\neq a}p(A_a\neq B_b|\M_a, \M_b; \P)
    \nonumber\\
    +&\frac{1}{9}\sum_{a, b : b=a}p(A_a=B_b|\M_a,\M_b; \P),
     \label{eq:winprobOST}
\end{align}
where $p(A_a=B_b|\M_a,\M_b;\P)$ refers to the probability of finding
$A_a=B_b$ conditioned on box $a$ being opened at \Alice and box $b$ being
opened at Babylon; likewise for $p(A_a\neq B_b|\M_a,\M_b;\P)$.

The OS correlations described in the two-wing Specker parable can be summarized as
\begin{align} \label{Eq:SpeckerCorrelations}
    \sum_{A_a,B_b: B_b=A_a} p(A_a,B_b |\M_a,\M_b; \P)=1\quad\text{for}\,a=b,\nonumber\\
    \sum_{A_a,B_b: B_b\neq A_a} p(A_a,B_b |\M_a,\M_b; \P)=1\quad\text{for}\,a\neq b,
\end{align}
The assumption that $\M_a$ and $\M_b$ are jointly measurable in the sense of
definition~(\ref{Dfn:JointMeasurability}) implies that they must satisfy a
condition of no superluminal signaling~\cite{PR94,J.Barrett:PRA:2005}, namely,
\begin{align}
p(A_a|\M_a;\P)=\sum_{B_b}p(A_a,B_b |\M_a,\M_b;\P)\quad\forall\,b,\nonumber\\
\label{eq:non-signaling}
p(B_b|\M_b;\P)=\sum_{A_a}p(A_a,B_b |\M_a,\M_b;\P)\quad\forall\,a
\end{align}
which asserts that the conditional marginal probabilities
$p(A_a|\M_a;\P)$ obtained by summing over $B_b$ are independent
of the choice of the distant measurement procedure $\M_b$, and
likewise for $p(B_b|\M_b;\P)$.  It is simple to show, as we do
in Appendix \ref{App:Correlations:OS2}, that by imposing the
no-signaling condition, the correlations are constrained to be
of the following form:
\begin{align}
    \forall\, a\ne b: &p(0,1|\mathcal{M}_a,\mathcal{M}_{b},
    \mathcal{P_*})=\frac{1}{2},\nonumber \\
    &p(1,0|\mathcal{M}_a,\mathcal{M}_{b};\mathcal{P_*})
    =\frac{1}{2}, \nonumber \\
    \forall\, a= b: &p(0,0|\mathcal{M}_a,\mathcal{M}_{b},
    \mathcal{P_*})=\frac{1}{2},\nonumber \\
    &p(1,1|\mathcal{M}_a,\mathcal{M}_{b};\mathcal{P_*})
    =\frac{1}{2}.
    \label{Eq:SpeckerCorrelations:Explicit}
\end{align}
We will henceforth call these the nonlocal OS correlations. The
winning probability for Eq.~(\ref{eq:winprobOST}) is unity for
these correlations, i.e.,
\begin{equation}
	R^{\textrm{NLOS}}_3=1.
\end{equation}
They are the only non-signaling correlations that can win this
prediction game deterministically. This implies, in particular,
that the nonlocal OS correlations represent an {\em extreme
point} of the convex set of non-signaling
correlations~\cite{N.Jones:PRA:2005}, very much like the
archetypical PR box\footnote{The terminology ``box'' is, in the present
circumstances, unfortunate. It refers to a ``black-box'' (i.e.
unexplained --- indeed, inexplicable --- source of
correlations) between two distant parties, just as in our above
scenario.} correlations~\cite{PR94}
for the scenario where $a,b$ only run from 1 to 2.
 Although these correlations do not allow for
superluminal signaling, they do violate Bell's assumption of
local causality \cite{Belllocalcausality}, as we now
demonstrate.

In order to enforce perfect positive correlations when the suitor's two
classmates make the same measurement, the Babylonian system must be prepared
with an answer for each possible query that matches the answer that the
\Aliceian system is prepared to provide. It follows that there are
deterministic noncontextual hidden variables determining the outcome on the
Babylonian system. This step is familiar from Bell's original derivation of
his theorem~\cite{Bell64}: locality together with the assumption of perfect
correlations implies the existence of deterministic noncontextual values for
each system.
Given such values, it is easy to see \footnote{For instance, if the values assigned to the
state of the boxes 1,2, and 3 are $0,0,1,$ then we have positive correlation for $(a,b)\in\{(1,1),(2,2),(3,3)\}$ and negative correlation for $(a, b)\in \left\{ (1,3),(3,1),(2,3),(3,2)\right\}$, in accordance with the nonlocal OS correlations.  However, we also have positive correlation for $(a, b)\in \left\{ (1,2),(2,1)\right\} $ in disagreement with the nonlocal OS correlations.  Thus the correlations are correct in only 7 out of 9 cases.
Alternative strategies that do not enforce the perfect correlation when $a=b$ cannot do
any better.  For instance, if all boxes on the left are empty and all boxes
on the right are full, then one has 6 out of 6 anti-correlations but 0 out of
3 correlations, leaving a winning probability of 2/3, which is smaller than 7/9.}
that the overall probability of winning the game in a locally causal model is at most
$\tfrac{7}{9}$. That is,
\begin{equation}\label{Ineq:BI:GameTheoretic}
   R_3\le R^{\text{local}}_3= \frac{7}{9}.
\end{equation}%
This is a Bell inequality. The fact that  $R^{\mbox{\tiny
NLOS}}_3=1$ for the seer's system is a violation of this Bell
inequality and a proof that no locally causal model of the nonlocal OS
correlations is possible.

The Bell inequality~\eqref{Ineq:BI:GameTheoretic} can also be written in terms
of the more conventional correlation function, or the so-called two-party
correlator $\langle \bar{A}_a \bar{B}_b \rangle$,  where $\bar{A}_a$, $\bar{B}_b$ take on values
$\{+1,-1\}$ as usual:
\begin{align}
    \langle \bar{A}_a \bar{B}_b \rangle&=\sum_{\bar{A}_a,\bar{B}_b: \bar{B}_b
    =\bar{A}_a} p(\bar{A}_a,\bar{B}_b |\M_a,\M_b;\P)\nonumber\\
    &-\sum_{\bar{A}_a,\bar{B}_b: \bar{B}_b\neq \bar{A}_a} p(\bar{A}_a,\bar{B}_b |\M_a,\M_b;\P)
    \label{Eq:Dfn:Correlator}
\end{align}
The two-party correlator is simply the average value of the product of the
result in \Alice when box $a$ was chosen, multiplied by the result in Babylon
when box $b$ was chosen. Together with the normalization condition
$\sum_{\bar{A}_a,\bar{B}_b} p(\bar{A}_a,\bar{B}_b |\M_a,\M_b, \P)=1$, we can
now re-express the winning probability as
$R_3=\tfrac{1}{18}\mathcal{S}_3+\half$, where
\begin{equation}\label{Eq:BellFunction}
    \mathcal{S}_3 = \sum_{a} \an{ \bar{A}_a \bar{B}_a} -\sum_{a\neq b} \an{\bar{A}_a\bar{B}_b}.
\end{equation}
In these notations, it is again easy to verify that if the
variables $\bar{A}_a$, $\bar{B}_b$ admit pre-existing values $\pm1$ (i.e.
are determined by hidden variables), then $\S_3 \leq 5$ [{\em cf.}
Eq.~\eqref{Ineq:BI:GameTheoretic}]. (As is now
well-known~\cite{Fine82}, the same bound also applies to any
locally causal model where the values of the variables are
determined by stochastic hidden variable models.) Specker's correlations require $\S_3=9$, thus clearly
violating the Bell-inequality.

\subsection{Mermin's proof of Bell's theorem} \label{Sec:Mermin}

What about correlations allowed in quantum theory? We know from
a celebrated theorem by Cleve {\em et~al.} (Theorem 5.12,
Ref.~\cite{Cleve04}) that there is no quantum strategy that can
give unit winning probability. While it is not possible to
realize the over-protective seer parable as formulated above,
it is nevertheless possible to demonstrate, using quantum
mechanics, correlations that approximate the desired
correlations better than any locally causal model can. As it
turns out, the largest winning probability allowed by quantum
theory is  (see Appendix~\ref{App:MaxQV:Bell:n} for details)
\begin{equation}
    R^{\text{quantum}}_3 = \frac{5}{6},
\end{equation}%
and hence $\S^{\text{quantum}}_3 = 6$ (which exceed the Bell-local bounds of
$\frac{7}{9}$ and $5$ respectively). That quantum theory allows such
non-trivial correlations can be verified by considering the two-qubit
maximally entangled state $\tfrac{1}{\sqrt{2}}(\ket{0}\ket{0}+\ket{1}\ket{1})$ (in the
$\s{z}$ basis) and letting $\bar{A}_1$, $\bar{A}_2$, and $\bar{A}_3$ be the
results of measuring the three Pauli operators equally spaced in the
$\hat{z}$-$\hat{x}$ plane, defined by
\begin{equation} \label{eq:3Paulis}
    \hat{A}_a=\cos \frac{2\pi(a-1)}{3} \s{z} + \sin\frac{2\pi(a-1)}{3}\s{x};
\end{equation}
 likewise for the $\hat{B}_b$, which are defined identically.
Thus quantum mechanics allows us to move
towards the extremal non-local correlations in our
formulation of the parable.   This proof that quantum theory violates Bell-locality (Bell's theorem) is in fact the one popularized by David Mermin \cite{Mermin85}.

\subsubsection{Generalization to all odd $n$}

It is straightforward to generalize this new parable to the case of $n$ boxes
for all odd $n\geq 5$. Specifically, posit a separated pair of $n$-box rings
such that if the box that is opened in \Alice is the same as the one opened
in Babylon, the outcomes agree, while if the index $a$ of the box opened in
\Alice differs by 1 from the index $b$ of the one opened in Babylon, that is,
if  $b=a\oplus 1$ or $a=b\oplus 1$,  then the outcomes disagree.  As it turns
out, the correlations must be of the form\footnote{The proof of this proceeds
analogously to the one given in Appendix~\ref{App:Correlations:OS2} for the
specific case of $n=3$.}
\begin{align}\label{Eq:nSpeckerCorrelations:Explicit}
    b= a\oplus 1\textrm{ or }\;a= b\oplus 1:
    &p(1,0 |\M_a,\M_b;\P)=\half \nonumber \\
    &p(0,1 |\M_a,\M_b;\P)=\half \nonumber \\
    b=a:
    &p(0,0 |\M_a,\M_b;\P)=\half \nonumber \\
    &p(1,1 |\M_a,\M_b;\P)=\half .
\end{align}
We do not specify the nature of the correlation for other values of $a$ and $b$.

We can define the average probability of success as
\begin{align}\label{Eq:Dfn:WinningProb:TwoWing}
    R_n \equiv &\frac{1}{3n} \sum_{a,b:b=a\oplus 1 \textrm{ or } a=b\oplus1}
    p(A_a\ne B_b|\M_a,\M_b;\mathcal{P})\nonumber \\
    +&\frac{1}{3n} \sum_{a,b:b=a} p(A_a=B_b|\M_a,\M_b;\mathcal{P})  .
\end{align}
It is evident that with a local strategy if one has perfect
correlation when $a=b$, then when $a=b\oplus 1$ or $b=a\oplus
1$, one can have perfect anti-correlation with probability at
most $(n-1)/n$.  Furthermore, no local strategy can do any
better than this. Consequently, given that the conditions $a=b$, $a=b\oplus 1$, and $b=a\oplus
1$ arise with probability $1/3$ each, the winning probability with a
local strategy is upper bounded by~\footnote{It is worth noting
that none of the following Bell inequalities is facet-inducing
(following the terminology of Ref.~\cite{D.Avis:JPA:2005}), or
tight (following the terminology of
Ref.~\cite{Ll.Masanes:QIC:2003,D.Collins:JPA:2004}). That is,
they do not correspond to the boundary of the set of
locally causal correlations with maximal dimension. }
\begin{equation}
R_n\le R_n^{\textrm{local}}=\frac{1}{3}+\frac{2}{3}\frac{n-1}{n}
= 1 - \frac{2}{3n}.
\end{equation}

Quantum theory can violate this inequality.  Using the same entangled state
as above, we generalize Eq.~(\ref{eq:3Paulis}) to
\begin{eqnarray} \label{eq:spinoperators}
    \hat{A}_a=\cos \varphi_a \s{z} + \sin\varphi_a \s{x}, \nonumber \\
\end{eqnarray}
where $\varphi_a = \tfrac{n-1}{n}\pi (a-1)$, and likewise for
the $\hat{B}_b$ (Fig.~\ref{Fig:GeneralizationOfMermin}). There are $3n$ kinds of measurement statistics
that appear in $R_n$.  We consider each in turn.  For the $n$
terms wherein $a=b$, we obtain perfect correlation with
probability 1, while for the $n$ terms wherein $a=b\oplus 1$
and the $n$ terms wherein $b=a\oplus 1$, we obtain
anti-correlated outcomes with probability $\cos^2(\pi/2n)$.  In
all then, we find the corresponding probability of success as:
\begin{align}
R_n^{\textrm{quantum}}=\frac{1}{3}+\frac{2}{3}\cos^2\frac{\pi}{2n}
\simeq 1-\frac{\pi^2}{6n^2}.
\label{Eq:MaxQuantumRn}
\end{align}
Once again, for a large number of suitors, the seer can choose $n$, the number of measurement settings, to
ensure that with very high probability all of the suitors will lose,
despite their classically founded expectation that one of their number is very likely to win.

\begin{figure}[h!]
    \begin{center}
    \includegraphics[width=7cm]{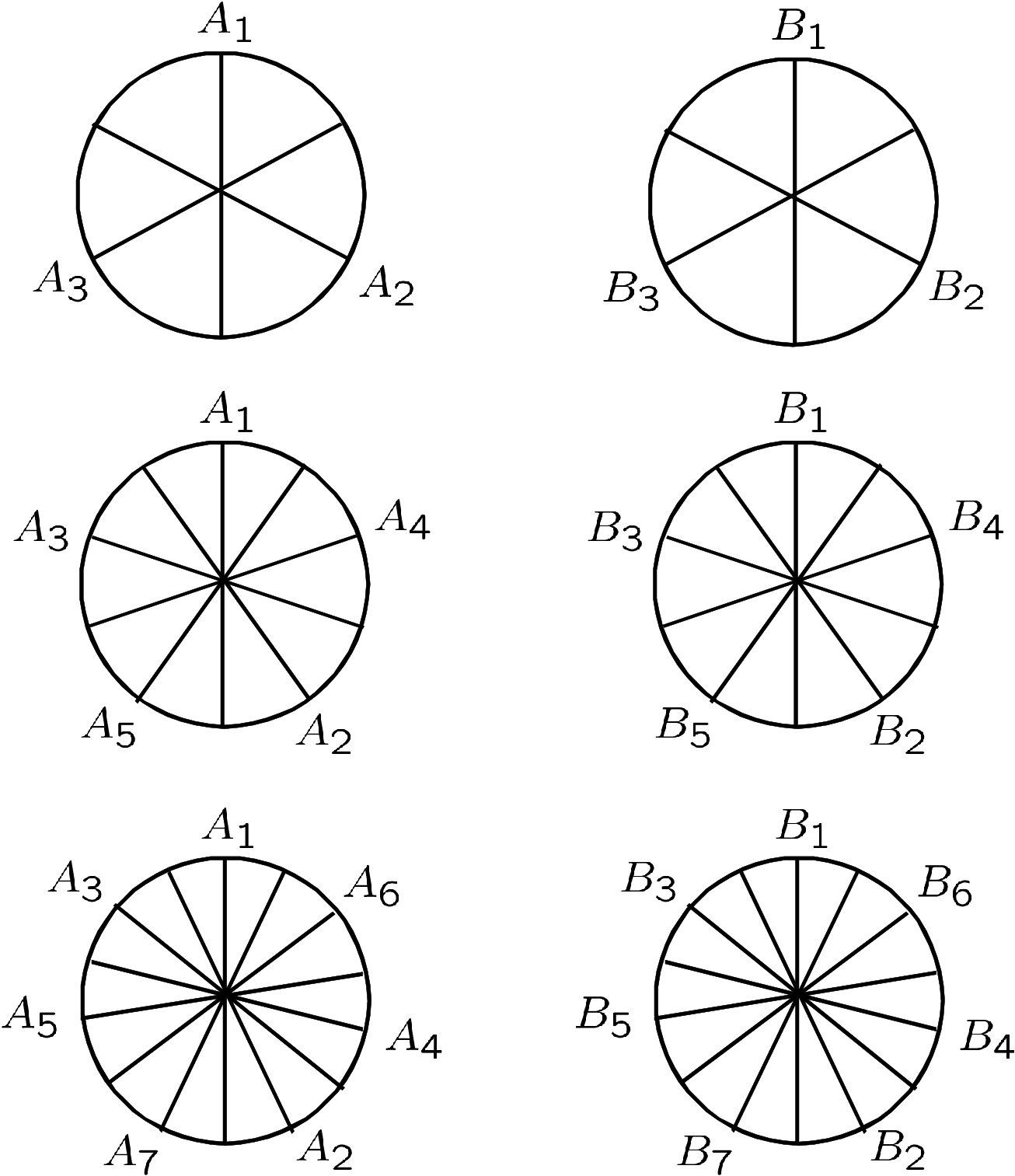}
     \end{center}
    \caption{\label{Fig:GeneralizationOfMermin}
    Representation of the observables used in the Mermin-type
    proofs of nonlocality for $n$=3, 5 and $7$. The observables are
    depicted by lines in a plane of the Bloch sphere.  The vertices
    of each line correspond to the eigenvectors of the associated
    observable, with the labeled vertex associated with eigenvalue
    1.}
\end{figure}

\subsection{Connection to previous work}

An analogous game is discussed by Vaidman~\cite{Vaidman} who
considers a slightly different narrative device: a necklace
having an even number $n$ of beads each of which can be one of
two colors and such that one finds all adjacent beads to be of
different colors except for the first and last beads which are
of the same color.  It is clear that by replacing the first
and last beads by a single bead, we have precisely the
correlations considered above.  Another variation of the game
was considered by Braunstein and Caves \cite{BraunsteinCaves}: there,
perfect correlation is required for all adjacent pairs of measurements except
that between the first and the last, in which case perfect anti-correlation is
 required. This game gives rise to the so-called
 ``chained Bell inequalities"~\cite{BraunsteinCaves}.

The problem of maximizing the winning probability $R_n$ is also relevant to
the strength of a two-prover interactive proof system of the type described
by Cleve~\emph{et al.}~\cite{Cleve04}.   The two provers are taken to be two
agents of the seer (one sent to \Alice and the other to Babylon), while the
suitor is the verifier. The provers' task is to convince the verifier that a
cyclic graph with an odd number $n$ of vertices is 2-colorable (despite the
fact that it is not).  The verifier sends the name of a vertex to each prover
such that the two vertices are either the same or adjacent.  The provers, who
cannot communicate with one another, must each respond with a color. The
existence of systems generating the nonlocal OS correlations would provide
the provers with a perfect winning strategy.

Cleve \emph{et al.} have analyzed a two-player interactive proof, called the
odd cycle game, which is very similar to the one we consider here.  The odd
cycle game is another natural generalization of Specker's parable to a pair
of systems where for a given measurement on the \Aliceian system, there are
two rather than three options for the measurement on the Babylonian system:
it is the same, i.e. $b=a$, or it has index one higher, i.e. $b=a\oplus 1$.
The possibility of $a=b\oplus 1$, which is allowed in the game we have
considered, and ensures symmetry between the two players, is excluded in the
odd cycle game. \footnote{The upper bound on the
winning probability with a local strategy for the odd cycle game is clearly
$R_n^{\textrm{local}} \le \frac{1}{2}+\frac{1}{2}\frac{n-1}{n} = 1 -
\frac{1}{2n}.$ The maximal quantum violation, which is determined in
Ref.~\cite{Cleve04}, is achieved if the measurements on Alice's system are
the spin operators $A_a$ in Eq.~(\ref{eq:spinoperators}), while the
measurements on Bob's system are a rotation by an angle of $\pi/4n$ of the
spin operators $B_b$ in Eq.~(\ref{eq:spinoperators}).  In this case, for the
$n$ terms wherein $a=b$, we have correlation with probability
$\cos^2(\pi/4n)$, and for the $n$ terms wherein $b=a\oplus 1$, we have
anti-correlation with probability $\cos^2(\pi/4n)$, such that
$R_n^{\textrm{quantum}}=\cos^2(\pi/4n) \simeq 1-\frac{\pi^2}{16n^2}.$}

\subsection{From OS correlations to PR-box correlations}
Another way of generalizing the single-query 3-box OS
correlations to a separated pair of parties is to imagine that
each party has a 3-box system, but the first party only ever
opens the first or second box, while the second party only ever
opens the second or third box.  If we imagine that there is
correlation when they both open the second box and
anti-correlation otherwise, then this set of measurements is
already sufficient to obtain a contradiction with a local
model,  Specifically, the local deterministic values must
satisfy
\begin{eqnarray} \label{eq:PRbox}
\bar{A}_2 \bar{B}_2 = +1, \\ \nonumber
\bar{A}_1 \bar{B}_2 = -1, \\ \nonumber
\bar{A}_1 \bar{B}_3 = -1, \\ \nonumber
\bar{A}_2 \bar{B}_3 = -1,
\end{eqnarray}
but the product of the left-hand-sides is $\bar{A}^2_1 \bar{A}^2_2
\bar{B}^2_1 \bar{B}^2_2 = +1$, while the product of the right-hand-sides is
$-1$.  The correlations of Eq.~\eqref{eq:PRbox} are precisely the PR box
correlations \cite{PR94} that have been extensively studied in recent years.

\subsection{Hardy-type no-go theorems for Bell-local models} \label{Sec:Hardy}

In outcome-deterministic ontological models that are local or noncontextual, implications among
value assignments of observables are transitive because these
value assignments do not depend on the context (local or
remote) of the measurement.  The failure of the transitivity of
implication therefore implies the impossibility of such models.  Again, we find that this conclusion has
been reached before in the literature on nonlocality.
Specifically, the Hardy-type proof of
nonlocality~\cite{L.Hardy:PRL:1665} can be
expressed in this fashion~\cite{D.Boschi:PRL:2755}, a fact that
was first noted by Stapp~\cite{H.Stapp:Hardy:Transitivity} (for
a simplified account, see Refs.~\cite{Mermin,Unruh}).

We begin by presenting Hardy's proof of nonlocality in its standard form. \
It uses a pair of binary-outcome observables on each wing of the experiment.
Hardy demonstrated a way of choosing these observables  such that
for any partially entangled pure state,
the correlations between these observables satisfy:%
\begin{eqnarray}
A_{1} &=&1\implies B_{1}=1,  \label{Eq:H1b} \\
B_{2} &=&1\implies A_{2}=1,  \label{Eq:H2b}
\end{eqnarray}%
while%
\begin{equation}
\text{sometimes }\left( A_{1}=1\text{ and }B_{2}=1\right)   \label{Eq:H3b}
\end{equation}%
(i.e. with probability $p_{\mathrm{Hardy}}\equiv p(A_{1}=1$ and $B_{2}=1)>0$%
), and%
\begin{equation}
\text{never }\left( A_{2}=1\text{ and }B_{1}=1\right) .\text{ }
\label{Eq:H4b}
\end{equation}%
\strut

We can express this as a failure of the transitivity of implication as
follows. From Eqs.~(\ref{Eq:H1b}), (\ref{Eq:H4b}) and (\ref{Eq:H2b}) (in its contrapositive form), we infer respectively,
\begin{eqnarray}
A_{1}=1 \implies B_{1}=1,  \label{Eq:HT1} \\
B_{1}=1 \implies A_{2}=0,  \label{Eq:HT2} \\
A_{2}=0 \implies B_{2}=0.  \label{Eq:HT3}
\end{eqnarray}
which we summarize graphically by
\begin{align*}
A_1=1&\Longrightarrow B_1=1\\
& \scalebox{1.05}{ \rotatebox{45}{$\Longleftarrow$}}\\
A_2=0&\Longrightarrow B_2=0
\end{align*}
If transitivity held, then these three inferences would imply that%
\begin{equation}
A_{1}=1 \implies B_{2}=0.  \label{Eq:HT4}
\end{equation}%
However, this contradicts Eq.~(\ref{Eq:H3b}) and consequently transitivity
must fail. \ More explicitly, taking $\implies$ to be material
implication, the negation of Eq.~(\ref{Eq:HT4}) is the conjunction of $%
A_{1}=1$ and $B_{2}=1$,%
\begin{equation}
\lnot \left( A_{1}=1\implies B_{2}=0\right) =\left( A_{1}=1\text{ and }%
B_{2}=1\right) ,
\end{equation}%
so that the probability $p_{\mathrm{Hardy}}\equiv p(A_{1}=1$ and $B_{2}=1)$
quantifies the frequency with which the transitivity of implication fails.

We now consider the status of this sort of proof for the PR box.  By
relabeling the outcomes of the standard PR box, one can obtain correlations
of the form%
\begin{eqnarray}
A_{1} &=&B_{1}  \label{Eq:PR1} \\
A_{1} &=&B_{2}  \label{Eq:PR2} \\
A_{2} &=&B_{1}\oplus 1  \label{Eq:PR3} \\
A_{2} &=&B_{2},  \label{Eq:PR4}
\end{eqnarray}
with marginals of the form $p(A_{1}=0)=p\left( A_{2}=0\right)
=p(B_{1}=0)=p\left( B_{2}=0\right) =1/2.$ \ Eqs.~(\ref{Eq:PR1}), (\ref{Eq:PR3}) and (\ref{Eq:PR4}) imply the inferences of Eqs.~(\ref{Eq:HT1}), (\ref{Eq:HT2}), and (\ref{Eq:HT3}) respectively.  Meanwhile, Eq.~(\ref{Eq:PR2}), together with the fact that $p(A_{1}=1)=1/2,$ implies that
sometimes $A_{1}=1$ and $B_{2}=1$, or equivalently, that sometimes Eq.~(\ref%
{Eq:HT4}) fails, so that we have a contradiction with transitivity. \
Indeed, the probability of this occurring is $p_{\mathrm{Hardy}}=p(A_{1}=1$
and $B_{2}=1)=1/2.$ \

Actually, $p_{\mathrm{Hardy}}$ only quantifies the probability for one
particular kind of contradiction, which requires $A_{1}=1$ to get going. In the rest
of the cases, where $A_{1}=0,$ we still obtain a contradiction because Eqs.~(\ref{Eq:PR1}), (\ref{Eq:PR3}) and (\ref{Eq:PR4}) also imply inferences of
the form of Eqs.~(\ref{Eq:HT1}), (\ref{Eq:HT2}), and (\ref{Eq:HT3}) where $%
A_{a} \Leftrightarrow A_{a}\oplus 1$ and $B_{b}\Leftrightarrow B_{b}\oplus 1.$
Transitivity then implies that $A_{1}=0\implies B_{2}=1,$ while Eq.~(\ref%
{Eq:PR2}) contradicts this. \ So one obtains a contradiction with certainty
for the PR box.

There is another aspect of these PR box implications that cannot be emulated by
quantum theory which has recently been pointed out by Fritz~\cite{Fritz}: if one
 supplements the implications in Eqs.~(\ref{Eq:HT1})-(\ref{Eq:HT3}) with
 the implication $B_2=1\implies A_1  =1$ or any of the two reverse implications, that is,  $A_{2}=0 \implies B_{1}=1$, or $B_{2}=0 \implies A_{2}=0$, then the resulting set of constraints cannot be satisfied by any quantum state and set of projective measurements.

As discussed in the introduction, and rehearsed in Sec.~\ref{sec:variantKlyachko}, Specker
introduced his parable of the overprotective seer in order to demonstrate
the possibility of a logic wherein there is a failure of the transitivity of
implication. \  One therefore expects that the nonlocal OS correlations from Sec.~\ref{sec:separatedpair3box}, which are
based on Specker's parable, ought to provide a proof of nonlocality via such a
failure of transitivity. \ This is indeed the case, as we now
show.
The nonlocal OS correlations, {\em cf.}  Eq.~\eqref{Eq:WinningStrategy},  imply the following chain of implications
\begin{align*}
A_1=1&\Longrightarrow B_1=1\\
& \scalebox{1.05}{ \rotatebox{45}{$\Longleftarrow$}}\\
A_2=0&\Longrightarrow B_2=0\\
& \scalebox{1.05}{ \rotatebox{45}{$\Longleftarrow$}}\\
A_3=1&\Longrightarrow B_3=1
\end{align*}
If the transitivity of implication held, we would have
\begin{equation}
A_{1}=1\implies B_{3}=1.  \label{Eq:kkk}
\end{equation}%
However, Eq. \eqref{Eq:WinningStrategy} together with the fact that $p(A_{1}=1)=1/2$, {\em cf.}  Eq.~\eqref{Eq:SpeckerCorrelations:Explicit},
implies that sometimes $A_{1}=1$ and $B_{3}=0,$ which contradicts Eq.~(\ref{Eq:kkk}). \ Indeed, we achieve this contradiction with probability $p_{\mathrm{Hardy}}=p\left( A_{1}=0\text{ and }B_{3}=1\right) =1/2.$ \ As with
the PR box, one can obtain a contradiction  with certainty  also in the cases where $A_{1}=0.$

Although the nonlocal OS correlations cannot be achieved in quantum theory, it is
interesting to ask whether the particular contradiction constructed above
might be achieved with some nonzero probability for some choice of state and
observables. \ Indeed, this is possible. In particular, this can be achieved with $p^{\mbox{\tiny
Hardy}}= 144/({27+\sqrt{3}})^2\approx 0.17443$ by using the
quantum state
$\ket{\psi}= ({1+\eta^2})^{-1/2}  (\ket{0}\ket{0}-\eta\ket{1}\ket{1})$
and the projectors defined by:
\begin{align*}
    \hat{A}_a=\proj{\uparrow_a^{(A)}},\quad
    \hat{B}_b=\proj{\uparrow_b^{(B)}},
\end{align*}
where
\begin{gather*}
    \ket{\uparrow_a^{(A)}}
    =\left\{
    \begin{array}{l@{\quad \quad}}
    \frac{1}{\sqrt{1+\kappa_1^2}}
    (\kappa_1\ket{0}+\ket{1}): a=1, \\
    \frac{1}{\sqrt{1+\kappa_2^2}}
    (\ket{0}-\kappa_2\ket{1}): a=2,\\
    \frac{1}{\sqrt{1+\kappa_3^2}}
    (\ket{0}+\kappa_3\ket{1}): a=3,
    \end{array}\right. ,\\
    \ket{\uparrow_{b}^{(B)}}
    =\left\{
    \begin{array}{l@{\quad \quad}}
    \frac{1}{\sqrt{1+\kappa_3^2}}
    (-\kappa_3\ket{0}+\ket{1}): b=1,\\
    \frac{1}{\sqrt{1+\kappa_2^2}}
    (\kappa_2\ket{0}+\ket{1}): b=2,\\
    \frac{1}{\sqrt{1+\kappa_1^2}}
    (\ket{0}-\kappa_1\ket{1}): b=3,
    \end{array}\right. ,
\end{gather*}
$\kappa_a=\eta^{(a+1\,\text{mod}\,3)+\frac{1}{2}}$, and $\eta=\sqrt{3}$.

The above Hardy-type proof of nonlocality via the failure of
the transitivity of implications is entirely equivalent to the
proof of nonlocality due to Boschi {\em
et~al.}~\cite{D.Boschi:PRL:2755}. Note that a slightly stronger contradiction with $p^{\mbox{\tiny Hardy}}\approx
0.17455$ can be obtained with a different choice of
$\eta$~\cite{D.Boschi:PRL:2755}. Moreover, this latter
value of $p^{\mbox{\tiny Hardy}}$ is only marginally different
from the quantum-mechanical upper bound $p^{\mbox{\tiny
Hardy}}\le 0.17456$ obtained from the tools of
Ref.~\cite{QMP.Hierarchy}. This suggests that the
strongest contradiction in this scenario may already be
achievable using a two-qubit partially entangled pure state.\footnote{See,
however, Ref.~\cite{K.Pal:1006.3032} for some strong evidence
that in some cases, it may require infinite-dimensional Hilbert space
to achieve the strongest correlations allowed by quantum mechanics
even though the two-qubit correlations are only marginally different from
the quantum mechanical upper bound derived from the tools of
Ref.~\cite{QMP.Hierarchy}.}

It is also worth noting that by considering a similar setup that involves an
increasing number of boxes, and hence a longer chain of {\em intransitive}
implications, quantum theory actually provides a contradiction with {\em
increasing} $p^{\mbox{\tiny Hardy}}$ that asymptotes to
$50\%$~\cite{D.Boschi:PRL:2755}.

We end this section with a demonstration that there is a particular
kind of failure of transitivity that one does not find in quantum theory. \
We begin by noting that with a PR box, we can get a contradiction with the
transitivity of implication in a manner which is different from that of Hardy's
proof, and in some ways more striking. \ In addition to deriving Eqs. (\ref%
{Eq:HT1}), (\ref{Eq:HT2}) and (\ref{Eq:HT3}) from Eqs. (\ref{Eq:PR1}), (\ref%
{Eq:PR3}) and (\ref{Eq:PR4}), we can derive%
\begin{equation} \label{Eq:priortoPRBoxInferences}
B_{2}=0\implies A_{1}=0.
\end{equation}%
Graphically, the chain of inferences is
\begin{align*}
A_1=1&\Longrightarrow B_1=1\\
& \scalebox{1.05}{ \rotatebox{45}{$\Longleftarrow$}}\\
A_2=0&\Longrightarrow B_2=0\\
& \scalebox{1.05}{ \rotatebox{45}{$\Longleftarrow$}}\\
A_1=0&
\end{align*}
Were transitivity of implication to hold, we would conclude that $
A_{1}=1\implies A_{1}=0$, which, together with the fact that $p(A_1=1)=1/2$, yields a contradiction. This sort of proof is
also available for the nonlocal OS correlations.  It can be characterized
as providing a sequence of inferences about values of observables wherein
the consequent of the last inference contradicts the antecedent of the first
inference. \ The question is whether this sort of contradiction can be
achieved in quantum theory. \ As it turns out, it cannot, as we now
demonstrate.

Consider an arbitrary pure bi-partite state $\left\vert \Psi \right\rangle
\in \mathcal{H}^{A}\otimes \mathcal{H}^{B}$ where $\mathcal{H}^{A}$ and $%
\mathcal{H}^{B}$ are Hilbert spaces of dimension $d.$\ \ Defining $\left\{
\left\vert k\right\rangle \right\} _{k=1}^{d}$ to be an orthonormal basis of
$\mathcal{H},$ $\rho $ to be a density operator$,$ $\unit$ to be the identity
operator and $U$ to be a unitary operator, we can always write $\left\vert
\Psi \right\rangle $ in the form%
\begin{equation}
\left\vert \Psi \right\rangle =\left( \unit\otimes U\sqrt{\rho }\right) \sum_{k}%
\left\vert k\right\rangle \otimes \left\vert
k\right\rangle .
\end{equation}

Now suppose that one measures system $A$ with the POVM $\left\{\proj{\phi} ,
\unit-\proj{\phi}\right\} $ and one obtains the $\proj{\phi} $ outcome.
This leads to an updating of the description of the state of system $B$ to
\begin{equation}
\left\vert \chi \right\rangle \equiv \mathcal{N}_{\chi} U\sqrt{\rho }\left\vert \phi ^{\ast
}\right\rangle
\end{equation}
where
\begin{equation}
\left\vert \phi ^{\ast }\right\rangle \equiv \sum_{k}\left\langle k|\phi
\right\rangle ^{\ast }\left\vert k\right\rangle ,
\end{equation}%
and $\mathcal{N}_{\chi}$ is a normalization factor.  Consequently, a subsequent measurement on system $B$ of the POVM $\left\{\proj{\chi} ,\unit-\proj{\chi}\right\} $ will yield the $%
\proj{\chi} $ outcome with
certainty.

Next, consider the experiment wherein  $\left\{\proj{\phi} ,\unit-\proj{\phi}
 \right\} $ is \emph{not }made on $A,$ but the measurement $\left\{\proj{\chi} ,
 \unit-\{\proj{\chi} \right\} $ is made on $B$ and the outcome
  $\proj{\chi} $ is obtained. \ One then updates the
description of the state of system $A$ to%
\begin{eqnarray}
\left\vert \phi ^{\prime }\right\rangle  &\propto & \left\langle \chi
\right\vert \left( \unit\otimes U\sqrt{\rho }\right) \sum_{k}
\left\vert k\right\rangle \otimes \left\vert k\right\rangle  \nn\\
&\propto& \left\langle \phi ^{\ast }\right\vert \left( \unit\otimes \rho \right)
\sum_{k}\left\vert k\right\rangle \otimes \left\vert
k\right\rangle  \nn\\
& \propto&\left\langle \phi ^{\ast }\right\vert \left( \rho ^{T}\otimes \unit\right)
\sum_{k}\left\vert k\right\rangle \otimes \left\vert
k\right\rangle  \nn\\
\ket{ \phi ^{\prime }} &=& \mathcal{N}_{\phi'} \rho ^{T}\left\vert \phi \right\rangle ,
\end{eqnarray}%
where $\mathcal{N}_{\phi'}$ is a normalization factor.
A subsequent measurement on system $A$ of the POVM
 $\left\{\proj{\phi ^{\prime}},\unit-\proj{\phi ^{\prime}} \right\} $
will then yield the $\proj{\phi ^{\prime}} $ outcome with
certainty.

The state $\left\vert \chi \right\rangle $ on $B$ is called the \emph{%
relative state} to $\left\vert \phi \right\rangle $ on $A$ given $\left\vert
\Psi \right\rangle $ on $AB$ \cite{Everett}. \ Similarly, $\left\vert \phi
^{\prime }\right\rangle $ on $A$ is the \emph{relative state} to $\left\vert
\chi \right\rangle $ on $B$ given $\left\vert \Psi \right\rangle $ on $AB.$
\ If we find a particular state on one system, then we are certain to find
the relative state on the other should we measure for it. \ Consequently, we
can consider an arbitrary chain of such pairs of measurements, and at every
step in the chain we can make a perfect inference from the positive outcome
of one to the positive outcome of the other.

We pause at this point in the proof to note that this analysis provides a
particularly simple way of understanding Hardy's proof of nonlocality. \
Using reasoning analogous to that above, the relative state to $\left\vert
\phi ^{\prime }\right\rangle $ is $\left\vert \chi ^{\prime }\right\rangle $
where $\left\vert \chi ^{\prime }\right\rangle \equiv U \rho\, U^\dag \left\vert \chi
\right\rangle $ (note that there clearly exist choices of $\rho$ and $U$ such that $%
\left\vert \chi ^{\prime }\right\rangle \neq \left\vert \chi \right\rangle $)%
$.$ If the transitivity of implication held, then by this sequence of
perfect inferences, we would conclude that whenever $\proj{\phi}$
 is found on $A,$ it would be the case that $\proj{\chi
^{\prime }} $ is necessarily found on $B.$ \ However, this
conclusion is false because the relative state to $\left\vert \phi
\right\rangle $ is $\left\vert \chi \right\rangle ,$ so that the probability
of finding $\proj{\chi^{\prime}}$ on $B$ is $|\left\langle \chi ^{\prime }|\chi
\right\rangle |^{2}$ which is less than one if $\left\vert \chi ^{\prime
}\right\rangle \neq \left\vert \chi \right\rangle .$ \ Thus transitivity
must fail.

We now show that a quantum proof of nonlocality cannot be constructed
in terms of a sequence of
inferences wherein the consequent of the last inference contradicts the
antecedent of the first. We define a set of $N$ observables on $A$%
, each of which is projective, namely, $\left\{\left.\proj{\phi
^{(i)}} \,\right| i=1,\dots
,N\right\} $ and a set of $N$ similar observables on $B$, $\left\{\left.\proj{\chi
^{(i)}} \,\right| i=1,\dots ,N\right\} $. \ Here, $\left\vert \chi ^{(i)}\right\rangle $ is
the relative state to $\left\vert \phi ^{(i)}\right\rangle $ and $\left\vert
\phi ^{(i+1)}\right\rangle $ is the relative state to $\left\vert \chi
^{(i)}\right\rangle .$ \ This implies that we can infer from finding $%
\proj{\phi^{(i)}} $ on $A$ to the necessity of finding $\proj{\chi
^{(i)}} $ on $B,$ and from finding $\proj{\chi
^{(i)}} $ on $B$ to the
necessity of finding $\proj{\phi^{(i+1)}}  $ on $A.$ \ If transitivity of implication held,
then we could chain these inferences together such that from finding $%
\proj{\phi^{(1)}} $ on $A,$ we would infer the necessity of finding
$\proj{\phi^{(N)}}$ on $A.$ \

The question is whether we can ever have such a chain where $\left\vert \phi
^{(N)}\right\rangle $ is orthogonal to $\left\vert \phi ^{(1)}\right\rangle .
$ \ Note from the analysis above that $\left\vert \phi ^{(i+1)}\right\rangle
\propto \rho ^{T}\left\vert \phi ^{(i)}\right\rangle ,$ so that
\begin{equation}
\left\vert \phi ^{(N)}\right\rangle \propto \left( \rho ^{T}\right) ^{N}\left\vert
\phi ^{(1)}\right\rangle .
\end{equation}%
Therefore the condition for orthogonality is
\begin{equation}
|\left\langle \phi ^{(1)}\right\vert \left( \rho ^{T}\right) ^{N}\left\vert
\phi ^{(1)}\right\rangle |^{2}=0.
\end{equation}%
Given the non-negativity of $\rho $ (and hence of $\rho ^{T}$ and $\left(
\rho ^{T}\right) ^{N}),$ this condition is only satisfied if $\rho
^{T}\left\vert \phi ^{(1)}\right\rangle =0,$ but this would imply that the
probability of finding $\proj{\phi^{(1)}} $ on the bipartite state $\left\vert \Psi
\right\rangle $ vanishes. \ In other words, the only bipartite state for
which we can have a chain of inference wherein the final consequent denies
the initial antecedent is one that denies the initial antecedent. \
Therefore, such a contradiction cannot be achieved. \ This is in contrast to
what occurs in the case of PR\ boxes and the nonlocal OS correlations, and is
therefore a feature which distinguishes quantum theory from these foil
theories.  It is interesting to note, however, that in quantum proofs
of contextuality one \emph{can }find a chain of inferences
where the final consequent denies the initial antecedent and the initial antecedent is sometimes true, as shown in Sec.~\ref{sec:variantKlyachko}.

\section{Frustrated networks} \label{sec:frustratednetworks}

It is instructive to consider a network representation of the
various correlations  that we have considered thus far.   The bit associated
with the outcome of a binary-outcome measurement (this is the
only type of measurement we've considered) is associated with a node.  Perfect
positive correlation between outcomes of distinct measurements
is represented by a solid line between the nodes, perfect
negative correlation by a dashed line.   Such representations of correlations have been discussed before in the context of nonlocality proofs, in particular by Mitchell, Popescu and Roberts \cite{MPR04} and in the Ph.D. thesis of Collins \cite{Collinsthesis} and by Schmidt \cite{Schmidt}.

Fig.~\ref{Fig:FN1} provides network representations of the extremal correlations that were used in the no-go theorems for measurement-noncontextual outcome-deterministic models.  The triangular network represents the OS correlations in Specker's parable; the square
network represents the PR-box correlations (understood as a proof of contextuality, i.e. where the four measurements are considered to be implemented in one spatial location); the pentagonal
network represents the extremal version of the correlations in
Klyachko's no-go theorem; the hexagonal network
represents the  kind of  correlations described by Vaidman in
Ref.~\cite{Vaidman}.
\begin{figure}[h!]
     \begin{center}
   \includegraphics[width=8.4cm]{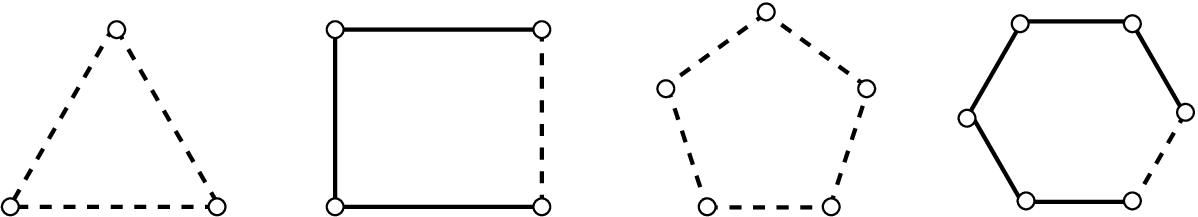}
     \end{center}
    \caption{\label{Fig:FN1}
    Frustrated networks representing the extremal correlations in various proofs of contextuality. }
\end{figure}

Fig.~\ref{Fig:NLFN} provides network representations of the extremal correlations that were used in proofs of nonlocality.  We have labeled the nodes to highlight the spatial region in which each of the outcomes occurs.
The network on the left, which is graph-isomorphic to the square network above, represents the correlations generated by a PR-box~\cite{PR94}.  The network on the right depicts the correlations found in the separated pair of
single-query 3-box systems of Sec.~\ref{sec:separated}.
\begin{figure}[h!]
     \begin{center}
   \includegraphics[width=6.2cm]{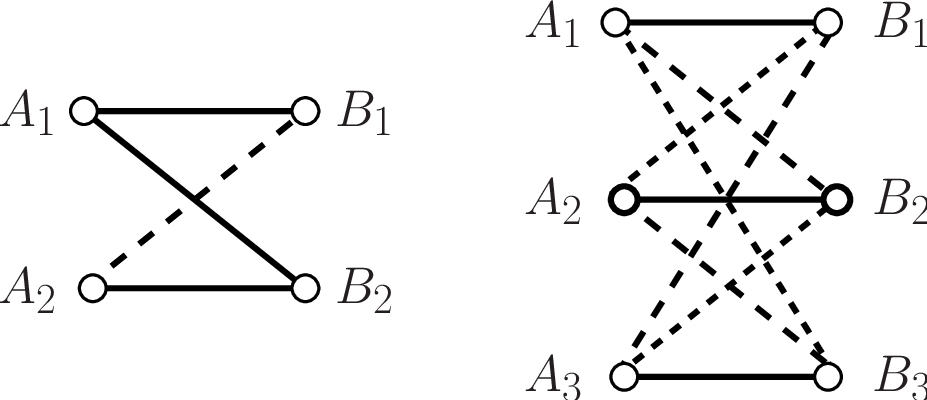}
     \end{center}
    \caption{\label{Fig:NLFN}
Frustrated networks representing the extremal correlations in various proofs of nonlocality.}
\end{figure}

Let the bit describing whether there is an even or an odd
number of dashed lines along a path be called the \emph{parity}
of the path.  We shall say that a network is \emph{frustrated}
if for any pair of nodes, there exist paths with different
parities connecting those nodes.  Clearly, each of the networks in Fig.~\ref{Fig:FN1} is frustrated.  It is this frustration which captures
the impossibility of an outcome-deterministic measurement-noncontextual
model of these correlations. For the networks given in Fig.~\ref{Fig:NLFN},
this impossibility also gives rise to a simple proof of nonlocality
of the depicted correlations.

For any network, we can determine whether or not it is frustrated by looking only at its cycles.  This is because frustration occurs when there are two paths with differing parities and this fact will reveal itself by examining the cycle consisting of that pair of paths.  Thus, to see the ways in which a network can be frustrated, it suffices to consider the ways in which cycles can be frustrated.  For any integer number of nodes, it is straightforward to find all the frustrated cycles with that number of nodes.  For two nodes, there is only a single path and therefore no possibility for frustration.  At 3 nodes, the frustrated networks are those indicated in Fig.~\ref{Fig:FN2}.  The case of two correlations and one anti-correlation corresponds, in the imagery of Specker's parable, to a case where if boxes 1 and 2 or boxes 1 and 3 are opened, one finds the same outcome, but if boxes 2 and 3 are opened, the outcomes always differ.
Note, however, that these different networks are equivalent up to a relabeling of the outcomes and consequently represent essentially the same correlations.
\begin{figure}[h!]
     \begin{center}
   \includegraphics[width=7.4cm]{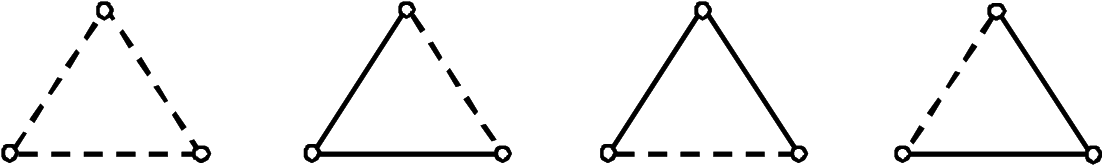}
     \end{center}
    \caption{\label{Fig:FN2}
    All the ways in which a triangular network can be frustrated.}
\end{figure}

Indeed, all the frustrated networks with a given number of nodes can be
obtained one from another by a relabeling of the outcomes.  It therefore
suffices to consider a single representative of the equivalence class of
frustrated networks with a given number of nodes.

It is also possible to have a similar graphical representation for some of the no-go theorems for noncontextuality and locality that are based on a failure of transitivity of implication.  We represent a set of implications among the values of binary-outcome observables by a \emph{directed} graph with decorated edges.  The implications of interest are of the form: $X_1=x \implies X_2=y$ where $x,y\in\{0,1\}$ and either $y=x$ or $y=x \oplus 1$.  We depict this by inserting a directed edge (i.e. an arrow) from the node for $X_1$ to the node for $X_2$ and decorating the base of the arrow with the value $x$; the directed edge is solid if $y=x$ and dashed if $y=x \oplus 1$.  Note that this implication can also be written in its contrapositive form as $X_2=y\oplus 1 \implies X_1=x\oplus 1$.  Therefore, we can always represent the same implication with an arrow in the opposite direction.  When reversing an arrow, the value decorating the arrow stays the same if the arrow is solid and flips if the arrow is dashed.

If the parity is odd around a closed loop in such a directed graph, then the antecedent of the first implication is denied by the consequent of the last implication.  Therefore, as long as the antecedent has non-zero probability, we have a failure of the transitivity of implication.  Such a directed network is said to be frustrated.

In the introduction, we described how Specker's parable implies a failure of the transitivity of implication (under the assumption that value-assignments to observables are context-independent).  Letting $s_i$ denote the proposition that $X_i=1$ (box $i$ contains a gem), the set of implications are: $s_1 \implies \lnot s_2$, $\lnot s_2 \implies s_3$, and $s_3 \implies \lnot s_1$.
These are represented by the directed network of Fig.~\ref{Fig:FDNtp}(a), which is clearly frustrated.
\begin{figure}[h!]
     \begin{center}
   \includegraphics[width=5.2cm]{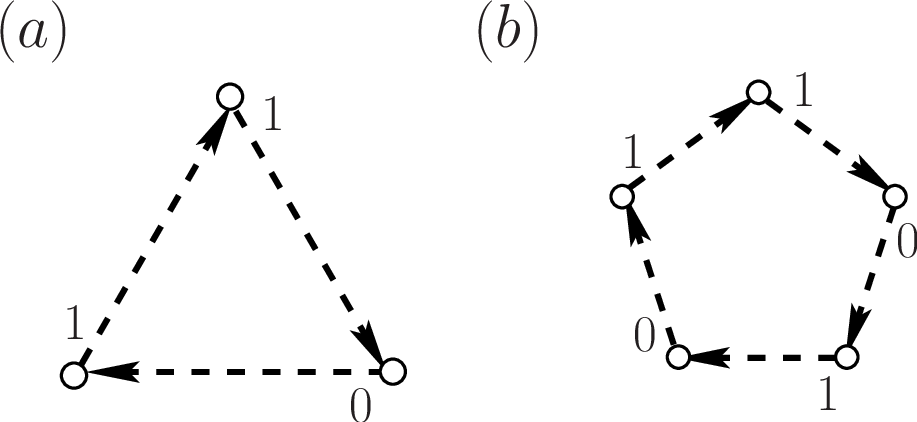}
     \end{center}
    \caption{\label{Fig:FDNtp}
    Some frustrated directed networks corresponding to contextual correlations.}
\end{figure}
The set of implications that are used in the transitivity-based no-go theorem of Sec.~\ref{sec:variantKlyachko} are represented by the pentagonal version of this directed network, Fig.~\ref{Fig:FDNtp}(b), which is also frustrated.

Unlike the undirected frustrated networks, which are composed of a set of correlations some or all of which are only \emph{approximated} by the quantum correlations, the \emph{directed} frustrated network of Fig.~\ref{Fig:FDNtp}(b) is an exact specification of implications one finds in quantum theory, specifically, those described in the proof of Sec.~\ref{sec:variantKlyachko}.  The only sense in which one could imagine a theory being ``more contextual'', according to this sort of proof, is by assigning a higher probability to the contradiction-generating valuation of the first observable in the chain.  An \emph{extremal} version of such a proof would be one wherein both possible valuations of the first observable yielded a contradiction.   We conjecture that such a proof cannot be found in quantum theory.

Finally, the ``striking'' form of the PR box correlations, presented in Sec.~\ref{Sec:Hardy} and associated with the set of implications below Eq.~(\ref{Eq:priortoPRBoxInferences}), is represented by the frustrated directed network in Fig.~\ref{Fig:FDNPRboxNLOS}(a), and the generalization of this to the case of the nonlocal OS correlations is represented in Fig.~\ref{Fig:FDNPRboxNLOS}(b).
\begin{figure}[h!]
     \begin{center}
   \includegraphics[width=6.2cm]{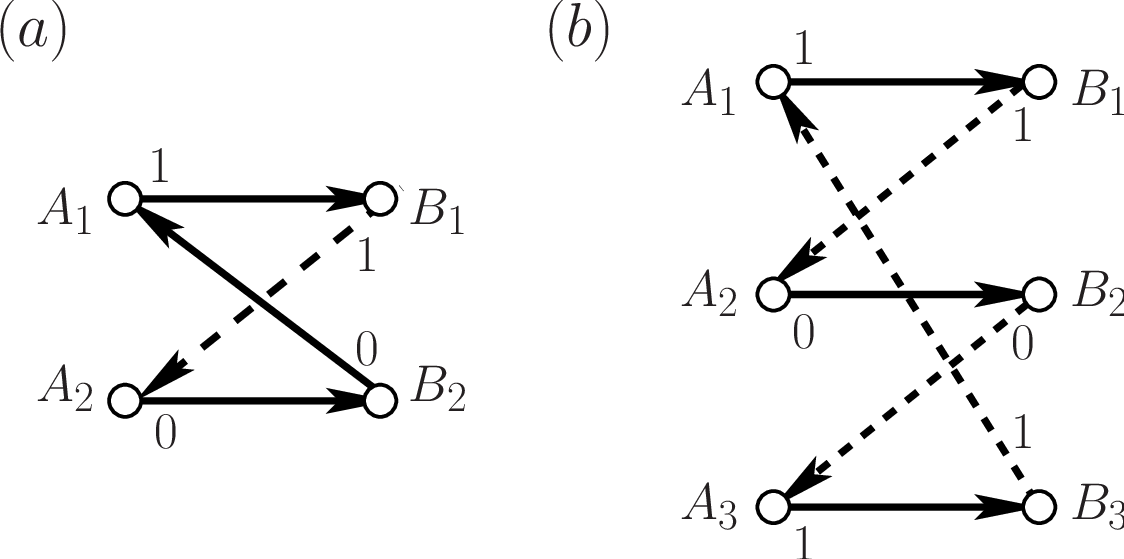}
     \end{center}
    \caption{\label{Fig:FDNPRboxNLOS}
    Some frustrated directed networks corresponding to Bell-nonlocal correlations.}
\end{figure}
As was shown at the end of Sec.~\ref{Sec:Hardy}, it is not possible to find a quantum state and a set of observables that instantiates such a set of implications while assigning a nonzero probability to the contradiction-generating valuation of the first observable.

\section{No-go theorems for preparation-noncontextual models} \label{quantumprepcontextuality}

So far, in all of our quantum analogues of Specker's parable, the
correlations examined were between the outcomes of pairs of measurements that could
be implemented \emph{jointly}. In this section, we consider the possibility
of achieving these correlations between the outcomes of pairs of
measurements that are implemented \emph{consecutively}\footnote{Because implementing the first measurement and selecting a particular outcome constitutes a preparation, one can equally well describe
this section as a consideration of the possibility of achieving analogues of
the OS correlations between preparations and measurements. This is discussed further below.}.

It is important to recognize that one need not rule out the possibility of
consecutive measurements to ensure the impossibility of joint measurements.
The original version of the Specker parable is misleading in this respect.
It asks us to imagine that after opening two boxes, one is simply unable to
open the third (as if its lid were glued shut with an unbreakable seal). The
literal generalization to arbitrary measurements $\mathcal{M}_{1},$
$\mathcal{M}_{2}$ and $\mathcal{M}_{3}$ that can be measured jointly pairwise
but not triplewise would seem to be that if $\mathcal{M}_{1}$ and
$\mathcal{M}_{2}$ have been implemented, a mysterious force prevents us from
carrying out the instructions that correspond to implementing
$\mathcal{M}_{3}.$ However, this conclusion does not follow from a denial of
joint measurability as it is defined in Sec.~\ref{sec:JM}.  One can always
implement $\mathcal{M}_{3}$ following a measurement of $\mathcal{M}_{1}$ and
$\mathcal{M}_{2}$ on a preparation $\mathcal{P}.$  It is just that the
statistics of outcomes of $\mathcal{M}_{3}$ that one thereby obtains is not
the same as one would have obtained if $\mathcal{M}_{3}$ were implemented on
$\mathcal{P}$ directly. To be precise, if the joint statistics of outcomes of
a pair of measurements $\mathcal{M}$ and $\mathcal{M}^{\prime }$ are
independent of the order in which they are implemented, then the consecutive
implementation of the two measurements constitutes a joint measurement of
$\mathcal{M}$ and $\mathcal{M}^{\prime }$ by the definition of
Sec.~\ref{sec:JM}.  Consequently, a denial of joint measurability implies a
denial of the invariance of statistics under a reordering of the
measurements.
This way of interpreting a lack of joint measurability is precisely the one
that is familiar from the quantum  theory of projective measurements.

To see how the OS correlations might obtain for consecutive measurements, we present a new parable.
We consider a single-query 3-box system, that is, one where only a single box can be opened at a time.  A pair of boxes can be opened consecutively,
but the second box-opening need not reproduce the statistics of outcomes that would have been observed had it been opened first.  In this sense, the measurements associated with opening distinct boxes cannot be implemented jointly.

We now get to the specifics of the correlations, which are inspired by the
original Specker parable. We assume that there is a special preparation $%
\mathcal{P}_{\ast }$ of the 3-box system, such that if the \emph{same}
box is opened at the two times, then the same outcome is found, while if
\emph{different} boxes are opened at the two times, then different outcomes
are found.

So far, there is nothing in this set of correlations that prohibits their
being explained by a generalized-noncontextual ontological model. Because no two
measurements are ever implemented jointly in this parable, there is no sense
in which any measurement has a nontrivial context upon which its ontological
representation might depend. Indeed, there are ontological models that
explain the correlations easily. They need only posit that the first
measurement disturbs the ontic state of the three-box system in order to
enforce the appropriate correlations. For instance, suppose that three bits
specify the gem occupation numbers of the three boxes and completely
characterize the ontic state. It could be that finding a $0$ ($1)$ for a box
forces the other two boxes to have occupation number $1\left(0\right)$.
(Indeed, if the suitor is opening boxes on a table, this kind
of disturbance to the ontic state might be enforced by having a hidden
mechanism under the table that automatically
inserts or removes gems from the two boxes
that were not opened.)

To obtain a set of correlations that can challenge the assumption of
generalized-noncontextuality, we need to modify the thought experiment slightly by
adding the following assumption: \ in addition to the correlations
described, it is the case that after the early measurement is complete, for
every possible subsequent measurement (the theory may well allow more than
the three measurements that are used in the protocol), it is impossible to
obtain any information about the identity of the early measurement. \ We
call this the \emph{trit-obliviousness} condition (this terminology has its
precedent in Ref.~\cite{SBKPT09}).

Note that implementing the early measurement procedure and selecting
a particular outcome constitutes a preparation. For each of the three possible
measurement procedures, $M_{1},$ $M_{2}$ and $M_{3},$ and each of the
outcomes $0$ and $1$, we obtain a distinct preparation procedure. We
denote these by $P_{1,0},$ $P_{1,1},$ $P_{2,0},$ $P_{2,1},$ $P_{3,0}$ and $%
P_{3,1}$ in an obvious notation. \footnote{Note that measurement procedures $%
M_{1}$ and $M_{1}^{\prime }$ that are in the same operational equivalence
class $\mathcal{M}_{1}$, may nonetheless define preparation procedures that fail to
be operationally equivalent, because operational equivalence for preparation procedures is decided by the
statistics of all possible subsequent measurements.} We can
also define the preparations that result when one chooses not to condition
on the outcome of the measurement procedure.  We denote these by $
P_{1},P_{2}$ and $P_{3}$. Finally, we denote the probability of obtaining
outcome $0$ when the first measurement $M_{t}$ is implemented on the special
preparation $\mathcal{P}_{\ast }$ by $w_{t,0}\equiv p(0|M_{t};\mathcal{P}%
_{\ast }),$ and we define $w_{t,1}\equiv 1-w_{t,0}.$ The statistics for the
unconditional preparations are then given by%
\begin{equation}
    p(X|M;P_{t})=w_{t,0}p\left( X|M;P_{t,0}\right)
    +w_{t,1}p\left(X|M;P_{t,1}\right) .  \label{eq:cs}
\end{equation}%
The trit-obliviousness condition states that preparation procedures $%
P_{1},P_{2}$ and $P_{3}$ are operationally equivalent, that is,
\begin{equation}
\forall M:p(X|M;P_{1})=p(X|M;P_{2})=p(X|M;P_{3}).  \label{eq:oep}
\end{equation}%
\qquad We now see how the new parable might have troubles with
generalized-noncontextuality. In Sec.~\ref{sec:PNC}, we defined a preparation-noncontextual ontological model to be one wherein operational equivalence of
preparation procedures implies that they are represented by the same
distributions in the ontological model, {\em cf.} Eq.~\eqref{eq:PNC}. Thus, from
Eq.~(\ref{eq:oep}) and preparation noncontextuality, we infer that
\begin{equation}
p(\lambda |P_{1})=p(\lambda |P_{2})=p(\lambda |P_{3}).
\label{eq:consequenceofoep}
\end{equation}%
Given that convex combinations of preparation procedures are represented in
an ontological model by convex combinations of the associated distributions
(see Ref.~\cite{Spe05}), we infer from Eq.~(\ref{eq:consequenceofoep}) that
\begin{eqnarray}
&&w_{1,0}p\left( \lambda |P_{1,0}\right) +w_{1,1}p\left( \lambda
|P_{1,1}\right)  \notag\\
&=&w_{2,0}p\left( \lambda |P_{2,0}\right) +w_{2,1}p\left( \lambda
|P_{2,1}\right)  \notag\\
&=&w_{3,0}p\left( \lambda |P_{3,0}\right) +w_{3,1}p\left( \lambda
|P_{3,1}\right) .\label{eq:PNC2}
\end{eqnarray}
The ``preparation context'' is the specification of which of the three mixtures of preparation procedures was implemented, and the assumption of preparation noncontextuality is that the distribution over $\lambda$ does not depend on this context.

Equation~(\ref{eq:PNC2}) is a nontrivial constraint which is not
necessarily consistent with the posited correlations between the preparation
procedures $P_{1,0},$ $P_{1,1},$ $P_{2,0},$ $P_{2,1},$ $P_{3,0}$ and $P_{3,1}$ and the outcomes of the subsequent measurements of $M_{1},$ $M_{2}$ or $M_{3}.$  Indeed, in the ontological model we proposed above, where the ontic state is a triple of bits specifying the occupation numbers of each box, the
distributions corresponding to the six preparation procedures are:
\begin{eqnarray}
p\left( \lambda |P_{1,0}\right)  &=&\delta _{\lambda ,(0,1,1)} \\
p\left( \lambda |P_{1,1}\right)  &=&\delta _{\lambda ,(1,0,0)} \\
p\left( \lambda |P_{2,0}\right)  &=&\delta _{\lambda ,(1,0,1)} \\
p\left( \lambda |P_{2,1}\right)  &=&\delta _{\lambda ,(0,1,0)} \\
p\left( \lambda |P_{3,0}\right)  &=&\delta _{\lambda ,(1,1,0)} \\
p\left( \lambda |P_{3,1}\right)  &=&\delta _{\lambda ,(0,0,1)},
\end{eqnarray}%
where $\delta $ denotes the Kronecker delta function. Supposing that $w_{1,0}=w_{2,0}=w_{3,0}=1/2,$ we have
\begin{eqnarray}
p\left( \lambda |P_{1}\right)  &=&\frac{1}{2}\delta _{\lambda ,(0,1,1)}+%
\frac{1}{2}\delta _{\lambda ,(1,0,0)}, \\
p\left( \lambda |P_{2}\right)  &=&\frac{1}{2}\delta _{\lambda ,(1,0,1)}+%
\frac{1}{2}\delta _{\lambda ,(0,1,0)}, \\
p\left( \lambda |P_{3}\right)  &=&\frac{1}{2}\delta _{\lambda ,(1,1,0)}+%
\frac{1}{2}\delta _{\lambda ,(0,0,1)},
\end{eqnarray}%
and therefore
\begin{equation}
p\left( \lambda |P_{1}\right) \neq p\left( \lambda |P_{2}\right) \neq
p\left( \lambda |P_{3}\right) .
\end{equation}
So we find that the distributions representing $P_{1},$ $P_{2}$ and $P_{3}$
in the ontological model are distinct even though these preparation
procedures are operationally equivalent\footnote{Remember that in the
ontological model we are considering, the
measurement $M_{y}$ simply reveals the value of the $y$th bit, that is, $%
p(X=\lambda _{y}|M_{y};(\lambda _{1},\lambda _{2},\lambda _{3}))=1.$ It
follows that $\sum_{\lambda }$ $p(X=0|M_{y};\lambda )p(\lambda |P_{t})=\frac{%
1}{2}$ for all $y,t\in \left\{ 1,2,3\right\} ,$ and consequently, the
two outcomes of $M_{y}$ occur with equal probability given a preparation
procedure $P_{t}.$ Therefore, the ontological model captures the fact that $P_{1},$ $P_{2}$ and $P_{3}$ are operationally indistinguishable.}
--- a violation of the assumption of
preparation noncontextuality.

We have demonstrated that the simple ontological model suggested earlier to
explain the two-time OS correlations cannot also satisfy the condition of trit-obliviousness while preserving preparation noncontextuality.
In the next subsection, we will show that \emph{no} ontological model that
can explain the OS correlations and the trit-oblivious condition
can be preparation-noncontextual.   In
this sense, a suitor who is committed to generalized noncontextuality should be
surprised if he sees the specified two-time correlations after having
confirmed the trit-obliviousness condition.

It is useful to summarize the correlations that we have described above.

\subsection{Diachronic pair of single-query 3-box OS correlations}

There are six possible preparation procedures, denoted $P_{t,b}$ where $t\in \left\{
1,2,3\right\} $ ($t$ for trit) and $b\in \{0,1\},$ and three possible
measurement procedures, denoted $M_{y}$ where $y\in \left\{ 1,2,3\right\} .$
 For simplicity, we assume that the prior over each of $t,b$ and $y$ to be uniform.
The outcome $X$ of the measurement procedure $M_{y}$ given a preparation
procedure $P_{t,b}$ is the following function of $t,$ $b$ and $y,$
\begin{equation}\label{Dfn:c_y(t,b)}
c_y(t,b)\equiv
\begin{array}{c}
b\text{ if }t=y \\
b\oplus 1\text{ if }t\neq y
\end{array},
\end{equation}
that is, the correlations are such that
\begin{eqnarray}
p(X=c_y(t,b)|M_y;P_{t,b})=1.
\end{eqnarray}
Finally, defining the effective preparation procedure $P_t$ as the mixture of $P_{t,0}$ and $P_{t,1}$,
it is assumed that no measurement can reveal any information about which of $P_1$, $P_2$ or $P_3$ was implemented,
\begin{equation} \label{eq:tritoblivious}
\forall M: p(X|M;P_1)=p(X|M;P_2)=p(X|M;P_3).
\end{equation}
This is the \emph{trit-obliviousness condition}. \

Defining the average probability of success as
\begin{eqnarray}
R_3 \equiv \frac{1}{18}
\sum_{t,b,y}
p(X=c_y(t,b)|M_{y};P_{t,b}),
\end{eqnarray}
we can also characterize the two-time OS correlations as those achieving
$R_3=1$. Using the trit-obliviousness condition, we shall see that the
assumption of preparation noncontextuality places a bound on the average
probability of success, namely,
\begin{equation}\label{eq:winpropPNC}
R_3\le R^{\text{PNC}}_3 = \frac{7}{9}.
\end{equation}
We refer to this bound as a \emph{noncontextuality inequality}.

The proof is as follows. For any measurement $M,$ the probability of outcome
$X$ given preparation $P_{t}$ is simply
\begin{equation}
p(X|M;P_{t})=\frac{1}{2}\sum_{b\in \{0,1\}}p(X|M;P_{t,b}).
\end{equation}%
Similarly, the probability of the ontic state $\lambda $ given an
implementation of $P_{t}$ is simply
\begin{equation}
p(\lambda |P_{t})=\frac{1}{2}\sum_{b\in \{0,1\}}p(\lambda |P_{t,b}).
\end{equation}%

Given the trit-obliviousness condition, Eq.~(\ref{eq:tritoblivious}), and the assumption of preparation noncontextuality, Eq.~(\ref{eq:PNC}), we infer that $p(\lambda |P_{0})=p(\lambda |P_{1})=p(\lambda
|P_{2}),$ which states that mixed preparations corresponding to different
values of the trit $t$ are not only indistinguishable at the operational level, but at the ontic level as well.  Therefore, even if one knew $\lambda $, the posterior probabilities for $t=1,
$ $t=2$ and $t=3$ would be the same, that is, one would know nothing about
the trit $t.$  The argument so far can be summarized as follows: for
preparation-noncontextual models, trit-obliviousness at the operational
level implies trit-obliviousness at the ontic level. The ontic state $
\lambda $ provides a classical encoding of $(t,b)$, but one that does not contain any information about $t$.

To finish the argument, we take note of all the functions of
$t$ and $b$ that contain no information about $t$.\footnote{In
the sense that for any given value of the function $f(t,b)$,
the conditional probability $p(t|f(t,b))=1/3$ for all $t$.} \
These are equivalent, up to an affine transformation (i.e. up
to a scalar multiple and an additive
constant), to one of the following four functions%
\begin{equation}
\begin{array}{c||c|c|c|c}
t,b & b & c_1(t,b) & c_2(t,b) & c_3(t,b)
\\
\hline
1,0 & 0 & 0 & 1 & 1 \\
1,1 & 1 & 1 & 0 & 0 \\
2,0 & 0 & 1 & 0 & 1 \\
2,1 & 1 & 0 & 1 & 0 \\
3,0 & 0 & 1 & 1 & 0 \\
3,1 & 1 & 0 & 0 & 1%
\end{array}%
\end{equation}
 where $c_y(t,b)$ is defined in Eq.~\eqref{Dfn:c_y(t,b)}.
In an ontological model that respects preparation noncontextuality
and the trit-obliviousness condition,  the ontic
state must be given by one of these four functions, that is, $p\left( \lambda
|P_{t,b}\right) =\delta _{\lambda ,b}$ or $\delta _{\lambda ,c_1(t,b)}$ or
$\delta _{\lambda ,c_2(t,b)}$ or $\delta _{\lambda ,c_3(t,b)}.$ Note that in
each case, the ontic state space is a single bit \footnote{It is not a triple
of bits and therefore cannot specify the occupation numbers of each of the
three boxes. In this sense, the narrative device of a three-box system
cannot do justice to this version of the parable.
We must think about the preparations and measurements more abstractly.}.

In the case of an ontological model wherein $\lambda =b,$ the best the
measurement device can do is to always output $b\oplus 1$ because with probability 2/3, $y\neq t$ and $c_y(t,b)=b\oplus 1$, while with probability 1/3, $y= t$ and $c_y(t,b)=b$.  Thus, for this ontological model, the average success
probability is $2/3.$

In the case of an ontological model wherein $\lambda =c_1(t,b),$ the best the measurement device can do is to output $c_1(t,b)$ when $y=1$ and $c_1(t,b)\oplus 1$ when $y\neq 1$. \
Note that $c_1(t,b)\oplus 1=c_2(t,b)$ for $2/3$ of
the values of $t,b$ and $c_1(t,b)\oplus 1=c_3(t,b)$
also for $2/3$ of the values of $t,b.$ (To see this, it suffices to take
the negation of the $c_1(t,b)$ column of the table and compare it with the $c_2(t,b)$ and $c_3(t,b)$
columns.) So we see that this choice of
output generates the right correlations $2/3$ of the time for $y\neq 1$. Thus for this
ontological model, the overall success probability is $7/9$.

By symmetry, the cases of $\lambda = c_2(t,b)$ and $\lambda = c_3(t,b)$ also achieve a success probability of at
most $7/9.$ Therefore, the probability of success in a preparation-noncontextual ontological model is bounded above by $7/9.$

\subsection{Quantum case} We now consider to what extent one can achieve the diachronic OS
correlations in quantum theory.  The following is a protocol that uses a
single qubit.  The three measurements correspond to the three Pauli operators
$\hat{A}_t$ of Eq.~(\ref{eq:3Paulis}) corresponding to directions equally
spaced in an equatorial plane of the Bloch sphere.  The positive and negative
eigenvalues are mapped onto outputs $X=0$ and $X=1$ respectively.  The
preparation procedures $P_{t,0}$ and $P_{t,1}$ correspond to the two
eigenstates of $\hat{A}_t$, with positive and negative eigenvalues mapped
onto $b=0$ and $b=1$ respectively.  We denote these states by the Hilbert
space vectors $|\phi_{t,b}\rangle$.  The Bloch sphere representation of these
states and measurements is provided in Fig.~\ref{Fig:PC}.  When $y=t$, the
preparation corresponds to an eigenstate of the observable being measured,
and the outcome $X$ equals the bit $b$.  Thus, $X=c_y(t,b)$ with probability
1 in this case.  When $y\neq t$, the probability of obtaining $X=b$ is
$|\langle \phi_{t,b} | \phi_{y,b} \rangle|^2 = \cos^2(\pi/3)=1/4$ while the
probability of obtaining $X=b\oplus 1$ and thus $X=c_y(t,b)$ is $3/4$.  We
have $y\neq t$ in 2/3 of cases, so that the overall probability of success is
\begin{equation}
R^{\textrm{quantum}}_3=\frac{5}{6}.
\end{equation}
Meanwhile, no information about $t$ can be obtained by any quantum
measurement given that the mixtures associated with different values of $t$
are represented by the same density operator: $\frac{1}{2}\proj{\phi_{0,0}} +\frac{1
}{2}\proj{\phi_{0,1}}$= $\frac{1}{2}\proj{\phi_{1,0}} +\frac{1}{2}\proj{\phi_{1,1}}$
= $\frac{1}{2}\proj{\phi_{2,0}} +\frac{1}{2}\proj{\phi_{2,1}} = \unit/2.$
Thus we have a violation of the noncontextuality inequality of Eq.~(\ref{eq:winpropPNC}).
\begin{figure}[h!]
     \begin{center}
   \includegraphics[width=7cm]{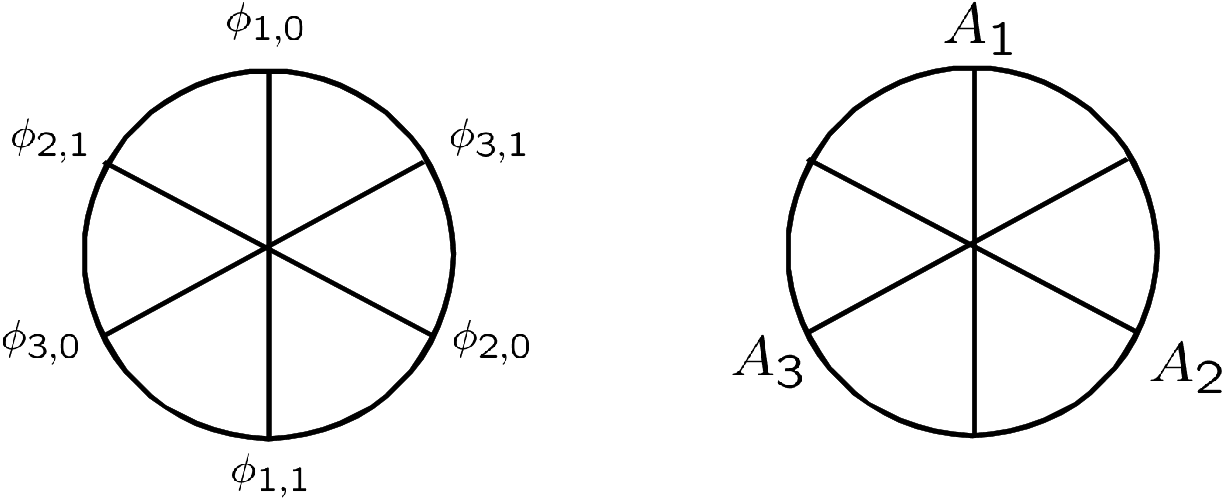}
     \end{center}
    \caption{\label{Fig:PC}
    Quantum states and observables used for proof of the impossibility of a preparation-noncontextual ontological model.
    }
\end{figure}

Note that the OS correlations are useful for achieving the following
two-party secure computation, which is a kind of multiplexing. Let the two
parties be called Alice and Bob.  Alice has as input a trit $t\in \{1,2,3\}$
and a bit $b\in \{0,1\}$, each chosen uniformly at random.   Bob has as input
a trit $y\in \{1,2,3\}$ chosen uniformly at random$.$ Bob outputs a bit $c$
and the goal of the task is for Bob to output $c=c_y(t,b),$ that is, Bob
should output $b$ if $y=t$ and the negation of $b$ otherwise. Alice can send
a system to Bob encoding information about her input, however there is a
cryptographic constraint: no information about the trit $t$ can be
transmitted to Bob, which is to say that the protocol must be trit-oblivious.
This information-theoretic manner of characterizing the correlations provides
a connection with the discussion of preparation noncontextuality found in
Ref.~\cite{SBKPT09}.

\subsection{Justifying preparation noncontextuality by locality}
As discussed in Ref.~\cite{Spe05}, it is sometimes possible to justify an
assumption of preparation noncontextuality using Bell's assumption of local
causality \cite{Belllocalcausality}.  This is the case for the assumptions of
preparation noncontextuality that appear in the derivation of the
noncontextuality inequality of Eq.~(\ref{eq:winpropPNC}).  It suffices to
note that if one implements a measurement procedure on half of a correlated
pair of systems and one conditions upon its outcome, then this procedure can also
be considered a preparation procedure for the other half of the correlated
pair.  Indeed, given the separated pair of single-query 3-box systems
considered in Sec.~\ref{sec:separatedpair3box}, every measurement procedure $M_t$ on the
3-box system in \Alice chosen from $t\in \{1,2,3\}$ and yielding outcome
$b\in\{0,1\}$ corresponds to a preparation procedure $P_{t,b}$ for the 3-box
system in Babylon.  If $M_t$ is measured in \Alice but one does not condition
on the outcome, then this corresponds to a preparation procedure $P_{t}$ of
the system in Babylon.  In this case, the probability of observing an outcome
$X$ for a measurement of $M_y$ in Babylon given a preparation $P_{t,b}$ is
precisely equal to the probability of observing an outcome $X$ for a
measurement of $M_y$ in Babylon given an outcome $b$ for $M_t$ in \Alice.
There is an isomorphism between the diachronic pair of single-query 3-box
systems and the separated pair.

Now suppose that the \Aliceian and Babylonian measurements are space-like
separated. In this case, the no-signaling constraint ensures that the choice
of $t$ in \Alice cannot influence the outcome statistics of any measurement in
Babylon and consequently that the three preparation procedures $P_{1}$,
$P_{2}$ and $P_{3}$ are operationally equivalent, that is, $\forall
M:p(X|M;P_1)=p(X|M;P_2)=p(X|M;P_3)$.  This is the condition of trit-obliviousness.

Furthermore, an assumption of local causality implies that the choice of
measurement in \Alice also cannot influence the distribution over ontic
states for the 3-box system in Babylon.  Denoting the ontic state of the
Babylonian system by $\lambda$, local causality implies
$p(\lambda|P_1)=p(\lambda|P_2)=p(\lambda|P_3)$.  But this is precisely the
content of the assumption of preparation noncontextuality for the
operationally equivalent procedures $P_{1}$, $P_{2}$ and $P_{3}$.  Therefore
local causality justifies this assumption.

This reasoning also shows that any local strategy for winning the prediction
game for the separated pair of single-query 3-box systems implies a
preparation-noncontextual strategy for winning the prediction game for the
diachronic pair with the same winning probability \footnote{It has been shown
that for every inequality on correlations between pairs of separated
measurements that is implied by the assumption of a local ontological model,
an equivalent inequality for the correlations between preparations and
measurements is implied by the assumption of a preparation-noncontextual ontological model~\cite{Barrettunpublished}.}  It follows that another way to
derive the local bound of $7/9$ for the probability of achieving the OS
correlations for the separated pair, Eq.~(\ref{Ineq:BI:GameTheoretic}), is to
appeal to this implication and the fact that the optimal preparation-noncontextual strategy achieves a winning probability of $7/9$ for the
diachronic pair, as shown in Eq.~(\ref{eq:winpropPNC}).

\section{Joint measurability of POVMs} \label{sec:JMPOVM}

\sBox{Regarding this section, see also the erratum}

As we showed early on, we cannot find a triple of projective
measurements in quantum theory that are jointly measurable
pairwise but not triplewise. However, not all measurements in
quantum theory are projective. The most general measurement
is one that is associated with a positive operator valued
measure (POVM). A POVM is a set of operators $\{E_{X}:X\in
S\}$ such that $E_{X}\geq0$, and $\sum_{X}E_{X}=\unit$. The
parameter $X$ labels the outcomes of the measurement, which we
assume form a discrete set. If the preparation procedure
preceding the measurement is represented by the density
operator ${\rho}$, then the probability of outcome $X$ is
given by $\mathrm{Tr}(\rho E_{X}).$

\ In this section, we consider the question of whether one could find a triple
of non-projective measurements in quantum theory that are pairwise but not
triplewise jointly measurable. As it turns out, this is indeed possible.

First, we adapt the definition of joint measurability to the case of POVMs. \
A pair of measurements associated with POVMs $\{E_{X_{1}}^{1}\}$ and
$\{E_{X_{2}}^{2}\}$ are jointly measurable iff there exists a third POVM
$\{F_{X_{1},X_{2}}\}$ such that $E_{X_{1}}^{1}=\sum_{X_{2}}F_{X_{1},X_{2}}$
and $E_{X_{2}}^{2}=\sum_{X_{1}}F_{X_{1},X_{2}}$. It is worth noting that
the problem of mathematically characterizing jointly measurable observables
when these are not projective is a subject of on-going research
\cite{BGL95,HRS08,Wolf:PRL:2009,Yu:PRA:2010}.

We will consider two examples of such triples of POVMs  such
that any pair can be implemented jointly, but the triple cannot. They both
make use of noisy spin observables. The three measurements we consider,
labelled by an integer $k\in\{1,2,3\},$ are associated with POVMs $\left\{
E_{+}^{k},E_{-}^{k}\right\}  $ %
\begin{equation}
E_{\pm}^{k}\equiv\frac{1}{2}\unit\pm\frac{\eta}{2}\vec{\sigma}\cdot \hat{n}_{k},
\quad 0\le\eta\le1,
\label{eq:noisyspinobservables}
\end{equation}
where $\vec{\sigma}=\left(  \sigma_{x},\sigma_{y},\sigma_{z}\right)$ is the
vector of Pauli spin operators, whilst $\hat{n}_{1},$ $\hat{n}_{2}$ and
$\hat{n}_{3}$ are the three axes along which the spin is measured. Note
that the POVM $\left\{  E_{+}^{k},E_{-}^{k}\right\}  $ can be written as a
convex combination of the projective spin measurement along $\hat{n}_{k}$ --- associated with
the projectors
$\Pi_{\pm}^{k}\equiv\frac{1}{2}\unit\pm\frac{1}{2}\vec{\sigma}\cdot\hat{n}_{k}$
--- and the trivial measurement $\{\unit/2,\unit/2\}.$ That is,%
\begin{equation}
E_{\pm}^{k}=(1-\eta)\frac{1}{2}\unit+\eta\Pi_{\pm}^{k}.\label{eq:sharpplusnoise}%
\end{equation}
This is the sense in which we can consider $\left\{  E_{+}^{k},E_{-}%
^{k}\right\}  $ with $\eta<1$ to be a noisy version of the observable
$\vec{\sigma}\cdot\hat{n}_{k}.$

\subsection{Orthogonal spin axes}
Our first example of such a triple of nonprojective measurements uses noisy versions
of spin operators along three orthogonal axes:%
\begin{align}
\hat{n}_{1}  &  =\hat{z}\equiv(0,0,1),\nonumber\\
\hat{n}_{2}  &  =\hat{x}\equiv(1,0,0),\nonumber\\
\hat{n}_{3}  &  =\hat{y}\equiv(0,1,0). \label{eq:orthogonalspinaxes}%
\end{align}
\begin{proposition} \label{prop:etafororthogonaltriple}
The triple of measurements defined by Eqs.~(\ref{eq:noisyspinobservables})
and (\ref{eq:orthogonalspinaxes}), that is, noisy spin observables along
three orthogonal axes, are pairwise jointly measurable iff
$\eta\leq1/\sqrt{2}\approx0.707$, but triplewise jointly measurable iff
$\eta\leq1/\sqrt{3}\approx0.577$.
\end{proposition}
In other words, the condition $1/\sqrt{3}<\eta\leq1/\sqrt{2}$ is
necessary and sufficient for the triple to be pairwise jointly measurable but
not triplewise jointly measurable.

This result is proven in Ref.~\cite{HRS08}, but for completeness, we provide
an independent proof in Appendix \ref{app:jointmmtPOVMs}.  For pedagogical
reasons, we also provide a geometric picture in the Bloch sphere of the
measurements that saturate these inequalities. To this end, defining the
index set $\mathcal{I}\subseteq\{1,2,3\}$, we introduce the (unnormalized
Bloch) vectors
\begin{equation}\label{Eq:vec:m2}
\vec{m}_{\{X_k\}_{k\in\mathcal{I}}}\equiv\sum_{k\in\mathcal{I}}X_{k}\hat{n}_{k},
\end{equation}
 where $X_k\in\{-1,+1\}$  and write the respective unit vectors as
$\hat{m}_{\{X_k\}_{k\in\mathcal{I}}}$.

The POVM that measures a noisy spin observable along the $\hat{z}$-axis
jointly with the one along the $\hat{x}$-axis and that saturates $\eta
\leq1/\sqrt{2}$ is of the form
\begin{equation} \label{eq:JMzandx}
\left\{  F_{X_{1}X_{2}}\equiv\frac{1}{2}
\Pi_{X_{1}X_{2}}\right\}
\end{equation}
where the projectors $\left\{  \Pi_{X_{1}X_{2}%
}\right\}$ are associated with Bloch vectors $\left\{  \hat{m}_{X_{1}X_{2}%
}\right\}$ forming the vertices of a square in the $\hat{z}$-$\hat{x}$ plane,
depicted in Fig.~\ref{Fig:JMXandZ}. Coarse-graining over $X_{2}$ yields the
POVM $\left\{  F_{\pm}^{1}\equiv \frac{1}{2}\unit+\frac{1}{2}
\vec{\sigma}\cdot\vec{s}_{\pm}^{1} \right\}$ where $\vec{s}_{\pm}^{1}\equiv
\pm\frac{1}{\sqrt{2}}\hat{z}$, which is to say, a measurement of the
$\eta$-sharp spin observable along the $\hat{z}$ axis with
$\eta=\frac{1}{\sqrt{2}}$, depicted in Fig.~\ref{Fig:JMXandZ}. Similarly,
coarse-graining over $X_{1}$ yields  noisy spin observable
associated with Bloch vectors $\vec{s}_{\pm}^{2}=\pm\frac{1}{\sqrt{2}}\hat
{x},$ which is to say along the $\hat{x}$ axis with
$\eta=\frac{1}{\sqrt{2}}$.  Joint measurements of every other pair of spin
axes are described similarly.

\begin{figure}[h!]
     \begin{center}
   \includegraphics[width=7cm]{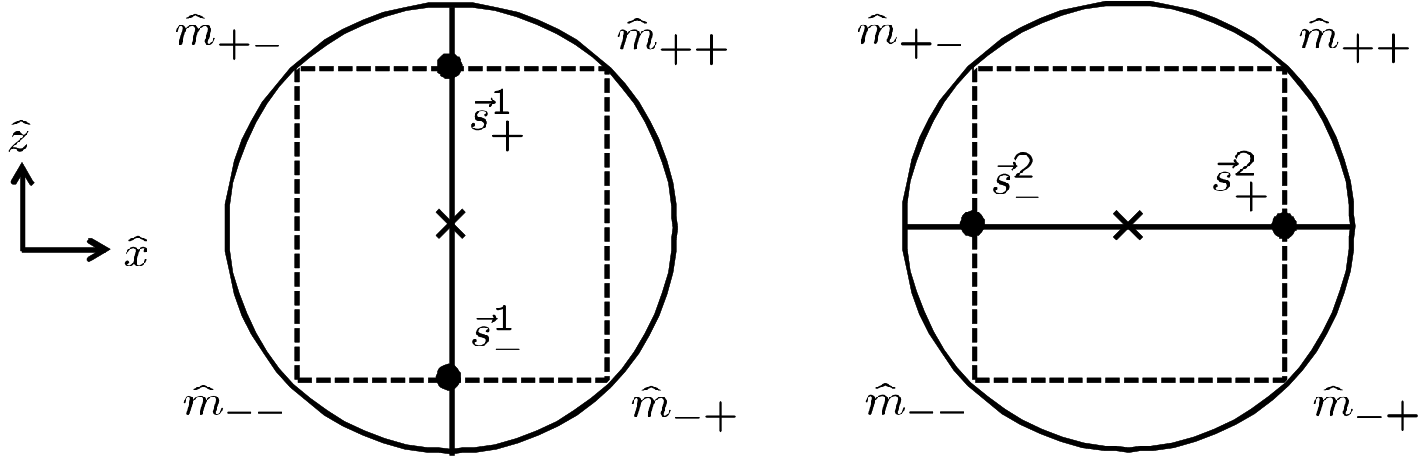}
     \end{center}
    \caption{\label{Fig:JMXandZ}
    Bloch sphere representation of the joint measurement of the noisy spin observables along the $\hat{x}$ and $\hat{z}$ axes.
    }
\end{figure}

The POVM that measures noisy spin observables along axes $\hat{z},\hat{x}$ and
$\hat{y}$ jointly and that saturates $\eta\leq1/\sqrt{3}$ is of the form
$\left\{  F_{X_{1}X_{2}X_{3}}\equiv\frac{1}{4}\Pi_{X_{1}X_{2}X_{3}}\right\}  $
where the projectors $\left\{  \Pi_{X_{1}X_{2}X_{3}}\right\}  $ are associated
with the Bloch vectors $\left\{  \hat{m}_{X_{1}X_{2}X_{3}}\right\}  $ forming
the vertices of a cube, depicted in Fig.~\ref{Fig:JMXandYandZ}. Coarse-graining over $X_{2}$
and $X_{3}$ yields the POVM $\left\{  F_{\pm}^{1}\equiv \frac{1}{2}\unit+\frac{1}{2} \vec{\sigma}\cdot\vec{s}_{\pm}^{1} \right\}$ where
$\vec{s}_{\pm}^{1}=\pm\frac{1}{\sqrt{3}}\hat{z}$, which is to say an $\eta$-sharp spin observable along the $\hat{z}$ axis with $\eta=1/\sqrt{3}$, also depicted in Fig.~\ref{Fig:JMXandYandZ}.
Similarly, coarse-graining over $X_{1}$ and $X_{3}$ yields a noisy spin
observable associated with Bloch vectors $\vec{s}_{\pm}^{2}=\pm\frac{1}%
{\sqrt{3}}\hat{x},$ while coarse-graining over $X_{1}$ and $X_{2}$ yields one
associated with $\vec{s}_{\pm}^{3}=\pm\frac{1}{\sqrt{3}}\hat{y}.$

\begin{figure}[h!]
     \begin{center}
   \includegraphics[width=8.5cm]{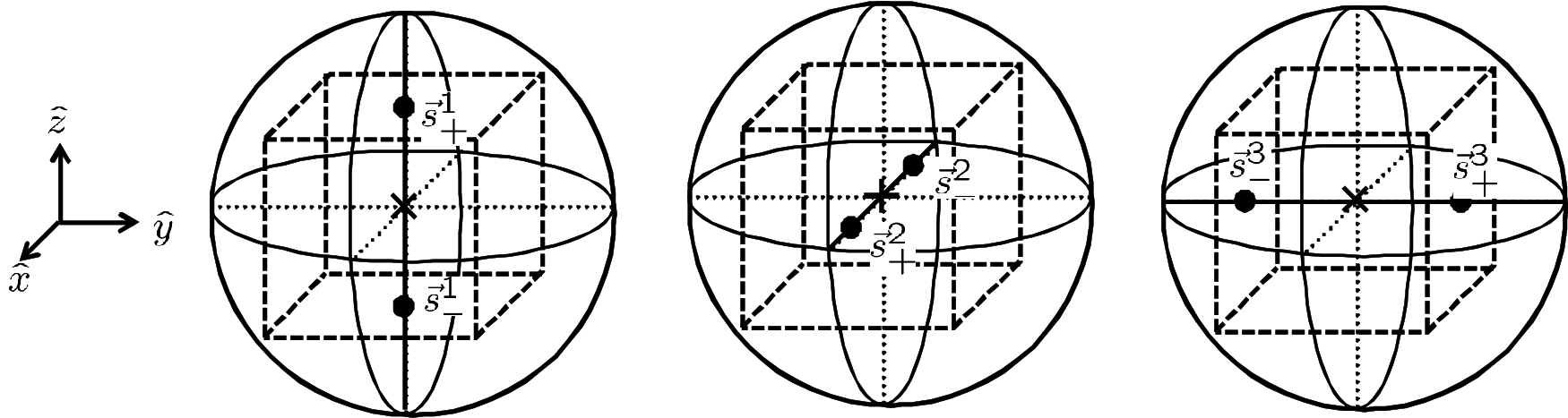}
     \end{center}
    \caption{\label{Fig:JMXandYandZ}
    Bloch sphere representation of the joint measurement of the noisy spin observables along the $\hat{x}$, $\hat{y}$ and $\hat{z}$ axes.
    }
\end{figure}

It is clear from these geometric representations that the reason there is a
gap between the noise required for jointly measuring a pair and that required
for jointly measuring the triple is that the length of the edge of a cube
inscribed in a sphere is less than that of a square inscribed in an equatorial
plane of that sphere.

Joint measurements of observables along orthogonal spin axes are not very
useful for approximating the OS correlations. Indeed, defining the probability of obtaining anti-correlated outcomes when a pair of nonprojective measurements is implemented jointly, averaged uniformly over the three pairs,
\begin{equation} \label{Retasharp}
R_{3}\equiv \frac{1}{3}\sum_{j\neq k}p(X_{j}\neq X_{k}|\mathcal{M}_{jk};\lambda
),
\end{equation}
we find the following result.

\begin{proposition} For the triple of measurements defined by Eqs.~(\ref{eq:noisyspinobservables}) and (\ref{eq:orthogonalspinaxes}), that is,
noisy spin observables along three orthogonal axes,
the quantum probability of anti-correlation when
a pair is measured jointly, averaged uniformly over the three pairs is
\begin{equation}
R_{3}^{\text{quantum}}=\frac{1}{2}.
\end{equation}
(independent of the quantum state).
\end{proposition}

\begin{proof}
The intuitive reason is that each pair of spin observables is unbiased.
More precisely, if we coarse-grain over the effects in the joint POVM
$\left\{  F_{X_{1}X_{2}}\right\}$ of Eq.~(\ref{eq:JMzandx}) with outcomes corresponding to
anti-correlation, we get
\begin{equation}
F_{+-}+F_{-+}=\frac{1}{2}\unit.
\end{equation}
Therefore, for all quantum states, the probability of finding anti-correlated results is $1/2.$
\end{proof}

There is consequently no bias towards anti-correlation and therefore this
triple of measurements is not helpful for approximating the OS correlations.

\subsection{Trine spin axes}
Our second example consists of noisy versions of spin observables along three
axes equally separated in a plane (i.e. separated by a trine or an angle of
$120^{\circ}$):
\begin{align}
\hat{n}_{1}  &  =(0,0,1)\nonumber\\
\hat{n}_{2}  &  =\left(  \frac{\sqrt{3}}{2},0,-\frac{1}{2}\right) \nonumber\\
\hat{n}_{3}  &  =\left(  -\frac{\sqrt{3}}{2},0,-\frac{1}{2}\right)  .
\label{eq:planarspinaxes}%
\end{align}
These are depicted in Fig.~\ref{Fig:JMn1andn2}.
\begin{proposition} \label{prop:etafortrine}
The triple of measurements defined by Eqs.~(\ref{eq:noisyspinobservables})
and (\ref{eq:planarspinaxes}), that is, noisy spin observables along three
equally-spaced axes in a plane, are pairwise jointly measurable if
$\eta\le\sqrt{3}-1\simeq0.732\,05,$ but triplewise jointly measurable only if
$\eta\leq2/3.$
\end{proposition}
In other words, the condition $2/3<\eta\leq\sqrt{3}-1$ is sufficient
for the triple to be pairwise jointly measurable but not triplewise jointly measurable.

Again, the proof is provided in Appendix \ref{app:jointmmtPOVMs}, but we can understand the result geometrically. The trine directions $\hat{n}_{1},\hat{n}_{2}$
and $\hat{n}_{3}$ of Eq.~(\ref{eq:planarspinaxes}) are indicated in Fig.~\ref{Fig:JMn1andn2}.
The POVM that measures a noisy spin observable along the $\hat{n}_{1}$-axis
jointly with the one along the $\hat{n}_{3}$-axis and that saturates $\eta
\leq\sqrt{3}-1$ is of the form
\begin{equation}\label{eq:defnJMtrinepair}
\left\{  F_{X_{1}X_{2}}\equiv w_{X_{1}X_{2}}\Pi_{X_{1}X_{2}}\right\}
\end{equation}
where
\begin{align}
w_{++}=w_{--}= \frac{1}{\sqrt{3}+1}\\
w_{+-}=w_{-+}= \frac{\sqrt{3}}{\sqrt{3}+1},
\end{align}
and where the projectors $\left\{  \Pi_{X_{1}X_{2}}\right\}
$ are associated with Bloch vectors $\left\{  \hat{m}_{X_{1}X_{2}}\right\}  $
forming the vertices of a square, depicted in Fig.~\ref{Fig:JMn1andn2}. Coarse-graining over
$X_{2}$ yields the POVM $\left\{  F_{\pm}^{1}\equiv\frac{1}{2}\unit+\frac{1}{2} \vec{\sigma}\cdot\vec{s}_{\pm}^{1} \right\}$ with $\vec{s}_{\pm}^{1}=\pm\left(  \sqrt
{3}-1\right)  \hat{n}_{1}$ depicted in Fig.~\ref{Fig:JMn1andn2}. Similarly, coarse-graining
over $X_{1}$ yields a noisy spin observable associated with Bloch vectors
$\vec{s}_{\pm}^{3}=\pm\left(  \sqrt{3}-1\right)  \hat{n}_{3}.$ Joint
measurements of every other pair of spin axes are described similarly.

\begin{figure}[h!]
     \begin{center}
   \includegraphics[width=8cm]{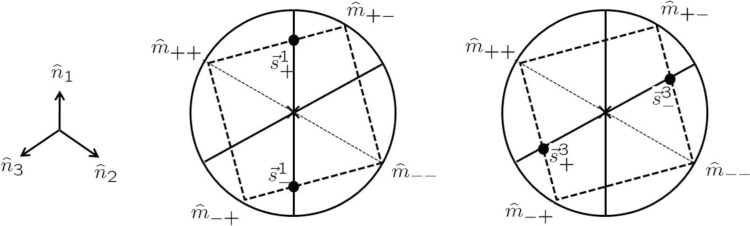}
     \end{center}
    \caption{\label{Fig:JMn1andn2}
    Bloch sphere representation of the joint measurement of the noisy spin observables along trine axes $\hat{n}_1$ and $\hat{n}_3$.
    }
\end{figure}

The POVM that measures noisy spin observables along axes $\hat{n}_{1},\hat
{n}_{2}$ and $\hat{n}_{3}$ jointly and that saturates $\eta\leq2/3$ is of the
form $\left\{  F_{X_{1}X_{2}X_{3}}\equiv w_{X_1 X_2 X_3} \Pi_{X_{1}X_{2}X_{3}%
}\right\}  $ where $w_{+++}=w_{---}=0$ (implying that one
never obtains a triplewise coincidence in the joint measurement) while $w_{+--}=w_{-++}=w_{+-+}=w_{-+-}=w_{--+}=w_{++-}=1/3$ and where the projectors $\left\{  \Pi_{X_{1}X_{2}X_{3}}\right\}$ are associated with Bloch vectors $\left\{  \hat{m}_{X_{1}X_{2}X_{3}%
}\right\} $
forming the vertices of a hexagon
for the six values of $X_{1}X_{2}X_{3}$ such that $w_{X_{1}X_{2}X_{3}}\ne 0$, as depicted in
Fig.~\ref{Fig:JMn1andn2andn3}. Coarse-graining over $X_{2}$ and $X_{3}$ yields the POVM $\left\{
F_{\pm}^{1}\equiv\frac{1}{2}\unit+\frac{1}{2} \vec{\sigma}\cdot\vec{s}_{\pm}^{1} \right\}$
where $\vec{s}_{\pm}^{1}=\pm\frac{2}{3}\hat{n}
_{1},$ depicted in Fig.~\ref{Fig:JMn1andn2andn3}. Similarly, coarse-graining over $X_{1}$ and
$X_{3}$ yields a noisy spin observable associated with Bloch vectors $\vec
{s}_{\pm}^{2}=\pm\frac{2}{3}\hat{n}_{2},$ while coarse-graining over $X_{1}$
and $X_{2}$ yields one associated with $\vec{s}_{\pm}^{3}=\pm\frac{2}{3}%
\hat{n}_{3}.$  Note that, unlike the three previous examples, the Bloch-directions of the fine-grained
(saturating) POVM elements coincide with the Bloch-directions of the coarse-grained POVM elements.
This is a peculiarity of geometry, and is a feature also seen in the dual problem of identifying pure-state
ensembles that saturate the bounds  of so-called EPR-steering inequalities \cite{Saunders10}.

\begin{figure}[h!]
     \begin{center}
   \includegraphics[width=8.5cm]{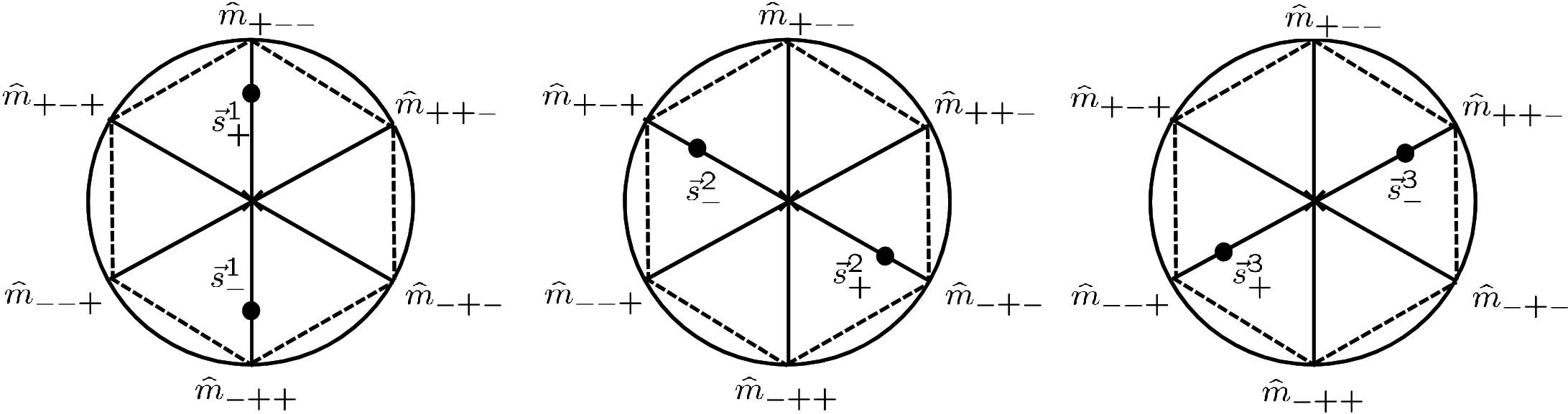}
     \end{center}
    \caption{\label{Fig:JMn1andn2andn3}
    Bloch sphere representation of the joint measurement of the noisy spin observables along trine axes $\hat{n}_1$, $\hat{n}_2$ and $\hat{n}_3$.
    }
\end{figure}

Given the discussion in Sec.~\ref{Sec:Mermin}, one might expect
the trine spin observables to instantiate a better
approximation of the OS correlations. Indeed, we have the
following proposition that supports this intuition.

\begin{proposition} For the triple of measurements defined by Eqs. (\ref{eq:noisyspinobservables}) and (\ref{eq:planarspinaxes}), that is, a triple of noisy spin observables along trine axes,
the quantum probability of anti-correlation when a pair is measured jointly,
averaged uniformly over the three pairs is
\begin{equation}
R_{3}^{\text{quantum}}= \frac{\sqrt{3}}{\sqrt{3}+1} \simeq 0.63397.
\end{equation}
(independent of the quantum state).
\end{proposition}

\begin{proof}
If, in the joint measurement of Eq.~(\ref{eq:defnJMtrinepair}), we coarse-grain the two effects that correspond to anti-correlation, we obtain
\begin{equation}
F_{+-} + F_{-+}= \frac{\sqrt{3}}{\sqrt{3}+1} \unit,
\end{equation}
from which the result follows trivially.
\end{proof}

Can we explain this degree of anti-correlation within a generalized-noncontextual ontological model? Given that the
measurements involved are nonprojective,  we need not represent them as assigning
deterministic outcomes for every ontic state.  Indeed, as discussed in Sec.~\ref{sec:PNC}, for nonprojective measurements, one is not warranted in assuming outcome
determinism. It follows that the bound of $2/3$ on the
probability of anti-correlation, Eq.~(\ref{eq:NCboundonR3}), which we derived under the
assumption of measurements being projective, need not apply.  Conceivably, the
bound implied by generalized noncontextuality could be smaller for nonprojective measurements, and the quantum
degree of anti-correlation might therefore still violate it. As it turns out
however, the bound is actually \emph{larger }for nonprojective
measurements, and therefore the quantum degree of anti-correlation is entirely consistent with an ontological model that is measurement-noncontextual and outcome-deterministic for projective measurements. We show this
now.

\subsection{Generalized-noncontextual models for joint measurements of POVMs}

Each measurement that is modeled by a POVM of the form of Eq.~(\ref{eq:noisyspinobservables}) can be considered as a convex combination of a projective measurement and a measurement of the trivial two-outcome POVM $\{\unit/2,\unit/2 \}$, as seen in Eq.~(\ref{eq:sharpplusnoise}).
In Ref.~\cite{Spe04}, it is proven that within any ontological model, the response function that represents a convex combination of measurement procedures is simply the convex combination of the associated response functions. Ref.~\cite{Spe04} also contains a proof that within a measurement-noncontextual model, the response function that represents each outcome of the trivial two-outcome POVM $\{\unit/2,\unit/2\}$ is the uniform function $1/2$, i.e., regardless of the value of $\lambda$ in the ontological model, the two outcomes occur with equal probability.   We also recall from Sec.~\ref{sec:JustifyingOD} that in models of quantum theory, preparation noncontextuality implies outcome determinism for projective measurements.   From these facts,
we obtain the following result.
\begin{lemma} In an ontological model that is generalized-noncontextual,
the response function for the $\eta$-sharp spin observable of Eq.~(\ref{eq:noisyspinobservables}),
denoted by $\mathcal{M}_{k}$,
is
\begin{equation}
p\left(  X_{k}|\mathcal{M}_{k}; \lambda\right)  =\eta\lbrack X_{k}\left(
\lambda\right)  ]+\left(  1-\eta\right)  \left(  \frac{1}{2}[0]+\frac{1}%
{2}[1]\right)  ,
\end{equation}
where $\left[  X(\lambda)\right]  $ denotes the response function
$p\left(  X|\lambda\right)  =1$ if $X=X(\lambda)$ and $0$ otherwise.
\end{lemma}
This yields a strong constraint on the response function for the joint
measurement, denoted $\mathcal{M}_{12}$, of $\eta$-sharp spin observables
along distinct axes. The joint response function $p\left(
X_{1},X_{2}|\mathcal{M}_{12}; \lambda\right) $ must yield $p\left(
X_{1}|\mathcal{M}_{1};\lambda\right)  $ when averaged over $X_{2}$ and
$p\left( X_{2}|\mathcal{M}_{2}; \lambda\right)  $ when averaged over $X_{1}.$
\ The most general form that can recover these
marginals is%
\begin{align}
p(X_{1},X_{2}|\mathcal{M}_{12};\lambda)  &  =\alpha\lbrack X_{1}(\lambda)][X_{2}%
(\lambda)]\\ \nonumber
&  +\beta\lbrack X_{1}(\lambda)]\left(  \frac{1}{2}[0]+\frac{1}{2}[1]\right)\\ \nonumber
&+\gamma\left(  \frac{1}{2}[0]+\frac{1}{2}[1]\right)  [X_{2}(\lambda)]\\ \nonumber
&  +\delta\left(  \frac{1}{2}[0]\left[  0\right]  +\frac{1}{2}[1][1]\right)\\ \nonumber
&+\varepsilon\left(  \frac{1}{2}[0][1]+\frac{1}{2}[1]\left[  0\right]  \right)
.
\end{align}
where the marginals are
\begin{align}
p(X_{1}|M_{12}; \lambda)  &  =\left(  \alpha+\beta\right)  [X_{1}
(\lambda)]\\ \nonumber
&+\left(  \gamma+\delta+\varepsilon\right)  \left(  \frac{1}%
{2}[0]+\frac{1}{2}[1]\right) \\
p(X_{2}|M_{12}; \lambda)  &  =\left(  \alpha+\gamma\right)  [X_{2}%
(\lambda)]\\ \nonumber
&+\left(  \beta+\delta+\varepsilon\right)  \left(  \frac{1}%
{2}[0]+\frac{1}{2}[1]\right),
\end{align}
so that we require%
\begin{align}
\alpha+\beta &  =\alpha+\gamma=\eta,\\
\gamma+\delta+\varepsilon &  =\beta+\delta+\varepsilon=1-\eta.
\end{align}
We infer that $\beta=\gamma.$

In order to give the model the best chance of reproducing the
operational statistics, we consider what values of $\alpha,\beta,\gamma,$
$\delta$ and $\varepsilon$ achieve the largest possible amount of
anti-correlation. The $\delta$ terms always yields correlation, while the
$\beta$ and $\gamma$ terms yield correlation as often as anti-correlation.
\ Only the $\alpha$ and $\varepsilon$ terms can have anti-correlation more
frequently than correlation. Thus, to maximize the amount of
anti-correlation, one sets $\beta=\gamma=\delta=0.$ It then follows that
$\alpha=\eta$ and $\varepsilon=1-\eta.$

 The same reasoning applies for the joint measurements of $\mathcal{M}%
_{1}$ and $\mathcal{M}_{3}$ and of $\mathcal{M}_{2}$ and $\mathcal{M}_{3},$ so
that for all $i,j\in \{1,2,3\}$ such that $i\ne j$,
\begin{align}
p(X_{i},X_{j}|\mathcal{M}_{ij};\lambda)  &  =\eta\lbrack X_{i}(\lambda
)][X_{j}(\lambda)] \\ \nonumber
&+\left(  1-\eta\right)  \left(  \frac{1}{2}[0][1]+\frac
{1}{2}[1]\left[  0\right]  \right).
\end{align}
\ The question then arises of how much anti-correlation one can have on average
for a pair of measurements (assuming the pair is chosen uniformly at
random), that is, what is the upper bound on $R_{3}$ of Eq.~(\ref{Retasharp})?
\ For every $\lambda,$ at most two out of the three products $[X_{1}%
(\lambda)][X_{2}(\lambda)],$ $[X_{1}(\lambda)][X_{3}(\lambda)]$ and
$[X_{2}(\lambda)][X_{3}(\lambda)]$ can yield anti-correlation, so the
probability of anti-correlation for the $\eta$ term is at most $2/3.$
\ Meanwhile, the $1-\eta$ term always yield anti-correlation. Therefore,%
\begin{equation}
    R_{3}  \leq\eta\left(  \frac{2}{3}\right)
    +\left(1-\eta\right) =1-\frac{\eta}{3}.
\end{equation}

One might have expected that the ability to add noise to the response function
in the ontological model would not help explain a high degree of
anti-correlation, but such an expectation fails to take into account the
fact that the noise can itself be anti-correlated and thereby explain \emph{more} anti-correlation in the statistics. Thus
rather than only being able to explain a probability of anti-correlation of 2/3
in a generalized-noncontextual model, we can explain a probability of anti-correlation of
$1-\frac{\eta}{3}$ which is always \emph{greater than }$2/3$ because $\eta
\leq1.$ For instance, for $\eta=1/\sqrt{2},$ the upper bound on $R_{3}$ is $1-1/(3\sqrt{2})\simeq 0.76430,$ while for $\eta=\sqrt{3}-1,$ it is $(4-\sqrt{3})/3\simeq 0.75598.$

Because the degree of anti-correlation we found in quantum theory was less than
$2/3$ in both examples, there is no problem providing a generalized-noncontextual model.
 More precisely, the degree of quantum anti-correlation obtained in the
example with orthogonal spin axes can be explained noncontextually because
$R_{3}^{\text{quantum}}=1/2<$ $0.76430$, and the degree
obtained in the example with the trine spin axes can be explained
noncontextually because  $R_{3}^{\text{quantum}}=0.63397<0.75598.$

Is it the case that for \emph{all} triples of nonprojective quantum measurements that can be implemented pairwise but not triplewise, the strength of anti-correlations can be explained by a generalized-noncontextual ontological model?  The question remains open, but we expect a positive answer.

\section{Concluding Remarks}

There has been a lot of work in recent years on ``foils to
quantum theory'', operational theories that one studies not
primarily as competitors to quantum theory, but as useful tools
for getting a handle on the principles underlying it.  Only by
situating quantum theory in a landscape of possible theories
does it make sense to speak of the principles that pick it out,
 to answer Wheeler's question: ``how come the quantum?''.  Specker's parable
provides an interesting new kind of foil, because the kind of
complementarity it exhibits --- three measurements
that can be implemented jointly pairwise but not triplewise --
is something that is not found among projective measurements in
quantum theory.  This prompts the question: why does quantum
theory not have this sort of complementarity?  It
might be interesting, for instance, to deduce the
information-processing power of a foil theory incorporating
such correlations. Furthermore, even if we consider a kind of
complementarity that can be accommodated in quantum theory,
such as five measurements that can be measured in adjacent
pairs, there is an interesting question about why the
correlations exhibited by quantum theory are not stronger.  Why
is  quantum theory not more contextual or more
nonlocal~\cite{PR94,vanDam, Ll.Masanes:PRA:012112,G.Brassard:PRL:2006, N.Linden:PRL:2007,
N.Brunner:PRL:2009,IC,J.Allcock:PRA:2009, M.Navascues,
A.Acin:PRL:140404,J.Oppenheim:1004.2507}? The same sort of question arises for
quantum examples of triples of \emph{nonprojective}
measurements that can be implemented pairwise but not
triplewise.  Why can these not yield the strength of
anti-correlations required to obtain a no-go theorem for generalized noncontextuality?  We
hope that these questions might provide a new angle on the
problem of deriving the structure of quantum theory from within
a landscape of operational foil theories.

\section{Acknowledgement}
This project was inspired by Ernst Specker's talk at the
workshop ``Information Primitives and Laws of Nature'' which
took place at ETH, Z\"{u}rich in May 2008.  On the topic of the
joint measurability of POVMs, we thank Robin Blume-Kohout for a
motivating discussion and David Pegg for helpful comments. On
the topic of nonlocal OS correlations, we acknowledge useful
discussions with Ben Toner and Jean-Daniel Bancal.
We also thank Allen Stairs for comments on a draft of this article and Lucien Hardy for
bringing Ref.~\cite{D.Boschi:PRL:2755} to our attention and for
suggesting the connection between Specker's parable and the
failure of transitivity of implication in proofs of
nonlocality.  Finally, we thank Ad\'an Cabello for pointing out the connection between our Kochen-Specker proof based on the failure of transitivity of implication and Clifton's proof.  YCL and HMW acknowledge funding from the Australian Research Council. Part of this work was conducted
during visits by YCL and HMW to Perimeter Institute, and RWS to
Australia, through
the PIAF (Perimeter Institute Australia
Foundations) collaboration. Perimeter Institute is supported by
the Government of Canada through Industry Canada and by the
Province of Ontario through the Ministry of Research and
Innovation. YCL was also supported by the Swiss NCCR
"Quantum Photonics" and the European ERC-AG QORE.

\appendix

\section{Explicit form of OS correlations in the double-query,
3-box system} \label{App:SpeckerCorrelations}

Perfect negative correlation in the outcomes of the joint
observables $\mathcal{M}_{12}, \mathcal{M}_{13}$ and
$\mathcal{M}_{23}$ constrains their statistics to be of the
form
\begin{align}\label{eq:MABstats}
p(0,1|\mathcal{M}_{12};\mathcal{P}_{\ast})  & =q_{12},\quad
p(1,0|\mathcal{M}_{12};\mathcal{P}_{\ast})  & =1-q_{12}, \nonumber \\
p(0,1|\mathcal{M}_{13};\mathcal{P}_{\ast})  & =q_{13},\quad
p(1,0|\mathcal{M}_{13};\mathcal{P}_{\ast})  & =1-q_{13}, \nonumber \\
p(0,1|\mathcal{M}_{23};\mathcal{P}_{\ast})  & =q_{23},\quad
p(1,0|\mathcal{M}_{23};\mathcal{P}_{\ast})  & =1-q_{23}
\end{align}
for $0\leq q_{12},q_{13},q_{23}\leq1.$
This fixes the statistics for the individual measurements
$\mathcal{M}_{1},$ $\mathcal{M}_{2}$ and $\mathcal{M}_{3}$ through the
marginals of Eq.~(\ref{eq:MABstats}). Specifically,
\begin{align}
p(0|\mathcal{M}_{1};\mathcal{P}_{\ast})  & =\sum_{X_{2}}p(0,X_{2}%
|\mathcal{M}_{12};\mathcal{P}_{\ast})=q_{12},\\
& =\sum_{X_{3}}p(0,X_{3}|\mathcal{M}_{13};\mathcal{P}_{\ast})=q_{13},\\
p(0|\mathcal{M}_{2};\mathcal{P}_{\ast})  & =\sum_{X_{1}}p(X_{1},0|\mathcal{M}%
_{12};\mathcal{P}_{\ast})=1-q_{12},\\
& =\sum_{X_{3}}p(0,X_{3}|\mathcal{M}_{23};\mathcal{P}_{\ast})=q_{23},\\
p(0|\mathcal{M}_{3};\mathcal{P}_{\ast})  & =\sum_{X_{1}}p(X_{1},0|\mathcal{M}%
_{13};\mathcal{P}_{\ast})=1-q_{13},\\
& =\sum_{X_{2}}p(X_{2},0|\mathcal{M}_{23};\mathcal{P}_{\ast})=1-q_{23}%
\end{align}
 All together, we find that we must have
\[
q_{12}=q_{13}=q_{23}=\frac{1}{2},
\]
which implies that the correlations are of the form of Eq.~(\ref{eq:OScorrelations}).

\section{Proof of theorem~\ref{Thm:Fine}}
\label{App:ProofFineTheorem}

\noindent Here, we provide the proof of theorem~\ref{Thm:Fine}
as follows.

\begin{proof} By measurement noncontextuality, the response function depends
only on the equivalence class of a measurement procedure.  By
outcome determinism, the response
function for every measurement $\mathcal{M}_{S}$ is deterministic, so that $%
p\left( X_{S}|\mathcal{M}_{S};\lambda \right) \in \left\{ 0,1\right\} .$ \
In particular, this is true for singleton sets.  It follows that%
\begin{equation}
p\left( X_{S}|\mathcal{M}_{S};\lambda \right) =\prod\limits_{s\in S}p\left(
X_{s}|\mathcal{M}_{s};\lambda \right) .
\end{equation}%

We can then define a joint distribution $p\left(
X_{1}...X_{N}|\lambda \right) $ yielding the correct marginals
by the product of the single measurement response functions,%
\begin{equation}
p\left( X_{1}...X_{N}|\lambda \right) \equiv \prod\limits_{s=1}^{N}p\left(
X_{s}|\mathcal{M}_{s};\lambda \right) .
\end{equation}
By assumption of the empirical adequacy of the ontological model, there exists
a distribution $p(\lambda |\mathcal{P})$ for all $\mathcal{P},$ such
that%
\begin{equation}
\int d\lambda p\left( X_{S}|\mathcal{M}_{S},\mathcal{\lambda }\right)
p(\lambda |\mathcal{P})=p\left( X_{S}|\mathcal{M}_{S};\mathcal{P}\right) .
\end{equation}%
Using $p(\lambda |\mathcal{P}),$ we can define%
\begin{equation}
p\left( X_{1}...X_{N}|\mathcal{P}\right) =\int d\lambda p\left(
X_{1}...X_{N}|\lambda \right) p(\lambda |\mathcal{P})
\end{equation}%
which has marginal on $X_{S}$ of%
\begin{eqnarray}
p\left( X_{S}|\mathcal{P}\right)  &=&\sum_{X_{\nu }:\nu \notin S}p\left(
X_{1}...X_{N}|\mathcal{P}\right)  \notag\\
&=&\int d\lambda \sum_{X_{\nu }:\nu \notin S}p\left( X_{1}...X_{N}|\mathcal{%
\lambda }\right) p(\lambda |\mathcal{P}) \notag\\
&=&\int d\lambda p\left( X_{S}|\mathcal{M}_{S},\mathcal{\lambda }\right)
p(\lambda |\mathcal{P}) \notag\\
&=&p\left( X_{S}|\mathcal{M}_{S};\mathcal{P}\right)
\end{eqnarray}%
We have therefore shown that $p\left(
X_{1}...X_{N}|\mathcal{P}\right) $ is a joint distribution
whose marginals yield the operational statistics of all
measurements.
\end{proof}

\section{Maximal quantum violation of the $n$-box-set Klyachko-type
Kochen-Specker inequality} \label{App:MaxQV:Klyachko}

To see that $\S_n^\text{quantum}$ given in
Eq.~\eqref{Eq:MaxKlyachkoViolation:n} is indeed the strongest
possible quantum violation of
inequality~\eqref{Ineq:Noncontextual:Sn}, it suffices to
consider the following polynomial of Hermitian operators
\begin{equation}\label{Eq:KlyachkoOperator}
    \hat{\mathcal{B}}_{n}=\sum_{a=1}^n\bar{X}_a\bar{X}_{a\oplus1},
\end{equation}
and note that for arbitrary Hermitian operators satisfying the
commutation relation $[\bar{X}_a, \bar{X}_{a\oplus1}]=0$, we
have
\begin{align}
    &\hat{\mathcal{B}}_{n}
    -n\left(1-\frac{4\cos\frac{\pi}{n}}{1+\cos\frac{\pi}{n}}\right)\unit\nonumber\\
    =&\frac{1}{4}\left(2-\sec\frac{\pi}{n}\right)\sum_{a=1}^n
    \left[\unit-\left(\bar X_a\right)^2\right]
    +\frac{1}{4}\sum_{a=1}^n
    \left[\unit-\left(\bar X_a\bar X_{a\oplus1}\right)^2\right]
    \nonumber\\
    &+\frac{1}{4}\sec\frac{\pi}{n}\sum_{a=1}^n \bar X_a\bar X_{a\oplus 2}\left[\unit-\left(\bar X_{a\oplus1}\right)^2\right]\nonumber\\
    &+\frac{1}{4n}\left(1+\sec\frac{\pi}{n}\right)v_{0}^\dag v_0^{\phantom{\dag}}
    + \frac{1}{n}\sum_{k=1,2}\sum_{j=1}^n \lambda_{kj} v_{kj}^\dag v_{kj}^{\phantom{\dag}},
    \label{Eq:KlyachkoDecomposition}
\end{align}
where
\begin{gather*}
    v_0=n\left(3-2\sec^2\frac{\pi}{2n}\right)\unit+
    \sum_{a=1}^n \bar X_a\bar X_{a\oplus 1},\\
    v_{1j}=\sum_{a=1}^n \omega_n^{ja} \bar X_a,\quad
    v_{2j}=\sum_{a=1}^n \omega_n^{ja} \bar X_a\bar X_{a\oplus 1}
\end{gather*}
$\omega_n={\rm e}^{-i2\pi/n}$ is the $n$-th root of unity,
$\lambda_{2n}=0$, and
\begin{align*}
\lambda_{1j}&=\left(1+\cos\frac{2\pi j}{n}\sec\frac{\pi}{n}\right)
\sin^2\frac{\pi j}{n},\quad j=1,\ldots,n,\\
\lambda_{2j}&=\frac{1}{4}\left(1+\cos\frac{2\pi j}{n}\sec\frac{\pi}{n}\right)
,\quad j=1,\ldots,n-1.
\end{align*}
It is straightforward to check that $\lambda_{1j}$ and $\lambda_{2j}$ are
non-negative for all $j$. Thus, for dichotomic Hermitian observables that
also satisfy $\left(\bar X_a\right)^2=\unit$, the right-hand-side of
Eq.~\eqref{Eq:KlyachkoDecomposition} become a sum of non-negative Hermitian
operators and hence non-negative. As a result, the smallest eigenvalue of
$\hat{\mathcal{B}}_n$ is lower bounded by
$n\left(1-\tfrac{4\cos\tfrac{\pi}{n}}{1+\cos\tfrac{\pi}{n}}\right)$,
therefore making $\S_n^\text{quantum}$ given in
Eq.~\eqref{Eq:MaxKlyachkoViolation:n} the strongest possible quantum
violation of inequality~\eqref{Ineq:Noncontextual:Sn}.

\section{Explicit form of correlations in the separated pair
of single-query 3-box systems} \label{App:Correlations:OS2}

Here, we will give a simple proof that after taking into
account of the no-signaling condition, the nonlcoal OS
correlations have to take the form of
Eq.~\eqref{Eq:SpeckerCorrelations:Explicit}.

First, note that by virtue of satisfying
Eq.~\eqref{Eq:SpeckerCorrelations}, the nonlocal OS
correlations may be written as
\begin{eqnarray}\label{eq:OS2withqs}
\forall a\ne b: &p(0,1|\mathcal{M}_a,\mathcal{M}_{b};\mathcal{P_*})=q_{ab}
\nonumber \\
&p(1,0|\mathcal{M}_a,\mathcal{M}_{b};\mathcal{P_*})=1-q_{ab} \nonumber \\
\forall a= b: &p(0,0|\mathcal{M}_a,\mathcal{M}_{b};\mathcal{P_*})=q_{ab}
\nonumber \\
&p(1,1|\mathcal{M}_a,\mathcal{M}_{b};\mathcal{P_*})=1-q_{ab},
\end{eqnarray}
where $0\le q_{ab}\le 1$. These joint probabilities are depicted in Table~\ref{tbl:OS2}.

\begin{table}[h!btp]
    \begin{tabular}{|c|c|cc|cc|cc|}\hline
     & & \multicolumn{2}{c|}{$B_1$} & \multicolumn{2}{c|}{$B_2$} & \multicolumn{2}{c|}{$B_3$} \\
    \hline
     & & 0 & 1 & 0 & 1 & 0 & 1 \\ \hline
    \multirow{2}{*}{$A_1$} & 0 & $q_{11}$ & 0 & 0 & $q_{12}$ & 0 & $q_{13}$ \\
                           & 1 & 0 & 1-$q_{11}$ & 1-$q_{12}$ & 0 & 1-$q_{13}$ & 0 \\ \hline
    \multirow{2}{*}{$A_2$} & 0 & 0 & $q_{21}$ & $q_{22}$ & 0 &  0 & $q_{23}$  \\
                           & 1 & 1-$q_{21}$ &0 & 0 & 1-$q_{22}$ & 1-$q_{23}$& 0 \\ \hline
    \multirow{2}{*}{$A_3$} & 0 & 0 & $q_{31}$ & 0 & $q_{32}$ & $q_{33}$ & 0 \\
                           & 1 & 1-$q_{31}$ & 0 & 1-$q_{32}$ & 0 & 0 & 1-$q_{33}$ \\ \hline
    \end{tabular}
    \caption{\label{tbl:OS2} Joint conditional probability distributions
    $p(A_a, B_b | \M_a, \M_{b};\P)$ of Eq.~(\ref{eq:OS2withqs}) for all pairs of values of $a$ and $b$.  Along the horizontal (vertical) are the three choices of measurement on the $B(A)$ wing together with the two outcomes for each.}
\end{table}

Now we consider the consequences of the no-signaling
conditions of Eq.~(\ref{eq:non-signaling}).   From the
independence on $b$ of $\sum_{B_b}p(A_a,B_b |\M_a,\M_b;\P)$, we
deduce that
\begin{equation}
q_{a1}=q_{a2}=q_{a3},
\end{equation}
implying that the joint distributions can be made to depend on just three
parameters, which we will denote by $s_a=q_{ab}$. It then follows from the
independence on $a$ of $\sum_{A_a}p(A_a,B_b |\M_a,\M_b;\P)$ that
\begin{equation}
s_a=1-s_a,
\end{equation}
which implies that $s_a=\frac{1}{2}$ for all $a$, and therefore $q_{ab} =
\frac{1}{2}$ for all $a,b$.  It follows that the nonlocal OS correlations
must be of the form given in Eq.~\eqref{Eq:SpeckerCorrelations:Explicit} if
they are to be non-signaling.

\section{Maximum quantum violation of the $n$-box-set Bell-Mermin inequality}
\label{App:MaxQV:Bell:n}

In general, the problem of determining the maximal quantum violation of a
Bell inequality is highly non-trivial (see, for example,
Refs.~\cite{QMP.Hierarchy,R.F.Werner:QIC:2001,YCL:PRA:2007} and references therein). Here, we
will show that $R_n^\text{quantum}$ defined in Eq.~\eqref{Eq:MaxQuantumRn} is
indeed the maximal winning probability, {\em cf.}
Eq.~\eqref{Eq:Dfn:WinningProb:TwoWing}, allowed in quantum mechanics. To this
end, it suffices to show that the winning probability $R_n$ is upper bounded
by $R_n^\text{quantum}$ in quantum theory. For convenience, we will show this
in terms of
\begin{equation}
    \mathcal{S}_n\equiv\sum_{a=1}^n\langle\bar A_a \bar B_a\rangle
    -\sum_{a, b: b=a\oplus1} \langle\bar A_a \bar B_b\rangle-\sum_{a, b: a=b\oplus1}
    \langle\bar A_a \bar B_b\rangle,%\notag\\
\end{equation}
which can be re-expressed as:
\begin{equation}\label{Eq:SnRn}
    \S_n=6n\left(R_n-\frac{1}{2}\right),
\end{equation}
using Eq.~\eqref{Eq:Dfn:Correlator}.

Now, consider the Bell operator~\cite{S.L.Braunstein:PRL:3259} corresponding
to above expression defining $S_n$:
\begin{equation}\label{Eq:BellOperator:n}
    \hat{\mathcal{B}}^{[n]}_{\mbox{\tiny NLOS }}=\sum_{a=1}^n\bar A_a \bar B_a
    -\sum_{a, b: b=a\oplus1} \bar A_a \bar B_b-\sum_{a, b: a=b\oplus1} \bar A_a \bar B_b.
\end{equation}
Following a procedure very similar to that described in Sec.~III of
Ref.~\cite{S.Wehner:PRA:022110} (see also Ref.~\cite{QMP.Hierarchy}), one
finds that for arbitrary Hermitian observables $\{\bar A_a\}_{a=1}^n$ and
$\{\bar B_b\}_{b=1}^n$ satisfying $\left[\bar A_a, \bar B_b\right]=0$,
\begin{align}
    &n\lambda_{\tfrac{n+1}{2}}\unit
    -\hat{\mathcal{B}}^{[n]}_{\mbox{\tiny NLOS}}\nonumber\\
    =&\frac{1}{2}\sum_{a=1}^n\left[\left(\lambda_{\tfrac{n+1}{2}}+\lambda_a\right)
    v_{a-}^\dag v_{a-}+\left(\lambda_{\tfrac{n+1}{2}}-\lambda_a\right)
    v_{a+}^\dag v_{a+}\right]\nonumber\\
    +&\frac{1}{2}\lambda_{\tfrac{n+1}{2}}\left\{\sum_{a=1}^n
    \left[\unit-\left(\bar A_a\right)^2\right]
    +\sum_{b=1}^n \left[\unit-\left(\bar B_b\right)^2\right]\right\},
    \label{Eq:Decomposition:n}
\end{align}
\noindent{where}
\begin{equation}
    v_{a\pm}=\frac{1}{\sqrt{2n}}\sum_{k=1}^n
    \omega^{ak}_n\left(\bar A_k\pm \bar B_k\right),\quad
    \quad\lambda_a=1-2\cos\frac{2\pi}{n}a,
\end{equation}
and $\omega_n=\text{e}^{-i2\pi/n}$.

It is easy to verify that
\begin{equation}
    \max_{a\in\{1,2,\ldots,n\}} \lambda_a
    =\lambda_{\tfrac{n+1}{2}}=4\cos^2\frac{\pi}{2n}-1.
\end{equation}
Thus, Eq.~\eqref{Eq:Decomposition:n} implies that whenever the constraints
$\left(\bar A_a\right)^2=\unit$ and $\left(\bar B_b\right)^2=\unit$ are
satisfied for all $a,b\in\{1,2,\ldots,n\}$, the right hand side of
Eq.~\eqref{Eq:Decomposition:n} becomes a sum of squares of polynomial of
Hermitian operators and hence $n\left(4\cos^2\frac{\pi}{2n}-1\right)\unit-
\hat{\mathcal{B}}^{[n]}_{\mbox{\tiny NLOS }}\ge 0$. As a result, the maximal
quantum mechanical expectation value of $\hat{\mathcal{B}}^{[n]}_{\mbox{\tiny
NLOS}}$ is upper bounded by $n\left(4\cos^2\frac{\pi}{2n}-1\right)$, so is
the maximal value of $\S_n$ allowed in quantum theory.

Equivalently, it follows from Eq.~\eqref{Eq:SnRn} that in quantum theory, the
maximal winning probability $R^n$ is upper bounded by:
\begin{equation}
    \frac{1}{2}+\frac{1}{6n}\times n\left(4\cos^2\frac{\pi}{2n}-1\right)=
    \frac{1}{3}+\frac{2}{3}\cos^2\frac{\pi}{2n}.
\end{equation}
which is just $R^{\mbox{\tiny quantum}}_n$ given in
Eq.~\eqref{Eq:MaxQuantumRn}.

\section{Necessary and sufficient conditions for joint measurability of noisy spin observables} \label{app:jointmmtPOVMs}

\begin{theorem} Consider a set of noisy spin observables along the axes
$\hat{n}_{k}$, that is, a set of POVMS $\left\{  E{_{X_{k}}^{k}}\right\}  $
with $X_{k}\in\left\{  +1,-1\right\}  $ of the form%
\begin{equation}
E_{X_{k}}^{k}=\frac{1}{2}\unit+\frac{1}{2}\vec{\sigma}\cdot X_{k}\eta\hat{n}_{k}.%
\end{equation}
Defining  $2^N$ different 3-vectors %
\begin{equation}\label{Eq:vec:m}
\vec{m}_{X_{1}\dots X_{N}}\equiv\sum_{k=1}^{N}X_{k}\hat{n}_{k},
\end{equation}
a necessary condition for the spin observables to be jointly measurable is
that%
\begin{equation}
\eta\leq\frac{1}{N}\frac{\sum_{X_{1}\dots X_{N}}|\vec{m}_{X_{1}\dots X_{N}%
}|^{2}}{\sum_{X_{1}\dots X_{N}}|\vec{m}_{X_{1}\dots X_{N}}|}%
\label{eq:etanecessary}%
\end{equation}
and a sufficient condition is that%
\begin{equation}
\eta\leq\frac{2^{N}}{\sum_{X_{1}\dots X_{N}}|\vec{m}_{X_{1}\dots X_{N}}%
|}.\label{eq:etasufficient}%
\end{equation}
\end{theorem}

\begin{proof}
Clearly,
\begin{equation}
\eta=\mathrm{Tr}\left[  \left(  \vec{\sigma}\cdot X_{k}\hat{n}_{k}\right)
E_{X_{k}}^{k}\right]  .
\end{equation}
But given that this equality holds for both values of $X_{k}$ and for all $k,$
we have%
\begin{equation}
\eta=\frac{1}{2N}\sum_{k=1}^{N}\sum_{X_{k}}\mathrm{Tr}\left[  \left(
\vec{\sigma}\cdot X_{k}\hat{n}_{k}\right)  E_{X_{k}}^{k}\right]  .
\end{equation}
Recall that joint measurability of POVMs $\left\{  E_{X_{k}}^{k}\right\}  $
for different $k$ implies the existence of another POVM $\left\{
E_{X_{1}X_{2}\dots X_{N}}\right\}  $ such that%
\begin{equation}
E_{X_{k}}^{k}=\sum_{X_{1}\dots X_{N},\text{ fix }X_{k}}E_{X_{1}\dots X_{N}}%
\end{equation}
Consequently,%
\begin{equation}
\eta=\frac{1}{2N}\sum_{X_{1}\dots X_{N}}\mathrm{Tr}\left[  \left(  \vec
{\sigma}\cdot\sum_{k=1}^{N}X_{k}\hat{n}_{k}\right)  E_{X_{1}\dots X_{N}%
}\right]  .
\end{equation}
Defining $\vec{m}_{X_{1}\dots X_{N}}$ as above, we may write $\eta$ as%
\begin{equation}
\eta=\frac{1}{2N}\sum_{X_{1}\dots X_{N}}|\vec{m}_{X_{1}\dots X_{N}%
}|\mathrm{Tr}\left[  \left(  \vec{\sigma}\cdot\hat{m}_{X_{1}\dots X_{N}%
}\right)  E_{X_{1}\dots X_{N}}\right]  .
\end{equation}
where $\hat{m}_{X_{1}\dots X_{N}}=\vec{m}_{X_{1}\dots X_{N}}/|\vec{m}%
_{X_{1}\dots X_{N}}|.$ We then note that
\begin{equation}
\mathrm{Tr}\left[  \left(  \vec{\sigma}\cdot\hat{m}_{X_{1}\dots X_{N}}\right)
E_{X_{1}\dots X_{N}}\right]  \leq\mathrm{Tr}\left[  E_{X_{1}\dots X_{N}%
}\right]  \label{eq:ineq1}%
\end{equation}
to obtain the inequality%
\begin{equation}
\eta\leq\frac{1}{2N}\sum_{X_{1}\dots X_{N}}|\vec{m}_{X_{1}\dots X_{N}%
}|\mathrm{Tr}\left[  E_{X_{1}\dots X_{N}}\right]  .
\end{equation}
We need only determine the maximum value of the right-hand-side in a
variation over all POVMs $\left\{  E_{X_{1}\dots X_{N}}\right\}  .$ Given
that $\sum_{X_{1}\dots X_{N}}E_{X_{1}\dots X_{N}}=\unit,$ we know that $\sum
_{X_{1}\dots X_{N}}\mathrm{Tr}\left[  E_{X_{1}\dots X_{N}}\right]  =2.$ \
Consequently, the $2^N$-dimensional  vector $\left(  \frac{1}{2}\mathrm{Tr}\left[ E_{X_{1}\dots
X_{N}}\right]  \right)  _{X_{1}\dots X_{N}}$ has unit 1-norm.
\ Thinking of $\sum_{X_{1}\dots X_{N}}|\vec{m}_{X_{1}\dots X_{N}}%
|\mathrm{Tr}\left[  E_{X_{1}\dots X_{N}}\right]  $ as a scalar product, we see
that it is maximized by taking $\left(  \frac{1}{2}\mathrm{Tr}\left[
E_{X_{1}\dots X_{N}}\right]  \right)  _{X_{1}\dots X_{N}}$ to be a unit vector
parallel to $\left(  |\vec{m}_{X_{1}\dots X_{N}}|\right)  _{X_{1}\dots X_{N}}$
with unit 1-norm, that is,%
\begin{equation}
\frac{1}{2}\mathrm{Tr}\left[  E_{X_{1}\dots X_{N}}\right]  =\frac{|\vec
{m}_{X_{1}\dots X_{N}}|}{\sum_{X_{1}^{\prime}\dots X_{N}^{\prime}}|\vec
{m}_{X_{1}^{\prime}\dots X_{N}^{\prime}}|},
\end{equation}
which yields the necessary condition on $\eta$, Eq.~\eqref{eq:etanecessary}.

To derive the sufficient condition, Eq.~\eqref{eq:etasufficient}, we
construct a POVM that jointly measures a set of spin observables with value
of $\eta$ saturating the inequality. Any set of observables with smaller
$\eta$ can then be jointly measured by simply adding uniformly random noise
to this POVM.

The simulating POVM is
\begin{equation}
E_{X_{1}\dots X_{N}}=\frac{2|\vec{m}_{X_{1}\dots X_{N}}|}{\sum_{X_{1}^{\prime
}\dots X_{N}^{\prime}}|\vec{m}_{X_{1}^{\prime}\dots X_{N}^{\prime}}|}\left[
\frac{1}{2}\unit+\frac{1}{2}\vec{\sigma}\cdot\hat{m}_{X_{1}\dots X_{N}}\right]
.\label{eq:POVMthatsimulates}%
\end{equation}
It suffices to demonstrate that this is indeed a POVM and that it
coarse-grains to the appropriate noisy spin observables. First, note that
\begin{align*}
\sum_{X_{1}\dots X_{N}}\vec{m}_{X_{1}\dots X_{N}}  & =\sum_{X_{1}\dots X_{N}%
}\left(  \sum_{j=1}^{N}X_{j}\hat{n}_{j}\right)  \\
& =\sum_{j=1}^{N}\hat{n}_{j}\left(  \sum_{X_{1}\dots X_{N}}X_{j}\right)  \\
& =0.
\end{align*}
while
\begin{align*}
\sum_{\{X_{i}\}_{i\neq k}}\vec{m}_{X_{1}\dots X_{N}}  & =\sum_{\{X_{i}%
\}_{i\neq k}}\left(  \sum_{j=1}^{N}X_{j}\hat{n}_{j}\right)  \\
& =\sum_{j=1}^{N}\hat{n}_{j}\left(  \sum_{\{X_{i}\}_{i\neq k}}X_{j}\right)
\\
& =\hat{n}_{k}2^{N-1}X_{k}.
\end{align*}
Also, note that because $\vec{m}_{X_{1}\dots X_{k}\dots X_{N}}=-\vec
{m}_{-X_{1}\dots-X_{k}\dots-X_{N}},$ it follows that $|\vec{m}_{X_{1}\dots
X_{N}}|=|\vec{m}_{-X_{1}\dots-X_{N}}|$ and consequently that
\begin{equation}
\sum_{\{X_{i}\}_{i\neq k}}|\vec{m}_{X_{1}\dots X_{N}}|=\frac{1}{2}\sum
_{X_{1}\dots X_{N}}|\vec{m}_{X_{1}\dots X_{N}}|.
\end{equation}
\newline

It is then easy to verify that
\begin{align*}
\sum_{X_{1}\dots X_{N}}E_{X_{1}\dots X_{N}}  & =\sum_{X_{1}\dots X_{N}}%
\frac{2|\vec{m}_{X_{1}\dots X_{N}}|}{\sum_{X_{1}^{\prime}\dots X_{N}^{\prime}%
}|\vec{m}_{X_{1}^{\prime}\dots X_{N}^{\prime}}|}\\
& \times\left[  \frac{1}{2}\unit+\frac{1}{2}\vec{\sigma}\cdot\hat{m}_{X_{1}\dots
X_{N}}\right]  \\
& =\unit,
\end{align*}
and that%
\begin{align*}
\sum_{\{X_{i}\}_{i\neq k}}E_{X_{1}\dots X_{N}}  & =\sum_{\{X_{i}\}_{i\neq k}%
}\frac{2|\vec{m}_{X_{1}\dots X_{N}}|}{\sum_{X_{1}^{\prime}\dots X_{N}^{\prime
}}|\vec{m}_{X_{1}^{\prime}\dots X_{N}^{\prime}}|}\\
& \times\left[  \frac{1}{2}\unit+\frac{1}{2}\vec{\sigma}\cdot\hat{m}_{X_{1}\dots
X_{N}}\right]  \\
& =\frac{1}{2}\unit+\frac{1}{2}\vec{\sigma}\cdot\frac{2^{N}}{\sum_{X_{1}^{\prime
}\dots X_{N}^{\prime}}|\vec{m}_{X_{1}^{\prime}\dots X_{N}^{\prime}}|}X_{k}%
\hat{n}_{k}.
\end{align*}
This establishes the sufficient condition.
\end{proof}

\begin{corollary}
The necessary and sufficient conditions for joint
measurability of a set of spin observables are: for a pair of orthogonal spin axes,%
\begin{equation}
\eta\leq\frac{1}{\sqrt{2}},\label{eq:eta1}%
\end{equation}
for a triple of orthogonal spin axes,
\begin{equation}
\eta\leq\frac{1}{\sqrt{3}},
\end{equation}
for a pair of trine spin axes,
\begin{equation}
\eta\leq\sqrt{3}-1,
\end{equation}
for a triple of trine spin axes,
\begin{equation}
\eta\leq\frac{2}{3}.\label{eq:eta4}%
\end{equation}
To saturate each of these inequalities, it suffices to implement the POVM
defined in Eq.~(\ref{eq:POVMthatsimulates}).
\end{corollary}

\begin{proof}
We begin by establishing the values of $\left\{  |\vec
{m}_{X_{1}\dots X_{N}}|\right\}  $ for each of our four examples. For
orthogonal spin axes, defined in Eq. (\ref{eq:orthogonalspinaxes}), we have
for $N=2,$%
\begin{equation}
\forall X_{1},X_{2}:|\vec{m}_{X_{1}X_{2}}|=|X_{1}\hat{n}_{1}+X_{2}\hat{n}%
_{2}|=\sqrt{2},
\end{equation}
and for $N=3,$%
\begin{align}
\forall X_{1},X_{2},X_3  & :|\vec{m}_{X_{1}X_{2}X_{3}}|\\ \nonumber
&=|X_{1}\hat{n}_{1}+X_{2}\hat{n}_{2}+X_{3}\hat{n}_{3}|
 =\sqrt{3}.
\end{align}
For trine spin axes, defined in Eq.~(\ref{eq:planarspinaxes}), we have for
$N=2,$%
\begin{align}
|\vec{m}_{++}|  & =|\vec{m}_{--}|=1,\\
|\vec{m}_{+-}|  & =|\vec{m}_{-+}|=\sqrt{3},
\end{align}
and for $N=3,$ we have (making use of the fact that $\hat{n}_{i}+\hat{n}%
_{j}=-\hat{n}_{k}$ for $i,j,k$ distinct),%
\begin{align}
|\vec{m}_{+++}|  & =|\vec{m}_{---}|=0,\\
|\vec{m}_{++-}|  & =|\vec{m}_{--+}|=|\vec{m}_{+-+}|\notag\\
& =|\vec{m}_{-+-}|=|\vec{m}_{-++}|=|\vec{m}_{+--}|=2.
\end{align}
It is then straightforward to verify in each case that the necessary and
sufficient conditions on $\eta$, Eqs. (\ref{eq:etanecessary}) and
(\ref{eq:etasufficient}), coincide and yield the bounds given in Eqs.
(\ref{eq:eta1})-(\ref{eq:eta4}). We have already shown that the bound of the
sufficient condition is saturated by the POVM of Eq. (\ref{eq:planarspinaxes}%
).
\end{proof}

\thispagestyle{empty}
\newpage
\phantom{-}
\setcounter{page}{35}
\addcontentsline{toc}{section}{{\bf \red Erratum}}

\end{document}

\title{Erratum to ``Specker's Parable of the Over-protective Seer: A Road to Contextuality, Nonlocality and Complementarity'' [Physics Reports 506  (2011) 1-39]}

\author{Yeong-Cherng Liang}
\email{ycliang@mail.ncku.edu.tw}
\address{Department of Physics, National Cheng Kung University, Tainan 701, Taiwan}

\author{Robert W. Spekkens}
\email{rspekkens@perimeterinstitute.ca (Corresponding author)}
\address{Perimeter Institute for Theoretical Physics, 31 Caroline St. N, Waterloo, \\
Ontario, Canada N2L 2Y5}
\author{Howard M. Wiseman}
\email{h.wiseman@griffith.edu.au}
\address{ARC Centre for Quantum Computation and Communication Technology, and \\
Centre for Quantum Dynamics, Griffith University, Brisbane 4111, Australia}
\maketitle

In our article~\cite{LSW}, we claim that the inequality in our Eq.~(F.3), namely,
\begin{equation}
\eta\leq\frac{1}{N}\frac{\sum_{X_{1}\dots X_{N}}|\vec{m}_{X_{1}\dots X_{N}%
}|^{2}}{\sum_{X_{1}\dots X_{N}}|\vec{m}_{X_{1}\dots X_{N}}|},
\label{eq:etanecessary}%
\end{equation}
where $\vec{m}_{X_{1}\dots X_{N}}\equiv\sum_{k=1}^{N}X_{k}\hat{n}_{k}$,
 is a necessary condition for joint measurability of noisy spin observables along axes $\{ \hat{n}_k : k\in\{1,\dots,N\}\}$, but this claim is incorrect.   
Our mistake was in the paragraph between Eq.~(F.11) and Eq.~(F.12).  The vector described by Eq.~(F.12) does not, in fact, maximize the scalar product we consider because we optimized over vectors having fixed 1-norm rather than fixed 2-norm.
The bound of Eq.~(F.3) {\em is} a sufficient condition for joint measurability because it coincides with the sufficient condition we proved in our article, namely the inequality of Eq.~(F.4),
\begin{equation}
\eta\leq\frac{2^{N}}{\sum_{X_{1}\dots X_{N}}|\vec{m}_{X_{1}\dots X_{N}}
|},
\label{eq:etasufficient}
\end{equation}
The equivalence of the two bounds can be established by noting that $\sum_{X_1 \dots X_N} |\vec m_{X_1\dots X_N}|^2 = N2^N$, which is derived with some simple algebra.

For the specific examples of joint measurability that we considered in Sec.~7 of our article, the bound of (F.3) {\em is}, in fact, necessary and sufficient, so this error does not affect our analysis of those cases.  The necessity of $\eta \le 1\sqrt{2}$ (respectively $\eta \le 1\sqrt{3}$) for the joint measurability of noisy spin observables along a pair (respectively triple) of orthogonal axes (our proposition 8) is, as we noted there, proven in Ref.~\cite{HRS08}.  The necessity of $\eta \le 2/3$ for the joint measurability of noisy spin observables along a triple of axes forming a trine (our proposition 10) was subsequently established by Kunjwal and Ghosh~\cite{KunjwalGhosh} (Appendix B), which can be seen as a special case of the necessary condition that they provided in Eq. (B3).

In general, however, explicit counterexamples to the necessity of the bound can be found.  A trivial example~\footnote{Thanks to Ravi Kunjwal for devising this example.} arises for spin axes $\hat{n}_1= \hat{z}$, $\hat{n}_2=\hat{n}_3 = -\hat{z}$, for which the triple of spin observables can be jointly measured even if they are noiseless, so that the correct bound in this case is $\eta \le 1$, while (F.3) implies $\eta \le 2/3$.  Nontrivial examples can also be found\footnote{T. Heinosaari and R. Uola, private communication; J. Bavaresco, private communication.}.

Proposition 11 of our article is also incorrect.  It claims to identify the highest degree of anti-correlation that can be achieved in a joint measurement of three noisy spin observables oriented along trine axes with purity parameter $\eta = 2/3$.  But higher degrees of anti-correlation can in fact be found,  
implying that our generalized noncontextuality inequality, Eq.~(113), can in fact be violated for such measurements,   
as noted in Refs.~\cite{KunjwalGhosh} and \cite{YuOh}.  Our mistake here was a conceptual one.  We identified a POVM that implemented a joint measurement of our three noisy spin observables and assumed, incorrectly, that 
it was the {\em only} POVM (up to fine-graining) that could do so.  Such an inference is unproblematic in the case of projective measurements where if two measurements are jointly implementable then there is a unique measurement (up to fine-graining) that achieves this, namely, the one associated to the projectors onto their joint eigenspaces.
In the case of POVMs, however, this is not the case, as is discussed in detail in Ref.~\cite{Kunjwal}.

\section{Acknowledgement}
Our thanks to Teiko Heinosaari, Ravi Kunjwal and Jessica Bavaresco, who brought the first of these errors to our attention.